\documentclass[conference]{IEEEtran}
\IEEEoverridecommandlockouts
\usepackage{cite}
\usepackage{amsmath,amssymb,amsfonts,amsthm}
\usepackage{algorithmic}
\usepackage{graphicx}
\usepackage{textcomp}
\usepackage{xcolor}
\usepackage{enumitem}
\usepackage{algorithm}
\usepackage{booktabs}
\usepackage{pifont}
\usepackage{url}
\usepackage{subcaption}
\usepackage{caption}
\usepackage{multirow}
\usepackage{adjustbox}
\usepackage{makecell}   
\usepackage{marvosym}
\usepackage[capitalize,noabbrev]{cleveref}

\newtheorem{theorem}{Theorem}

\newtheorem{definition}[theorem]{Definition}

\def\BibTeX{{\rm B\kern-.05em{\sc i\kern-.025em b}\kern-.08em
    T\kern-.1667em\lower.7ex\hbox{E}\kern-.125emX}}
\begin{document}



\title{
Piecewise Linear Approximation in Learned Index Structures: Theoretical and Empirical Analysis
\thanks{\Letter~Correspondence to Dr.~Qiyu Liu. 
}
}

\author{
\IEEEauthorblockN{Jiayong Qin\IEEEauthorrefmark{1}, Xianyu Zhu\IEEEauthorrefmark{2}, Qiyu Liu\textsuperscript{\Letter}\IEEEauthorrefmark{1}, Guangyi Zhang\IEEEauthorrefmark{3}, Zhigang Cai\IEEEauthorrefmark{1}, Jianwei Liao\IEEEauthorrefmark{1}, \\
Sha Hu\IEEEauthorrefmark{1}, Jingshu Peng\IEEEauthorrefmark{4}, Yingxia Shao\IEEEauthorrefmark{5}, Lei Chen\IEEEauthorrefmark{6}}

\IEEEauthorblockA{
\IEEEauthorrefmark{1}Southwest University, 
\IEEEauthorrefmark{2}RUC, 
\IEEEauthorrefmark{3}SZTU, 
\IEEEauthorrefmark{4}ByteDanace, 
\IEEEauthorrefmark{5}BUPT,
\IEEEauthorrefmark{6}HKUST (GZ)\\
\{jiayongqin1, xianyuzhuruc, qyliu.cs, guangyizhang.jan, liaotoad\}@gmail.com, 
\{czg, husha\}@swu.edu.cn,\\
jingshu.peng@bytedance.com, 
shaoyx@bupt.edu.cn, 
leichen@cse.ust.hk
}

}

\maketitle

\begin{abstract}
A growing trend in the database and system communities is to augment conventional index structures, such as B+-trees, with machine learning (ML) models. 
Among these, error-bounded Piecewise Linear Approximation ($\epsilon$-PLA) has emerged as a popular choice due to its simplicity and effectiveness. 
Despite its central role in many learned indexes, the design and analysis of $\epsilon$-PLA fitting algorithms remain underexplored. 
In this paper, we revisit $\epsilon$-PLA from both \emph{theoretical} and \emph{empirical} perspectives, with a focus on its application in learned index structures. 
We first establish a fundamentally improved \emph{lower bound} of $\Omega(\kappa \cdot \epsilon^2)$ on the expected segment coverage for existing $\epsilon$-PLA fitting algorithms, where $\kappa$ is a data-dependent constant. 
We then present a comprehensive benchmark of state-of-the-art $\epsilon$-PLA algorithms when used in different learned data structures. 
Our results highlight key trade-offs among model accuracy, model size, and query performance, providing actionable guidelines for the principled design of future learned data structures.  
\end{abstract}


\section{Introduction}\label{sec:introduction}

An emerging line of recent research explores augmenting conventional data structures, such as B$^+$-tree and Hash table, with simple machine learning (ML) models, leading to the concept of \emph{learned index}~\cite{kraska2018case,mitzenmacher2018model,qi2020effectively}. 
By exploiting data distribution characteristics, these structures optimize their internal layout, achieving superior performance with both space and time efficiency compared to their conventional counterparts, as demonstrated in recent benchmark studies~\cite{marcus2020benchmarking,wongkham2022updatable}. 

Rather than relying on sophisticated deep learning (DL) models, learned indexes typically favor simple ML approaches, such as linear models~\cite{ferragina2020pgm}. 
This is because DL models often depend on heavy runtimes like Pytorch and Tensorflow, which are costly and less flexible for systems with stringent space and time requirements. 
Moreover, unlike many modern ML tasks~\cite{lecun2015deep}, real-world datasets in learned data structure design are often \emph{easy to learn}~\cite{wongkham2022updatable}, making even simple models sufficiently accurate. 
Among various ML models, the error-bounded piecewise linear approximation model ($\epsilon$-PLA) is a popular choice due to its superior trade-off between its fitting power and inference efficiency~\cite{ferragina2020pgm,ferragina2020learned,galakatos2019fiting,ChenLLD0LZ23,dai2020wisckey,liu2021lhist,liu2022hap,liu2025bittuner,liu2024ldc}. 

Given an array of \emph{sorted} keys $\mathcal{K}=\{k_1,\cdots,k_n\}$, the conventional indexing problem~\cite{bayer1970organization} is to construct a mapping from keys to their corresponding sorting indexes $\mathcal{I}=\{1,\cdots,n\}$, defined as $f: \mathcal{K}\mapsto\mathcal{I}$. 
Intuitively, $f$ can be regarded as the cumulative distribution function (CDF) of $\mathcal{K}$, scaled by a factor of data size $n$. 
As illustrated in~\cref{fig:pla_eg}, existing fitting algorithms, such as SwingFilter~\cite{ORourke81} and GreedyPLA~\cite{XiePZZD14}, can find an $\epsilon$-PLA that approximates $f$ with a maximum error bounded by $\epsilon$. 
Querying an arbitrary $q$ involves two steps: 
\ding{182}~Identify the segment $\ell$ of the trained $\epsilon$-PLA that covers $q$ and compute the predicted index $\ell(q)$; 
\ding{183}~Perform an exact ``last-mile'' search within the range $[\ell(q)-\epsilon, \ell(q)+\epsilon]$ to locate $q$ in the sorted array $\mathcal{K}$. 
Existing learned indexes employ different strategies to efficiently locate the segment $\ell$ that covers a given query key $q$. 
For example, FITing-Tree~\cite{galakatos2019fiting} adopts a conventional B$^+$-tree to index the segments, while PGM-Index~\cite{ferragina2020pgm} recursively constructs $\epsilon$-PLA until a pre-defined segment size threshold is reached. 



\begin{figure}[t]
     \centering
     \begin{subfigure}[b]{0.48\textwidth}
         \centering
         \includegraphics[width=\textwidth]{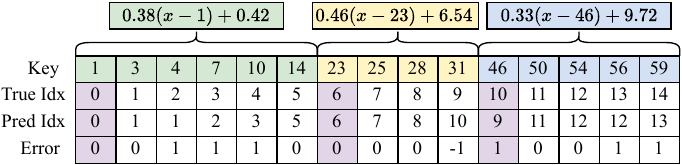}
         \caption{Illustration of a toy $\epsilon$-PLA model trained on 15 keys ($\epsilon=1$). To search for a query key $q=28$, the segment covering $q$ is first found to compute the predicted index for $q$, which is $\lfloor 0.46\cdot(28-23)+6.54\rfloor=8$. 
         As the provided $\epsilon$-PLA model is error-bounded, the true index for $q=28$ lies within the range $[7, 9]$.}
         \label{fig:pla_eg}
     \end{subfigure}
     \begin{subfigure}[b]{0.48\textwidth}
         \centering
         \includegraphics[width=\textwidth]{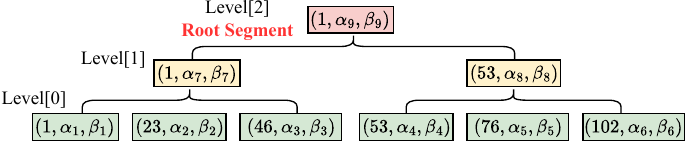}
         \caption{Illustration of a PGM-Index~\cite{ferragina2020pgm} which recursively constructs $\epsilon$-PLA until there is a single model. Note that $(s, \alpha, \beta)$ encodes a linear segment $f(x)=\alpha\cdot(x-s)+\beta$.}
         \label{fig:pgm_eg}
     \end{subfigure}
        \caption{Illustration of the $\epsilon$-PLA model and learned index structure based on $\epsilon$-PLA.}
        \vspace{-3ex}
\end{figure}

\begin{table*}[t]
    \centering
    \caption{Major results and core assumptions of existing theoretical results on the segment coverage $C(\epsilon)$. In summary, we establish a fundamentally improved lower bound under relaxed assumptions with broader applicability. }
    \label{tab:summary_theoretical_results}
    \begin{tabular}{cccc}
    \toprule
        \textbf{Existing Work} & \textbf{Major Result} & \textbf{Core Assumption(s)} & \textbf{Applicable PLA Algorithm(s)} \\\midrule
       VLDB 2020~\cite{ferragina2020pgm} & $C(\epsilon)\geq 2\epsilon$ & N.A. & Optimal $\epsilon$-PLA \\
       ICML 2020~\cite{ferragina2020learned} & $C(\epsilon)=\mu^2\cdot\epsilon^2/\sigma^2$ & i.i.d.~assumptions on \textbf{\emph{gaps}} and $\epsilon\gg\sigma/\mu$ & Optimal $\epsilon$-PLA \\
       \textbf{Ours} & $C(\epsilon)\geq\kappa\epsilon^2$ & i.i.d.~assumptions on \textbf{\emph{keys}} & Both optimal and greedy $\epsilon$-PLA\\
       \bottomrule
    \end{tabular}
    \vspace{-3ex}
\end{table*}

Efficient and effective $\epsilon$-PLA fitting algorithms are essential for building a wide range of learned indexes~\cite{ferragina2020pgm,ferragina2020learned,galakatos2019fiting,ChenLLD0LZ23,dai2020wisckey}. 
The study of $\epsilon$-PLA models dates back to the 1960s~\cite{cameron1966piece}, and has since been extensively explored, particularly in the context of time series data compression~\cite{chiarot2023time}. 
Typical $\epsilon$-PLA fitting algorithms can be classified into two types: 
\ding{182} \textbf{Optimal $\epsilon$-PLA}, which aims to minimize the number of segments under the given error bound $\epsilon$. 
Typical works include ParaOptimal~\cite{ORourke81} and OptimalPLR~\cite{XiePZZD14}. 
\ding{183} \textbf{Greedy $\epsilon$-PLA}, which prioritizes construction speed by extending segments incrementally until the error bound is violated, potentially at the cost of using more segments. 
Typical algorithms are GreedyPLR~\cite{XiePZZD14} and SwingFilter~\cite{elmeleegy2009online}.

Although PLA algorithms have been extensively studied and applied for decades, their performance within the context of learned indexes remains largely unexplored. 
For example, while the PGM-Index~\cite{ferragina2020pgm} employs the optimal SwingFilter algorithm~\cite{ORourke81} to construct its $\epsilon$-PLA segments, it does not investigate how alternative PLA algorithms may impact index performance. 
Motivated by this, in this paper, we comprehensively revisit $\epsilon$-PLA fitting algorithms for learned indexes from both \textbf{theoretical} and \textbf{empirical} perspectives.

From the theoretical aspect, a key research question for an arbitrary $\epsilon$-PLA fitting algorithm is: 
\textbf{\emph{How many line segments (or inversely, points covered by each segment) are required to satisfy the error bound $\epsilon$ for a given dataset?}} 
This directly impacts the performance of PLA-based learned index as more segments (or fewer points covered by each segment) improve accuracy but increase storage overhead, and vice versa.  
The PGM-Index~\cite{ferragina2020pgm} introduces the first lower bound for segment coverage as $2\epsilon$, based on the piecewise constant function (a special case of PLA). 
Ferragina et~al.~\cite{ferragina2020learned} further derives the expected segment coverage as a quadratic function $\mu^2\epsilon^2/\sigma^2$, where $\mu$ and $\sigma^2$ are the mean and variance of the \emph{gaps} between the consecutive keys, respectively. 
A key assumption, used in~\cite{ferragina2020learned} and inherited by subsequent studies~\cite{boffa2022learned,ChenLLD0LZ23}, is that the $i$-th sorted key $k_i$ is a realization of a random process $K_i=K_{i-1} + G_i$, where $G_i$'s, referred to as \emph{gaps}, are \textbf{\emph{i.i.d.}}~non-negative random variables. 

In this work, we present a simple yet nontrivial analytical framework to re-examine the theoretical foundations of $\epsilon$-PLA. 
Rather than impractical \emph{gap}-based assumptions in previous works~\cite{ferragina2020learned,ChenLLD0LZ23,boffa2022learned} (as detailed in \cref{subsec:existing_results}), we directly assume keys are i.i.d.~samples drawn from an \emph{arbitrary} distribution. 
Under relaxed assumptions, we derive a tighter \emph{lower bound} of the expected coverage per segment in an $\epsilon$-PLA, given by $\Omega(\kappa\epsilon^2)$, where $\kappa$ is a data-dependent constant. 
As summarized in \cref{tab:summary_theoretical_results}, to the best of our knowledge, our lower bound on $\epsilon$-PLA is the tightest result for both optimal and greedy $\epsilon$-PLA construction algorithms.

The established lower bound highlights the worst-case behavior of generic $\epsilon$-PLA fitting algorithms. 
Yet, it remains unclear how different $\epsilon$-PLA algorithms impact the practical performance of learned indexes. 
This raises our second research question:
\textbf{\emph{How do different $\epsilon$-PLA fitting algorithms influence construction time, index size, and query efficiency?}}
To answer this, we introduce \textbf{PLABench}\footnote{Our benchmark is made publicly available at \url{https://github.com/bdhxxnix/PLABench}.}, the first comprehensive benchmark designed to systematically evaluate the role of $\epsilon$-PLA fitting in learned indexes. 
PLABench provides unified, optimized implementations of a wide range of representative $\epsilon$-PLA algorithms, covering both \textbf{optimal and greedy} strategies, as well as \textbf{single- and multi-threading} variants, and integrates them into two canonical learned index structures: FITing-Tree~\cite{galakatos2019fiting} and PGM-Index~\cite{ferragina2020pgm}. 
PLABench enables an in-depth analysis of the space-time trade-offs induced by PLA choices across diverse datasets and query workloads. 
Moreover, PLABench can be easily extended to support future PLA-based learned index designs, benefiting model selection and hyperparameter tuning.

To the best of our knowledge, this is the first work to revisit $\epsilon$-PLA fitting algorithms from both theoretical and empirical perspectives in the context of learned index structures. 
Our technical contributions are summarized as follows:
\begin{enumerate}[leftmargin=*,label={\bfseries C\arabic*:}]
    \item We present a simple yet effective analytical framework for modeling the segment coverage of $\epsilon$-PLA, avoiding impractical assumptions commonly made in prior work. 
    \item Based on \textbf{C1}, we derive the \emph{tightest} known lower bound on the expected segment coverage, $\Omega(\kappa \cdot \epsilon^2)$, for any sorted key set sampled from unknown distributions. 
    \item We introduce PLABench, the first comprehensive benchmark for evaluating the impact of different $\epsilon$-PLA algorithms on learned index performance, offering practical guidance for the design of future PLA-based learned index structures. 
\end{enumerate}

The remainder of this paper is structured as follows. 
\cref{sec:preliminaries} overviews the basis of $\epsilon$-PLA and existing fitting algorithms. 
\cref{sec:theory} theoretically analyze the expected segment coverage of existing $\epsilon$-PLA algorithms 
\cref{sec:benchmark_setup} introduces the benchmark setup, and \cref{sec:exp_results} presents and discusses the benchmark results. 
\cref{sec:related_work} surveys related works, and finally, \cref{sec:conclusion} concludes the paper. 
\section{Preliminaries}\label{sec:preliminaries}
In this section, we provide an overview of $\epsilon$-PLA fitting algorithms and discuss existing theoretical results. 
\cref{tab:notations} summarizes the major notations used hereafter. 

\begin{table}[t]
    \centering
    \caption{Summary of major notations.}
    \label{tab:notations}
    \begin{tabular}{c|c}
    \toprule
        \textbf{Notation} & \textbf{Explanation} \\\midrule
        $\mathcal{K}, \mathcal{I}$ & the input sorted key set and the corresponding index set \\
        $K_{(i)}$ & the $i$-th order statistics of $n$ random keys\\
        $\epsilon$ & the error bound of an $\epsilon$-PLA model \\
        $(s_i, \alpha_i, \beta_i)$ & the starting point, slope, and intercept of the $i$-th segment\\
        $\ell(\cdot)$ & a specific linear segment of an $\epsilon$-PLA model \\
        $C(\epsilon)$ & the expected segment coverage given $\epsilon$\\
        $\rho, \gamma, \xi$ & data-dependent constants\\ 
    \bottomrule
    \end{tabular}
\end{table}

\subsection{Error-bounded Piecewise Linear Approximation}
In general, the error-bounded piecewise linear approximation model ($\epsilon$-PLA) is defined as follows.

\begin{definition}[$\epsilon$-PLA]\label{def:pla}
     Given two monotonically increasing lists $\mathcal{X}=\{x_1,\cdots, x_n\}$ and $\mathcal{Y}=\{y_1,\cdots,y_n\}$, 
     an $\epsilon$-PLA of $m$ line segments on $(\mathcal{X}, \mathcal{Y})=\{(x_i,y_i)\}_{i=1,\cdots,n}$ is defined as,
     \begin{equation}\label{eq:pla}
         f(x)=\begin{cases}
             \alpha_1\cdot (x-s_1)+\beta_1&\text{ if } x\in[s_1,s_2)\\
             \alpha_2\cdot (x-s_2)+\beta_2&\text{ if } x\in[s_2,s_3)\\
             \quad\cdots&\qquad\cdots\\
             \alpha_{m}\cdot (x-s_m)+\beta_{m}&\text{ if } x\in[s_m,+\infty)\\
         \end{cases}
     \end{equation}
     such that $|f(x_i)-y_i|\leq\epsilon$ for $\forall i=1,2,\cdots,n$. 
\end{definition}

When used for data indexing, $\mathcal{X}$ is set to a collection of $n$ sorted keys $\{k_1,\cdots,k_n\}$ and $\mathcal{Y}$ is set to the corresponding index set $\{1,\cdots,n\}$. 
Then, a fitted $\epsilon$-PLA model functions as a conventional index structure like B$^+$-tree. 

\subsection{$\epsilon$-PLA Fitting Algorithm}

The first optimal $\epsilon$-PLA fitting algorithm, ParaOptimal, was proposed by O'Rourke in the 1980s~\cite{ORourke81}, aiming to minimize the number of line segments while satisfying a given error bound $\epsilon$.
Subsequent variants, such as SlideFilter~\cite{elmeleegy2009online} and OptimalPLR~\cite{XiePZZD14}, are theoretically \textbf{equivalent} to ParaOptimal in terms of their segmentation output. 
These algorithms typically operate by updating convex hulls in the $\mathcal{X}\times\mathcal{Y}$ space when processing input points in an online manner (i.e., $(x_1,y_1), (x_2,y_2), \cdots$).
All aforementioned optimal methods achieve $O(n)$ time complexity with $O(n)$ auxiliary space during a single pass over the data. 

In addition to optimal algorithms, greedy heuristics have also been explored to further reduce the space overhead from $O(n)$ to $O(1)$. 
Typical greedy algorithms include SwingFilter~\cite{elmeleegy2009online} and GreedyPLR~\cite{XiePZZD14}. 
SwingFilter fits a segment using the first point in the stream as a pivot, then incrementally adjusts the upper and lower slope bounds to fit subsequent points within the error constraint. 
GreedyPLR follows a similar strategy but selects the pivot as the intersection of two maintained extreme lines, offering a more balanced slope initialization. 
\cref{tab:pla_algorithms} summarizes existing $\epsilon$-PLA fitting algorithms.


\subsection{Limitations of Existing Theoretical Results}\label{subsec:existing_results}
Although $\epsilon$-PLA fitting algorithms have been explored for decades in various domains, understanding why these algorithms perform well and how good their solutions truly are remains an underexplored area. 
In particular, we focus on the \textbf{segment coverage}, the average number of keys covered by each segment in a fitted $\epsilon$-PLA model.
As discussed in \cref{sec:introduction}, the PGM-Index~\cite{ferragina2020pgm} establishes a lower bound of $2\epsilon$ on segment coverage based on a piecewise constant model.  
Furthermore, by assuming that the $i$-th sorted key $k_i$ is sampled from a random process $K_i=K_{i-1} +G_i$, where $G_i$'s are \textbf{\emph{i.i.d.}} non-negative random \textbf{\emph{gaps}}, another recent work~\cite{ferragina2020learned} derives the segment coverage (\textbf{not} lower bound) as $\mu^2\epsilon^2/\sigma^2$, where $\mu$ and $\sigma^2$ are the mean and variance of $G_i$. 

\begin{table}[t]
    \centering
    \caption{Summary of existing $\epsilon$-PLA fitting algorithms.}
    \label{tab:pla_algorithms}
    \begin{tabular}{c|c|c|c}
    \toprule
        \multirow{2}{*}{\textbf{Algorithm}} & \multirow{2}{*}{\textbf{\makecell{Time\\Complexity}}} & \multirow{2}{*}{\textbf{\makecell{Space\\Complexity}}} & \multirow{2}{*}{\textbf{\makecell{Is\\Optimal?}}} \\
        & & & \\
    \midrule
        ParaOptimal~\cite{ORourke81} & $O(n)$ & $O(n)$ & Yes \\
        SlideFilter~\cite{elmeleegy2009online} & $O(n)$ & $O(n)$ & Yes \\
        OptimalPLA~\cite{XiePZZD14} & $O(n)$ & $O(n)$ & Yes \\
        SwingFilter~\cite{elmeleegy2009online} & $O(n)$ & $O(1)$ & No \\
        GreedyPLA~\cite{XiePZZD14} & $O(n)$ & $O(1)$ & No \\
    \bottomrule
    \end{tabular}
\end{table}


The above statistical assumptions are overly idealized in practice. 
First, the unbounded assumption on $K_i$ conflicts with real-world systems, where keys are typically fixed-width integers (e.g., {uint\_32} or {uint\_64}). 
Second, the ``i.i.d.'' assumption on gaps $G_i$ does not hold when drawing $n$ random keys from an arbitrary distribution and then sorting them. 
This is because consecutive order statistics $K_{(i-1)}$, $K_{(i)}$, and $K_{(i+1)}$ are correlated, resulting in dependencies between adjacent gaps $G_i=K_{(i)}-K_{(i-1)}$ and $G_{i+1}=K_{(i+1)}-K_{(i)}$. 
Additionally, the results in~\cite{ferragina2020learned} are valid only under the condition $\epsilon\gg\sigma/\mu$, which does not hold when $\epsilon$ is small. 

These limitations reveal a critical theoretical gap, which undermines the practical applicability of this fundamental building block of learned indexes and highlights the need for a more general and realistic theoretical study on $\epsilon$-PLA. 
\section{Lower Bound of Segment Coverage}\label{sec:theory}
In this section, we formally prove the segment coverage of \textbf{\emph{both}} optimal and greedy $\epsilon$-PLA algorithms, avoiding impractical assumptions on gaps. 
The roadmap is given as follows:

\noindent
\ding{182} \cref{subsec:coverage_def} formulates the segment coverage and introduces our improved assumption on input keys.

\noindent
\ding{183} Based on a constructive algorithm, \cref{subsec:uniform} and \cref{subsec:arbitrary} analyze the lower bound of segment coverage for uniformly and arbitrarily distributed keys. 

\noindent
\ding{184} Finally, \cref{subsec:extension} extends the results to existing $\epsilon$-PLA fitting algorithms.

\subsection{Segment Coverage and Key Assumptions}\label{subsec:coverage_def}
Let $\mathcal{K}=\{k_1,\cdots,k_n\}$ denote a set of $n$ sorted integer keys and $\mathcal{I}=\{1,\cdots,n\}$ represent the set of corresponding indexes. 
Rather than relying on previous impractical settings, we directly model the input sorted keys $\mathcal{K}=\{k_1,\cdots,k_n\}$ as the realization of \textbf{order statistics} $K_{(1)}, \cdots, K_{(n)}$ of $n$ i.i.d.~random samples $K_1,\cdots,K_n$, drawn from an \textbf{arbitrary} distribution with a continuous cumulative function (CDF). 

We assume that all $K_i$'s are bounded within the range $[0, \rho]$, where $\rho=\Theta(n)$. 
This assumption is reasonably made to ensure the slope of a line segment remains bounded. 
For example, when indexing $n$ keys drawn from a standard uniform distribution $\mathcal{U}(0, \rho)$, as shown later in \cref{theorem:index_coverage_uniform}, the slope of the segments in $\epsilon$-PLA should be $\alpha=\frac{n+1}{\rho}$. 
Note that, in real-world systems, such as indexes in database management system~\cite{marcus2020benchmarking} and inverted lists in information retrieval systems~\cite{pibiri2020techniques}, keys are typically unsigned integers stored as {uint\_32} (within $[0, 2^{32}-1]$) or {uint\_64} (within $[0, 2^{64}-1]$), whose upper bounds can be safely regarded as large as the data size $n$. 

\begin{definition}[Segment Coverage]\label{def:coverage}
    Given a dataset $(\mathcal{X}, \mathcal{Y})=\{(x_i,y_i)\}_{i=1,\cdots,n}$ and an error constraint $\epsilon$ ($\epsilon\geq1$), let $f=\mathcal{A}(\mathcal{X},\mathcal{Y})$ denote the $\epsilon$-PLA learned by a fitting algorithm $\mathcal{A}$, and let $\ell(x)=\alpha\cdot x+\beta$ represent a segment in $f$. 
    The segment coverage, denoted by $C(\epsilon)$, is defined as:  
    \begin{equation}\label{eq:coverage_def}
        C(\epsilon) = \mathbf{E}[I] = \sum\nolimits_{i\in\mathbb{I}} i\cdot\Pr\left(E_1\wedge E_2\cdots \wedge E_i\right),
    \end{equation}
    where $E_i$ denotes the event of $|\ell(x_i) - y_i|\leq\epsilon$, r.v.~$I$ represents the number of consecutive occurrences of $E_i$, and $\mathbb{I}=\{1,2,\cdots,|\mathbb{I}|\}$ is the set of all possible values of $i$. 
\end{definition}

W.l.o.g., we always consider $\ell$ to be the first segment of $f$. 
We further assume that $|\mathbb{I}|=o(n)$, i.e., $\lim_{n\to\infty}|\mathbb{I}|=\infty$ and $\lim_{n\to\infty}{|\mathbb{I}|}/{n}=0$, implying that the maximum segment coverage is of lower order than the data size $n$. 
Note that, an alternative to the probability term in \cref{eq:coverage_def} is $\Pr(E_1\wedge\cdots\wedge E_i\wedge \overline{E_{i+1}})$. 
Both formulations yield equivalent asymptotic behavior in subsequent theoretical analysis. 
For simplicity, we adopt the formulation given in \cref{eq:coverage_def} throughout the rest of this paper.


\textbf{Remarks.} 
In~\cite{ferragina2020learned}, segment coverage is modeled as the mean exit time (MET) problem of a random process~\cite{gardiner1985handbook,redner2001guide}, based on i.i.d.~assumptions on \textbf{gaps}. 
In contrast, here, we adopt a fundamentally different statistical model. 
To provide an intuitive interpretation, segment coverage can be related to the count of consecutive ``heads'' in a coin-flipping game, where each toss has varying probabilities and correlations.

\begin{algorithm}[t]
   \caption{Fixed Range Segmentation (FRS)}\label{alg:segment}
\begin{algorithmic}
   \STATE {\bfseries Input:} $n$ sorted keys $\{k_1,\cdots,k_n\}$, an error parameter $\epsilon$
   \STATE {\bfseries Output:} an $\epsilon$-PLA
    
    \STATE Initialize an empty segment set $S\gets \{\}$
    \STATE Initialize a segment $\ell(x)=\frac{n+1}{k_n-k_1}\cdot(x-k_1)+1$
   \FOR{$i=1$ {\bfseries to} $n$}
   \IF{$|\lfloor\ell(k_i)\rfloor - i|>\epsilon$}
   \STATE Append current segment $S\gets S\cup \{\ell\}$ 
   \STATE Update the segment $\ell(x)\gets \frac{n+1}{k_n-k_1}\cdot(x-k_i)+i$
   \ENDIF
   \ENDFOR

   {\bfseries return} $S$
\end{algorithmic}
\end{algorithm}

\subsection{Segment Coverage for Uniform Keys}\label{subsec:uniform}
We begin by analyzing the behavior of segment coverage for a simple algorithm constructed just for proof, and then generalize to practical $\epsilon$-PLA fitting algorithms. 
As detailed in \cref{alg:segment}, the Fixed Range Segmentation (FRS) algorithm always fixes the slope of each segment as $\frac{n+1}{\rho}$, where $\rho=k_n-k_1$ is an estimator of the range of keys.

\begin{theorem}[Segment Coverage for Uniform Keys]\label{theorem:index_coverage_uniform}
    Given a set of sorted integers $\mathcal{K}=\{k_1,\cdots,k_n\}$, w.l.o.g., suppose that $k_1,\cdots,k_n$ is a realization of order statistics $K_{(1)},\cdots,K_{(n)}$ of $n$ i.i.d.~random samples $K_1,\cdots,K_n$ drawn from $\mathcal{U}(0, \rho)$. 
    Given an error threshold $\epsilon$, when $n$ is sufficiently large, the expected coverage of a line segment $\ell(x)=\frac{n+1}{\rho}\cdot x$ is given as follows:
    \begin{equation}
    \label{eq:coverage_bound_uniform}
        C(\epsilon) = \Omega(\epsilon^2).
    \end{equation}
\end{theorem}

\begin{figure}[t]
    \centering
    \includegraphics[width=0.38\textwidth]{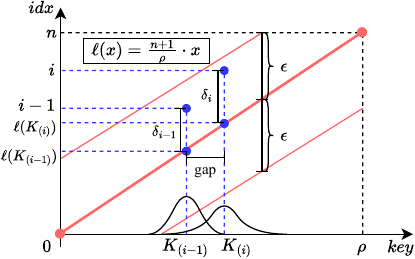}
    \caption{Illustration of segment coverage by constructing the segment as $\ell(x)=\frac{n+1}{\rho}\cdot x$, where $K_{(i-1)}$ and $K_{(i)}$ are two consecutive order statistics, and ``gap'' stands for the difference of $K_{(i)}$ and $K_{(i-1)}$. }
    \label{fig:index_coverage}
    \vspace{-3ex}
\end{figure}

\begin{proof}
    According to~\cite{gentle2009computational}, the $i$-th order statistics $U_{(i)}$ of standard uniform distribution (i.e., $\mathcal{U}(0, 1)$) follows a beta distribution $\mathcal{B}(i, n+1-i)$. 
As $K_{(i)}=\rho\cdot U_{(i)}$, the mean and variance of $K_{(i)}$ can be derived as $\mathbf{E}[K_{(i)}] = \rho\cdot\frac{i}{n+1}$ and $\mathbf{Var}[K_{(i)}] = \rho^2\cdot\frac{i\cdot(n+1-i)}{(n+1)^2\cdot(n+2)}$. 
As illustrated in \cref{fig:index_coverage}, let $\delta_i=\ell(K_{(i)})-i$ denote the \textbf{residual} between the predicted index for $K_{(i)}$ (i.e., $\frac{n+1}{\rho}\cdot K_{(i)}$) and the true index $i$,
we have $\mathbf{E}[\delta_i] = 0$ and $\mathbf{Var}[\delta_i]=\frac{i\cdot(n+1-i)}{n+2}$. 
According to \cref{def:coverage} and \cref{eq:coverage_def}, we have:
\begin{equation}\label{eq:mean_of_first_success}
    C(\epsilon) \geq \sum_{i=1}^{|\mathbb{I}|} i\cdot\prod_{j=1}^{i}\Pr\left(E_j\right), 
\end{equation}
given that events $E_j$ positively correlates\footnote{Intuitively, event $\overline{E_i}$ is more likely to happen when $i$ gets larger, given that the variance of $\delta_i$ is monotonically increasing for $i=1,\cdots,|\mathbb{I}|=o(n)$. 
When $n\to\infty$, $\mathbf{Var}[\delta_i]\to\infty$, leading to $\Pr(E_i)\to 0$. Thus, when $E_i$ occurs, it becomes more likely that $E_{j<i}$ also occurs. A more rigorous proof can be made by evaluating $\Pr(E_{j<i}|E_i)$, which is ignored here for brevity.} with each other, i.e., $\Pr(E_1\wedge\cdots\wedge E_i)\geq\prod_{j=1}^{i}\Pr(E_j)$. 
By applying the Chebyshev's inequality, i.e., $\Pr(\overline{E_j})=\Pr(|\delta_j|>\epsilon)\leq\mathbf{Var}[\delta_j]/\epsilon^2$, it follows that: 
\begin{equation}\label{eq:chebyshev}
    \begin{aligned}
        \prod_{j=1}^{i}\Pr(E_j)&\geq \prod_{j=1}^{i}\left(1-\frac{j\cdot(n-j+1)}{(n+2)\cdot\epsilon^2}\right)\\
        &\approx\exp\left(-\sum_{j=1}^{i}\frac{j}{\epsilon^2}\right)
        \approx\exp\left(-\frac{i^2}{2\epsilon^2}\right).
    \end{aligned}
\end{equation}

Combining \cref{eq:mean_of_first_success} and \cref{eq:chebyshev}, the following lower bound holds:
\begin{equation}\label{eq:temp_sum}
    C(\epsilon) \geq \sum_{i=1}^{|\mathbb{I}|} i\cdot\exp\left(-\frac{i^2}{2\epsilon^2}\right).
\end{equation}

For $i \leq \epsilon$, we have $\frac{i^2}{2\epsilon^2} \leq \frac{1}{2}$, so $\exp(-\frac{i^2}{2\epsilon^2}) \geq e^{-1/2}$. 
As $\epsilon<|\mathbb{I}|$ and $i\cdot\exp({-\frac{i^2}{2\epsilon^2}})$ is non-negative, \cref{eq:temp_sum} can be further lower bounded as follows:
\begin{equation}
    \begin{aligned}
        C(\epsilon) &\geq \sum_{i=1}^{\epsilon} i\cdot \exp\left(-\frac{i^2}{2\epsilon^2}\right) \\
        &\geq e^{-1/2} \sum_{i=1}^{\epsilon} i \\
        &\geq \frac{e^{-1/2}\cdot\epsilon^2}{2} \approx 0.303\cdot\epsilon^2.
    \end{aligned}
\end{equation}

Thus, we complete the proof of \cref{theorem:index_coverage_uniform}.
\end{proof}


\textbf{Remarks.}
The theoretical results in~\cite{ferragina2020learned} rely on an implicit assumption that $\epsilon\gg\sigma/\mu$ to apply central limit theorems, limiting their validity for small $\epsilon$. 
In contrast, our work removes this restriction, only assuming a much weaker condition $n'\gg\epsilon$ and thereby addressing the theoretical gap for small $\epsilon$.

\subsection{Segment Coverage for Arbitrary Keys}\label{subsec:arbitrary}

\begin{theorem}[Segment Coverage for Arbitrary Keys]\label{theorem:index_coverage_general}
    Given a set of sorted keys $\mathcal{K}=\{k_1,\cdots,k_n\}$ and an error threshold $\epsilon$, consider that: 
    \begin{enumerate}[leftmargin=*,label={\bfseries A\arabic*:}]
        \item $\mathcal{K}=\{k_1,\cdots,k_n\}$ is the realization of order statistics $K_{(1)},\cdots,K_{(n)}$ of $n$ i.i.d.~samples $K_1,\cdots,K_i$ drawing from an arbitrary distribution. 
        \item The inverse cumulative function $F^{-1}(\cdot)$ exists and for arbitrary $i\in\{1,\cdots,|\mathbb{I}|\}$, there exists constant $\gamma$ such that $|\frac{n}{\rho}\cdot F^{-1}(\frac{i}{n})- i|\leq \gamma$. 
        \item For arbitrary $i\in\{1,\cdots,|\mathbb{I}|\}$, the density function $f(\cdot)$ is continuous and non-zero at $F^{-1}(\frac{i}{n})$. Moreover, there exists a constant $\xi$ such that $f(F^{-1}(\frac{i}{n}))\geq 1/\xi$. 
    \end{enumerate}
    Under A1, A2, and A3, when $n$ is large enough, the expected coverage of a line segment $\ell(x)=\frac{n}{\rho}\cdot x$ is given by:
    \begin{equation}
    \label{eq:coverage_bound_general}
        C(\epsilon) = \Omega\left(\frac{\sqrt{\rho}\cdot(\epsilon-\gamma)^2}{\xi}\right).
    \end{equation}
\end{theorem}

\begin{proof}
    For general distributions, when $n\to\infty$, according to~\cite{cardone2022entropic}, the $i$-th order statistics $K_{(i)}$ asymptotically follows a normal distribution as follows:
\begin{equation}\label{eq:asmp_normal}
    K_{(i)} \xrightarrow{\mathit{d}} \mathcal{N}\left(F^{-1}\left(\frac{i}{n}\right), \frac{i\cdot(n-i)}{n^2\cdot[f(F^{-1}(\frac{i}{n}))]^2}\right). 
\end{equation}
Thus, the $i$-th residual $\delta_i=\ell(K_{(i)})-i$ also exhibits an asymptotic normal distribution:
\begin{equation}\label{eq:error_normal}
    \delta_i \xrightarrow{\mathit{d}} \mathcal{N}\left(m_i, \frac{i\cdot(n-i)}{\rho^2\cdot[f(F^{-1}(\frac{i}{n}))]^2}\right),
\end{equation}
where $m_i=\frac{n}{\rho}\cdot F^{-1}\left(\frac{i}{n}\right)-i$. 
As $m_i\neq 0$, different from the uniform case, $\ell(K_{i})$ is no longer an unbiased estimator of the true index $i$. 
Under A2 and A3 as introduced in \cref{theorem:index_coverage_general}, by employing a similar technique, the segment coverage with a revised error constraint $\epsilon+\gamma$ is given by:
\begin{equation}
\begin{aligned}
    C(\epsilon+\gamma) &\geq\sum_{i=1}^{|\mathbb{I}|}i\cdot\prod_{j=1}^{i}\Pr(|\delta_i-m_i|\leq\epsilon) \\
    &\geq\sum_{i=1}^{|\mathbb{I}|}i\cdot\prod_{j=1}^{i}\left(1-\frac{i\xi^2}{\rho\epsilon^2}\right)\\
    &\approx\sum_{i=1}^{|\mathbb{I}|}i\cdot\exp\left(-\frac{i^2\xi^2}{2\rho\epsilon^2}\right)\\
    &= \Omega\left(\frac{\sqrt{\rho}\cdot\epsilon^2}{\xi}\right).
\end{aligned}
\end{equation}
Thus, we achieve the results in \cref{theorem:index_coverage_general}. 
\end{proof}

\textbf{Remarks.} 
In both \cref{theorem:index_coverage_uniform} and \cref{theorem:index_coverage_general}, the slope of line segments is fixed to $\frac{n+1}{\rho}$, naturally leading to the simple FRS algorithm described in \cref{alg:segment}. 
An interesting observation is that, despite being based on relaxed assumptions, FRS exhibits an interesting mathematical connection to the MET algorithm introduced in~\cite{ferragina2020learned}.
In MET, the slope is fixed to $1/\mu$ and $\mu$ is the mean of \emph{gaps} between consecutive keys, which is conceptually equivalent to FRS as $\rho=k_n-k_1=\sum_{i=2}^{n}(k_i-k_{i-1})$.

\subsection{Extension to Other $\epsilon$-PLA Algorithms}\label{subsec:extension}

The lower bounds in \cref{theorem:index_coverage_uniform} and \cref{theorem:index_coverage_general} naturally extend to the optimal $\epsilon$-PLA fitting algorithms (e.g., ParaOptimal~\cite{ORourke81}, SlideFilter~\cite{elmeleegy2009online}, and OptimalPLA~\cite{XiePZZD14}), as these methods guarantee the minimal number of segments under a given error bound. 
In this section, we further extend our results to cover greedy $\epsilon$-PLA algorithms, such as SwingFilter~\cite{elmeleegy2009online} and GreedyPLA~\cite{XiePZZD14}. 

\begin{figure}[h]
    \centering
    \begin{subfigure}[t]{0.5\columnwidth}
        \includegraphics[width=\linewidth]{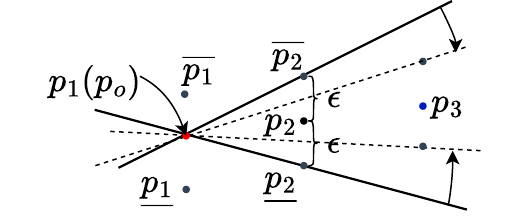}
        \caption{SwingFilter~\cite{elmeleegy2009online}}
    \end{subfigure}
    \begin{subfigure}[t]{0.45\columnwidth}
        \includegraphics[width=\linewidth]{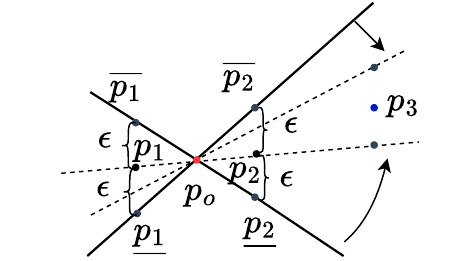}
        \caption{GreedyPLA~\cite{XiePZZD14}}
    \end{subfigure}
    
    \caption{Illustration of two greedy $\epsilon$-PLA fitting algorithms. The point in red ($p_o$) is the selected pivot point.}
    \label{fig:SwingandGreedy}
\end{figure}

As illustrated in \cref{fig:SwingandGreedy}, given an error bound $\epsilon$, SwingFilter picks the first point of each segment as the pivot $p_o$, and initializes the slope range by drawing lines from $p_o$ to the upper and lower bounds of the second point, denoted as $\overline{p_2}$ and $\underline{p_2}$. 
The slope range is then progressively refined by connecting $p_o$ to subsequent extreme points (e.g., $\overline{p_3}$ and $\underline{p_3}$), until violating the error bound $\epsilon$. 
GreedyPLA further improves SwingFilter by picking the midpoint of $p_1$ and $p_2$ as the pivot and updates the slope range similar to SwingFilter.

\begin{theorem}[Superiority of Greedy $\epsilon$-PLA over FRS]\label{theorm:superiority_of_greedy}
    Given a set of sorted keys $\mathcal{K} = \{k_1, \cdots, k_n\}$ and an error bound $\epsilon$, the segment coverage created by SwingFilter and GreedyPLA is always \textbf{larger} than that of FRS (\cref{alg:segment}).  
\end{theorem}

\begin{proof}
We prove the claim by induction. 
For brevity, let $L(i)$ denote the claim: 
\emph{if the first $i$ points are covered by segments in both FRS and SwingFilter, and the $(i+1)$-th point is covered by the FRS segment, then it is also covered by the corresponding SwingFilter segment}.

\noindent
\textbf{Base case:} 
$L(1)$ is trivially true as SwingFilter initializes its segment using the first two points.

\noindent
\textbf{Inductive step:} 
Assuming $L(i)$ holds, we then show that \( L(i+1) \) also holds. 
For SwingFilter, the upper and lower bounds of the segment slope when traversing the $i$-th point are:
\begin{equation}
\begin{aligned}
    \underline{\alpha_S} &= \max_{j=2,\dots,i} \left( \frac{j - \epsilon - 1}{k_j - k_1} \right),\\ 
\overline{\alpha_S} &= \min_{j=2,\dots,i} \left( \frac{j + \epsilon - 1}{k_j - k_1} \right).
\end{aligned}
\label{eq:Swing_bounds}
\end{equation}
For FRS, according to \cref{alg:segment}, a fixed slope of $\alpha_F = \frac{n+1}{k_n - k_1}$ is adopted. 
Thus, we have:
\begin{equation}
    |\alpha_F\cdot (k_j - k_1) + 1 - j| \leq \epsilon, \quad j = 2, \dots, i,
\end{equation}
which implies:
\begin{equation}\label{eq:frs_bounds}
    \frac{j - 1 - \epsilon}{k_j - k_1} \leq \alpha_F \leq \frac{j - 1 + \epsilon}{k_j - k_1}, \quad j = 2, \dots, i.
\end{equation}
Combining \cref{eq:Swing_bounds} with \cref{eq:frs_bounds}, it holds that $\underline{\alpha_S} \leq \alpha_F \leq \overline{\alpha_S}$. 
Then, if the $(i+1)$-th point $(k_{i+1}, i+1)$ can be covered by FRS's segment, i.e., $|\alpha_F\cdot (k_{i+1} - k_1) + 1 - (i+1)| \leq \epsilon$, 
it always holds that $\underline{\alpha_S}\cdot (k_{i+1} - k_1) + 1 - \epsilon \leq i+1 \leq \overline{\alpha_S}\cdot (k_{i+1} - k_1) + 1 + \epsilon$, 
implying that $(k_{i+1}, i+1)$ is also covered by SwingFilter's segment. 
Therefore, $L(i+1)$ holds, and by induction, $L(i)$ holds for all $i$. 

Thus, we formally show that any points covered by a segment in FRS must also be covered by the corresponding segment in SwingFilter, indicating an intrinsically higher coverage. 
Notably, the same results hold for GreedyPLA following a similar approach, which we omit here due to page limits.
\end{proof}

\section{Benchmark Setup}\label{sec:benchmark_setup}

The lower bounds established in \cref{sec:theory} fill the theoretical gap in understanding the asymptotic behavior of various $\epsilon$-PLA algorithms under general data distribution assumptions. 
However, existing learned indexes such as PGM-Index~\cite{ferragina2020pgm} typically adopt the OptimalPLA~\cite{ORourke81} algorithm for index construction, leaving the practical impact of alternative PLA fitting algorithms largely unexplored. 
To address this, we propose \textbf{PLABench}, a comprehensive and flexible $\epsilon$-PLA benchmark tailored for learned index scenarios. 

\subsection{Research Questions}
In particular, we aim to answer the following research questions through our PLABench. 
\begin{enumerate}[leftmargin=*,label={\bfseries RQ\arabic*:}]
    \item Do the lower bounds established in \cref{sec:theory} empirically hold for existing $\epsilon$-PLA fitting algorithms? 
    \item How good are existing $\epsilon$-PLA fitting algorithms in practical multi-threading environments? 
    \item When integrated within learned index structures, what are the trade-offs made by greedy $\epsilon$-PLA algorithms compared to the more widely adopted optimal $\epsilon$-PLA algorithms? 
    \item How does the error bound $\epsilon$ affect the space-time trade-offs for each approach?
\end{enumerate}

\subsection{Baselines and Implementation Details}\label{subsec:impl}
We provide optimized implementations for three representative $\epsilon$-PLA fitting algorithms, OptimalPLA~\cite{XiePZZD14}, GreedyPLA~\cite{XiePZZD14}, and SwingFilter~\cite{elmeleegy2009online}, covering both optimal and greedy approaches. 
Other optimal algorithms, such as ParaOptimal~\cite{ORourke81} and SlideFilter~\cite{elmeleegy2009online}, are excluded since they are theoretically equivalent to OptimalPLA and produce identical $\epsilon$-PLA models given the same input dataset~\cite{XiePZZD14}.

Given that parallel index construction is widely adopted on modern computing platforms~\cite{akhremtsev2016fast}, we evaluate all algorithms under both single-threading and multi-threading settings. 
For the parallel implementation, we follow the strategy used in the PGM-Index implementation~\cite{pgm}, where the input key set $\mathcal{K}$ is divided into consecutive, disjoint chunks, and the $\epsilon$-PLA algorithm is applied independently to each chunk. 
Apparently, such an embarrassingly parallel strategy will break the optimality of the OptimalPLA algorithm, which is rarely discussed in existing work.
Interestingly, our theoretical analysis and experimental results consistently demonstrate that the number of additional segments introduced is upper bounded by the number of threads. 

 \begin{table}[t]
    \centering
    \caption{Summary of benchmark datasets.}
    \label{tab:datasets_stats}
    \begin{tabular}{lccc}
    \toprule
    \textbf{Dataset} & \textbf{Category} & \textbf{\#Keys} &  \textbf{Data Distribution} \\
    \midrule
    \textsf{fb}        & Real      & 200M  & Higly skewed \\
    \textsf{books}     & Real      & 800M  & Mildly skewed \\
    \textsf{osm}       & Real      & 800M & Spatially complex, near-uniform \\
    \textsf{uniform}   & Synthetic & 200M  & Perfectly uniform \\
    \textsf{normal}    & Synthetic & 200M  & Perfectly Gaussian \\
    \textsf{lognormal} & Synthetic & 200M  & Perfectly right-skewed \\
    \bottomrule
    \end{tabular}
\end{table}

We provide three evaluation scenarios in order to answer questions \textbf{RQ1} to \textbf{RQ4}. 

\begin{itemize}[leftmargin=*]
    \item \textbf{Standalone Evaluation:} Each $\epsilon$-PLA fitting algorithm is compared against FRS (\cref{alg:segment}) to validate theoretical results and to provide a detailed analysis of their practical performance on a multi-threading platform. 
    \item \textbf{FIT:} Each $\epsilon$-PLA algorithm is integrated into the FITing-Tree~\cite{galakatos2019fiting}, a B$^+$-tree-like index where the last-mile tree search is replaced with a PLA-guided search. 
    We implement FIT by plugging an $\epsilon$-PLA algorithm into an STX B$^+$-tree~\cite{stx}. 
    \item \textbf{PGM:} Each $\epsilon$-PLA fitting algorithm is integrated into the PGM-Index~\cite{ferragina2020pgm}, a representative learned index framework that recursively applies $\epsilon$-PLA to predict index positions. 
    We extend the original implementation of PGM-Index~\cite{pgm} to support arbitrary $\epsilon$-PLA fitting algorithms.  
\end{itemize}

All methods are implemented in C++ and compiled using \texttt{g++} with the \texttt{-O0} optimization flag to prevent the compiler from applying vectorization optimizations to some algorithms while not others, thereby ensuring fairness and direct comparability in the performance comparison of the PLA algorithms.
All experiments are conducted on an Ubuntu 22.04 LTS server equipped with an Intel(R) Xeon(R) Gold 6430 CPU and 512 GB of RAM.

\subsection{Datasets and Query Workloads}
We adopt 3 commonly used real datasets from a recent learned index benchmark SOSD~\cite{marcus2020SOSD}. 
\textsf{fb} is a set of user IDs randomly sampled from Facebook~\cite{van2019ESISA}. 
\textsf{books} is a dataset of the popularity of books on Amazon. 
\textsf{osm} is a set of cell IDs from OpenStreetMap~\cite{openstreetmap}. 
In addition, we also generate three synthetic datasets by sampling from uniform, normal, and lognormal distributions. 
All keys in datasets are stored as 64-bit unsigned integers (\texttt{uint64\_t} in C++). 
\cref{tab:datasets_stats} summarizes the statistics of evaluated datasets.

For the query workload used in our evaluation, we uniformly sample 1,000 query keys from each dataset, ensuring that the same set of queries is used across all experiments. 
We focus primarily on the segment count of the resulting $\epsilon$-PLA model, as well as the space overhead and query time when integrated with PGM-Index and FITing-Tree. 
Each experiment is repeated 10 times, and we report the average values of the measured metrics.

\begin{figure}[t]
    \centering
    \begin{subfigure}[t]{0.48\columnwidth}
        \includegraphics[width=\linewidth]{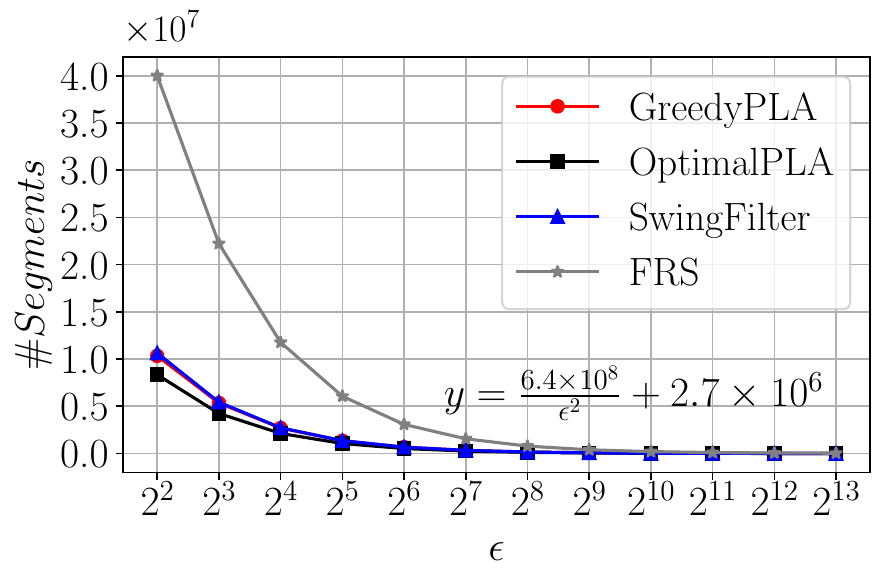}
        \caption{\textsf{fb}: \#Segments vs. $\epsilon$}
    \end{subfigure}
    \hfill
    \begin{subfigure}[t]{0.48\columnwidth}
        \includegraphics[width=\linewidth]{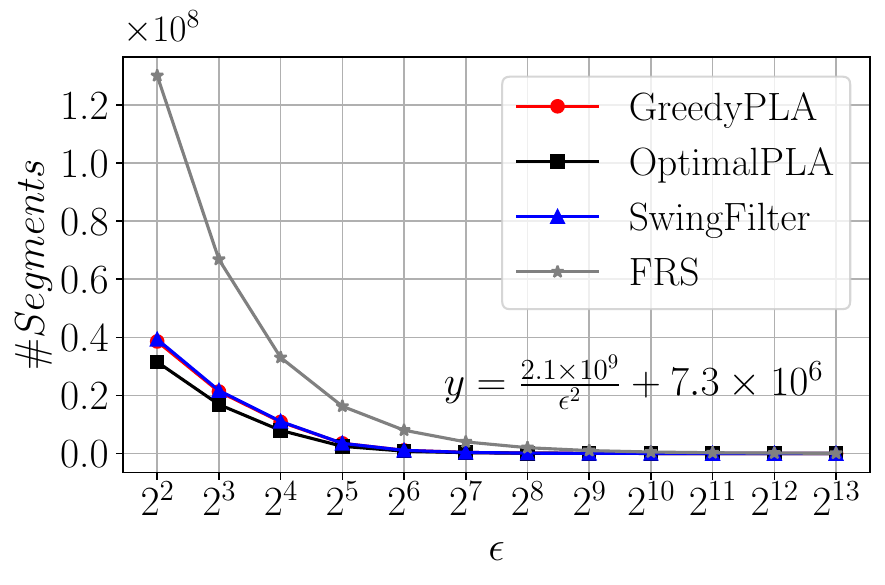}
        \caption{\textsf{books}: \#Segments vs. $\epsilon$}
    \end{subfigure}

    \begin{subfigure}[t]{0.48\columnwidth}
        \includegraphics[width=\linewidth]{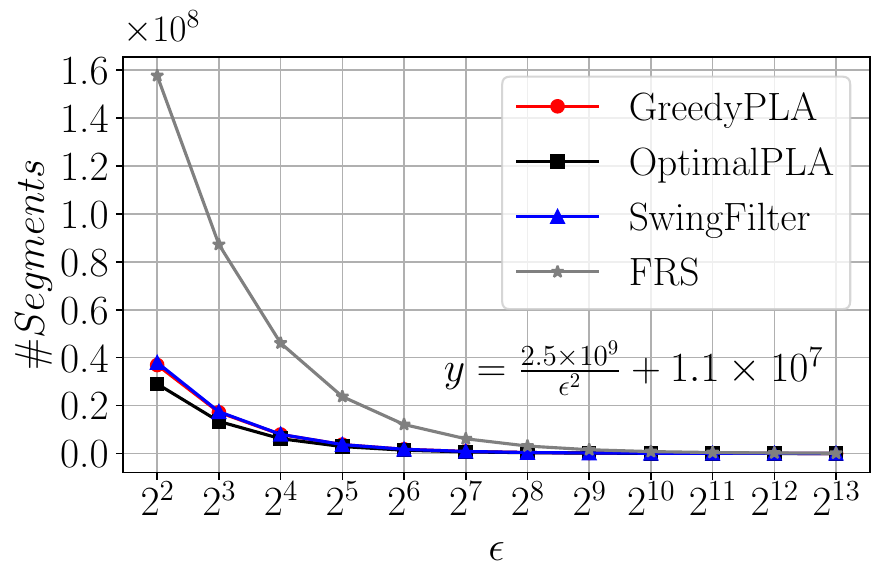}
        \caption{\textsf{osm}: \#Segments vs. $\epsilon$}
    \end{subfigure}
    \hfill
    \begin{subfigure}[t]{0.48\columnwidth}
        \includegraphics[width=\linewidth]{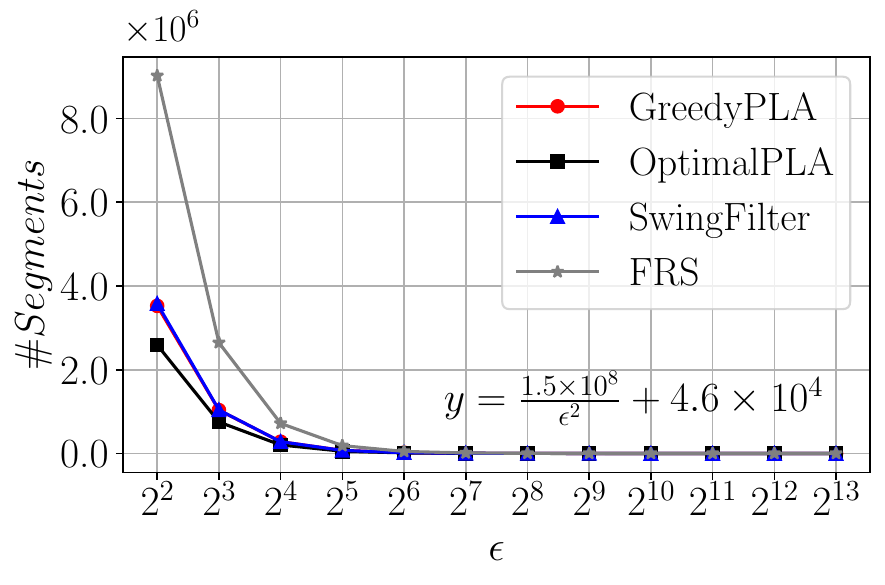}
        \caption{\textsf{uniform}: \#Segments vs. $\epsilon$}
    \end{subfigure}

    \begin{subfigure}[t]{0.48\columnwidth}
        \includegraphics[width=\linewidth]{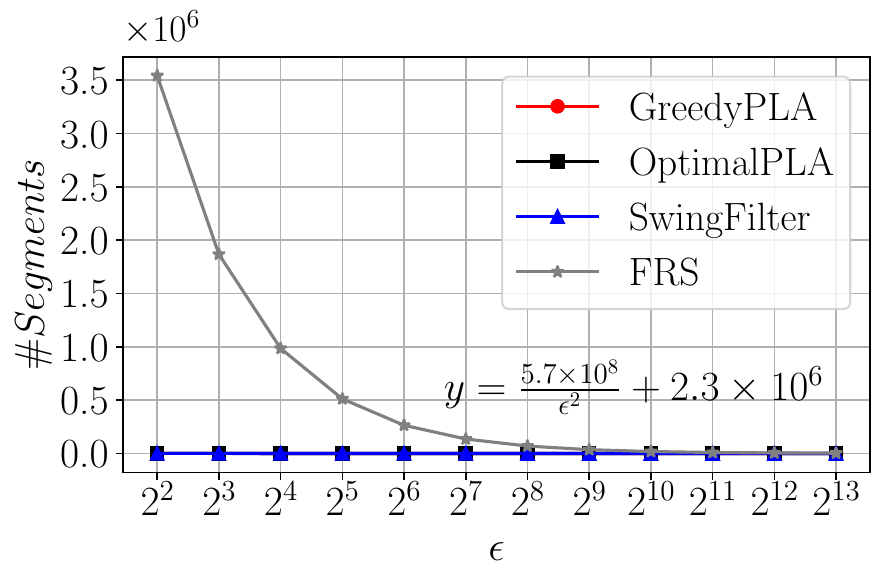}
        \caption{\textsf{normal}: \#Segments vs. $\epsilon$}
    \end{subfigure}
    \hfill
    \begin{subfigure}[t]{0.48\columnwidth}
        \includegraphics[width=\linewidth]{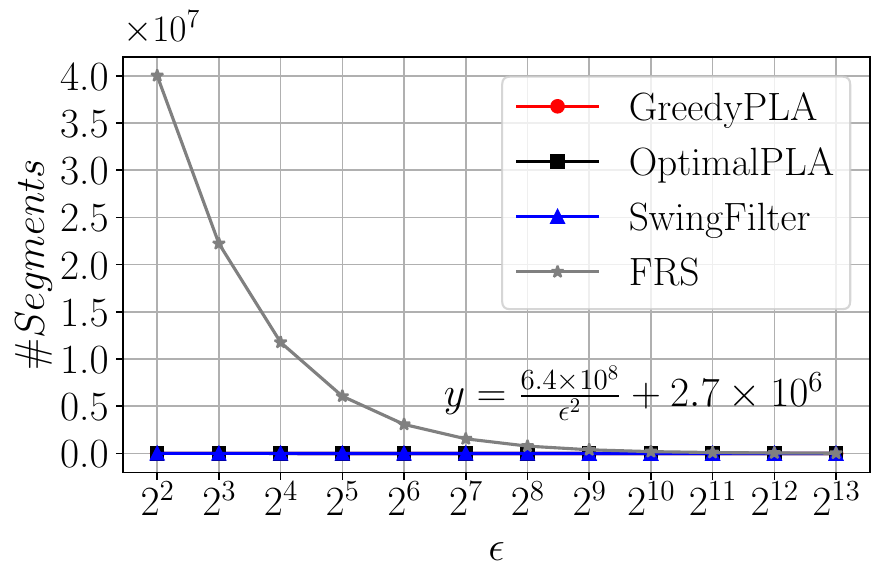}
        \caption{\textsf{lognormal}: Segments vs. $\epsilon$}
    \end{subfigure}

    \caption{Segment count w.r.t.~$\epsilon\in\{2^2, 2^3,\cdots, 2^{13}\}$ for different $\epsilon$-PLA fitting algorithms. The fitted curve for FRS is shown in the figure as $\text{\#Segments} = a \cdot \epsilon^{-2} + b$. Note that, $\text{\#segments}\propto n/\text{coverage}$.}
    \label{fig:pla_segment_cnt}
    \vspace{-2ex}
\end{figure}

\section{Experimental Results}\label{sec:exp_results}

In this section, we present the benchmark results addressing the research questions outlined earlier. 
Specifically, \cref{subsec:exp_standalone} reports the standalone evaluation of $\epsilon$-PLA fitting algorithms (\textbf{RQ1} and \textbf{RQ2}). 
\cref{subsec:exp_fit} and \cref{subsec:exp_pgm} show the results when integrated with FITing-Tree and PGM-Index, respectively (\textbf{RQ3}). 
Finally, \cref{subsec:exp_parameter} analyzes the impact of the error parameter $\epsilon$ (\textbf{RQ4}).

\begin{figure}[t]
    \centering
    \begin{subfigure}[t]{0.48\columnwidth}
        \includegraphics[width=\linewidth]{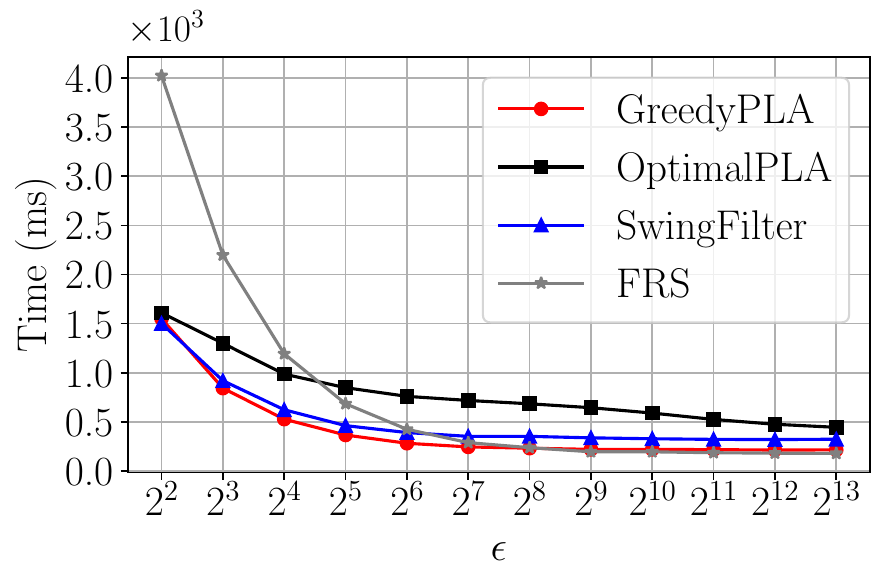}
        \caption{\textsf{fb}: Build Time vs. $\epsilon$}
    \end{subfigure}
    \hfill
    \begin{subfigure}[t]{0.48\columnwidth}
        \includegraphics[width=\linewidth]{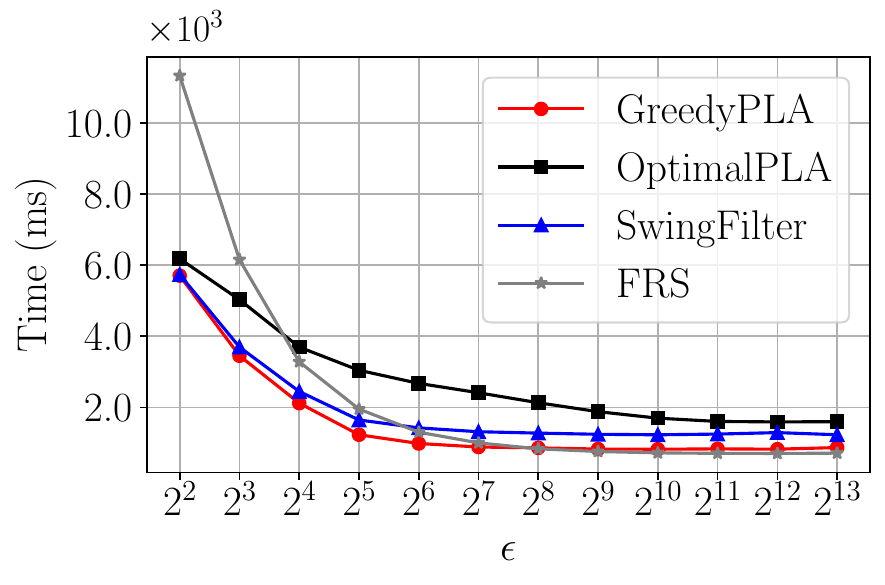}
        \caption{\textsf{books}: Build Time vs. $\epsilon$}
    \end{subfigure}
    
    \begin{subfigure}[t]{0.48\columnwidth}
        \includegraphics[width=\linewidth]{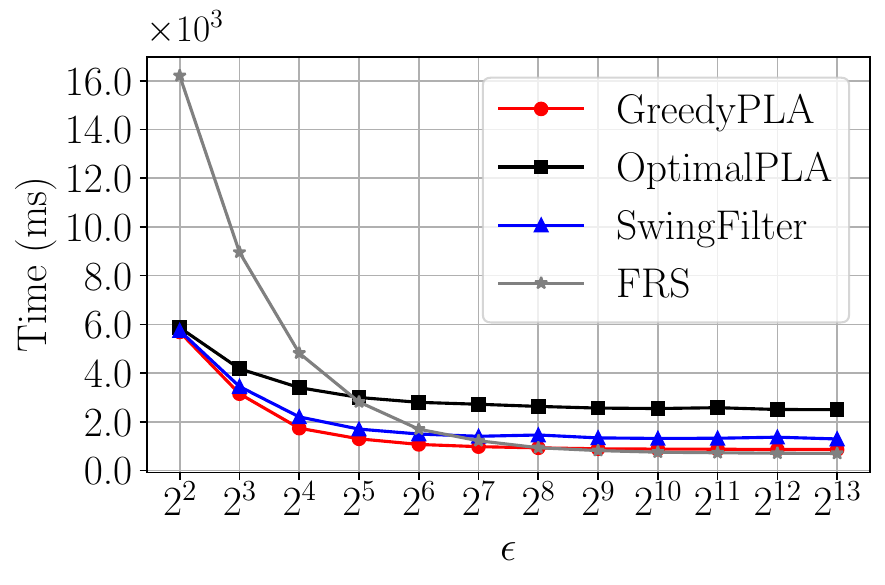}
        \caption{\textsf{osm}: Build Time vs. $\epsilon$}
    \end{subfigure}
    \hfill
    \begin{subfigure}[t]{0.48\columnwidth}
        \includegraphics[width=\linewidth]{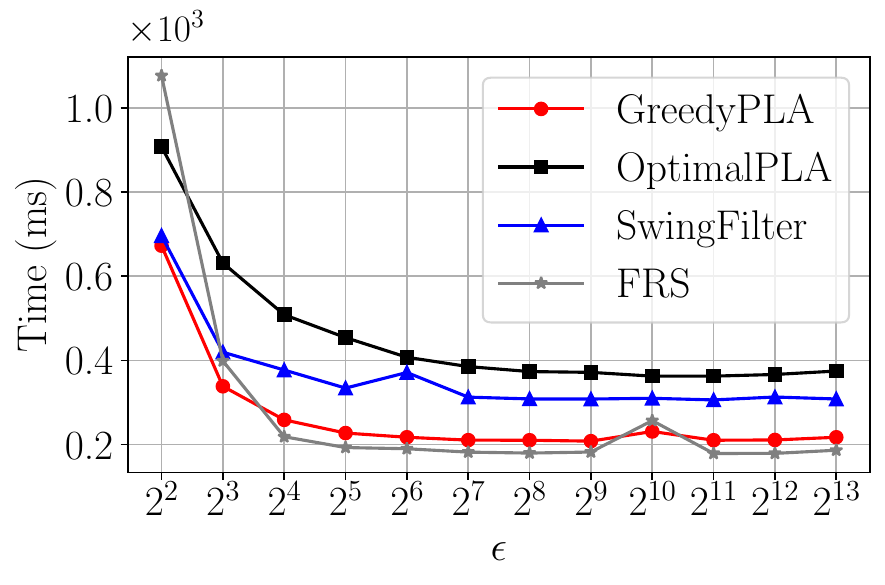}
        \caption{\textsf{uniform}: Build Time vs. $\epsilon$}
    \end{subfigure}
    
    \begin{subfigure}[t]{0.48\columnwidth}
        \includegraphics[width=\linewidth]{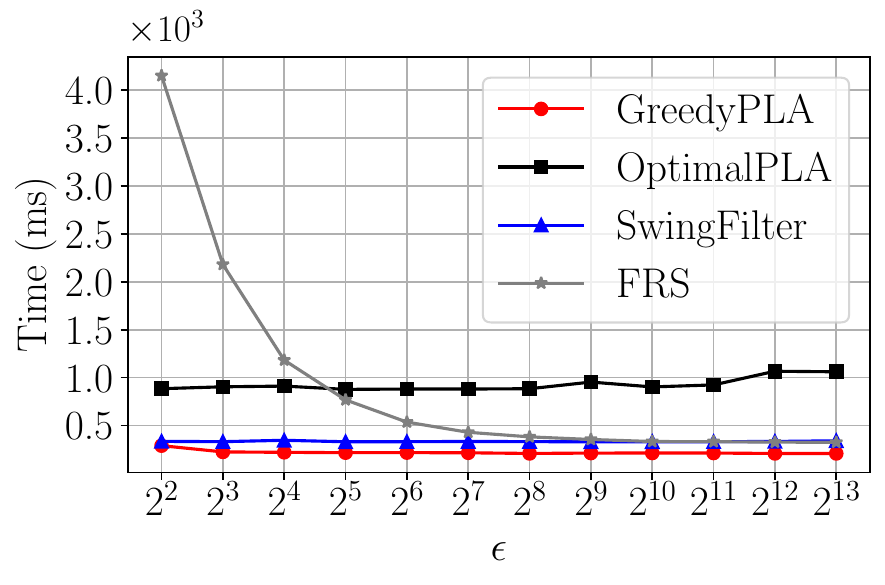}
        \caption{\textsf{normal}: Build Time vs. $\epsilon$}
    \end{subfigure}
    \hfill
    \begin{subfigure}[t]{0.48\columnwidth}
        \includegraphics[width=\linewidth]{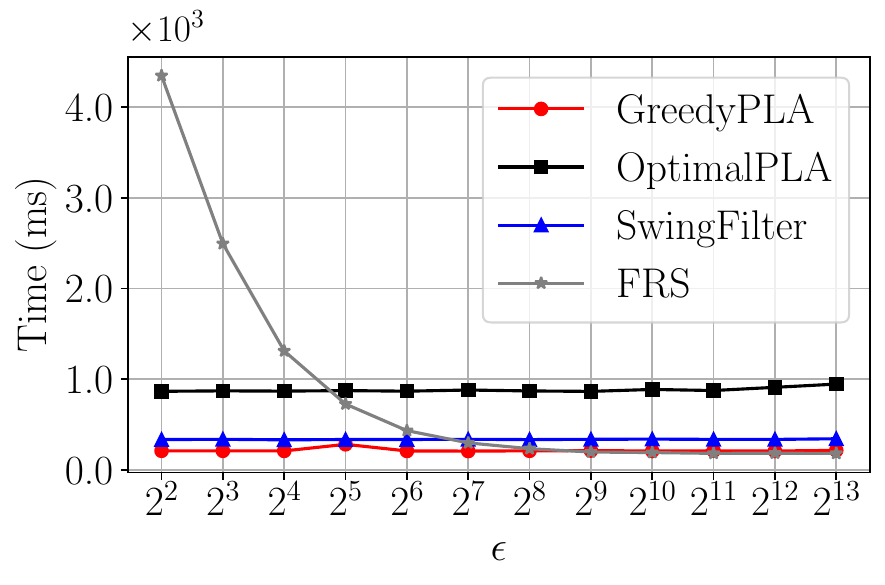}
        \caption{\textsf{lognormal}: Build Time vs. $\epsilon$}
    \end{subfigure}
    \caption{Construction time w.r.t.~$\epsilon\in\{2^2, 2^3,\cdots, 2^{13}\}$ for different $\epsilon$-PLA fitting algorithms.}
    \vspace{-2ex}
    \label{fig:pla_construction_time}
\end{figure}

\subsection{Standalone Evaluation Results}\label{subsec:exp_standalone}
\subsubsection{$\epsilon$-PLA Segment Count (\textbf{RQ1})}

\cref{fig:pla_segment_cnt} reports the number of segments produced by each $\epsilon$-PLA fitting algorithm when varying the error bound $\epsilon$ from $2^2$ to $2^{13}$. 
The experimental results clearly show that, for FRS, the number of segments is inversely proportional to $\epsilon^2$ across both real and synthetic datasets, aligning with the lower bounds established in \cref{theorem:index_coverage_uniform} and \cref{theorem:index_coverage_general}. 

When compared to other $\epsilon$-PLA algorithms, the experimental results clearly show that the number of segments produced by FRS consistently upper-bounds those generated by other methods. 
For example, on the \textsf{fb} dataset, FRS yields approximately $\mathbf{29.34}\times$ to $\mathbf{41.31}\times$ more segments than OptimalPLA, and $\mathbf{24.05}\times$ to $\mathbf{28.89}\times$ more than GreedyPLA and SwingFilter. 
These findings align with the theoretical analysis in \cref{theorm:superiority_of_greedy}, which demonstrates the superiority of SwingFilter and GreedyPLA over FRS in terms of segment coverage. 
In addition, an interesting observation is that the optimal algorithm shows a significant advantage when the error bound $\epsilon$ is small. 
However, when $\epsilon$ exceeds $2^7=128$, both optimal and greedy $\epsilon$-PLA algorithms achieve similar segment counts on both real and synthetic datasets.

\textbf{Takeaways.} 
We empirically validate the derived bounds for a wide range of $\epsilon$-PLA algorithms. 
In practice, OptimalPLA can reduce the number of segments by up to $\mathbf{1.4}\times$ compared to greedy algorithms such as SwingFilter and GreedyPLA. 
However, this performance gap narrows as $\epsilon$ grows; when $\epsilon>128$, both optimal and greedy algorithms produce a similar number of segments.

\subsubsection{$\epsilon$-PLA Construction Time}
We further evaluate the construction time for different $\epsilon$-PLA algorithms, and the results are presented in \cref{fig:pla_construction_time}. 
Although all the compared $\epsilon$-PLA algorithms have a construction time complexity of $O(n)$ (see \cref{tab:pla_algorithms}), their actual wall-clock time varies notably with different $\epsilon$ values. 
For example, when $\epsilon$ is small (e.g., $\epsilon\leq 2^{5}$), the simple FRS algorithm takes significantly longer time than the much more sophisticated OptimalPLA. 
This is because FRS generates a large number of segments under small $\epsilon$, and the overhead of storing these segments into in-memory data structures (e.g., \texttt{std::vector}) dominates the total construction time. 
When $\epsilon$ increases, the segment counts produced by all algorithms converge, this overhead diminishes, and FRS eventually outperforms OptimalPLA due to its simplicity.  

A similar trend can be observed for the other two greedy algorithms. 
SwingFilter, GreedyPLA, and OptimalPLA demonstrate comparable performance for small $\epsilon$ values (e.g., $\epsilon=2^2$). 
However, as $\epsilon$ increases, the greedy algorithm's efficiency advantage becomes more significant, achieving up to a $\mathbf{2.0}\times$ speedup over OptimalPLA.

\textbf{Takeaways.} 
GreedyPLA and SwingFilter can achieve a segment count comparable to OptimalPLA when the error bound $\epsilon$ is large, while saving up to half of the construction time, demonstrating strong potential for practical use. 
Notably, FRS is excluded from this discussion, as it is a constructive algorithm designed only to assist analysis. 

\begin{figure}[t]
    \centering
    \begin{subfigure}[t]{0.48\columnwidth}
        \includegraphics[width=\linewidth]{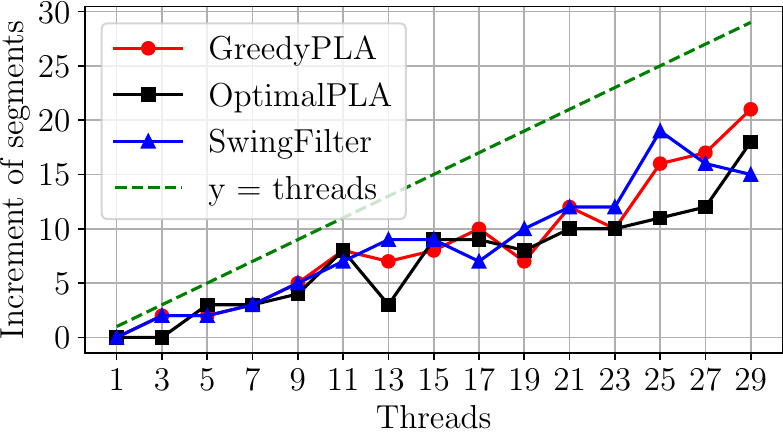}
        \caption{\textsf{fb}: Increment of \#Segments vs. \#Tthreads}
    \end{subfigure}
    \hfill
    \begin{subfigure}[t]{0.48\columnwidth}
        \includegraphics[width=\linewidth]{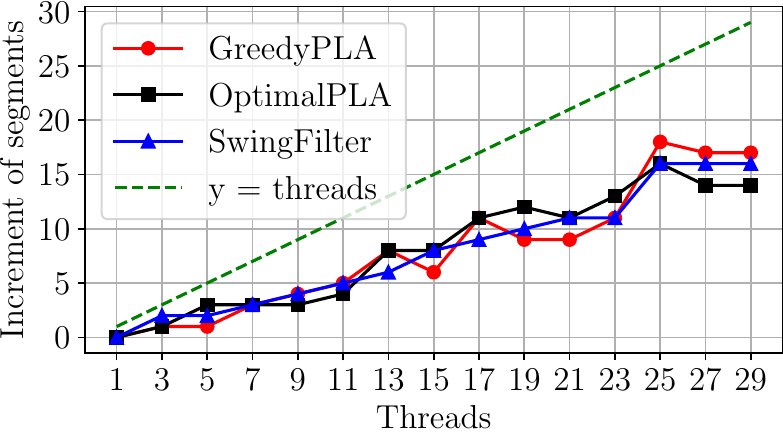}
        \caption{\textsf{books}: Increment of \#Segments vs. \#Threads}
    \end{subfigure}

    \begin{subfigure}[t]{0.48\columnwidth}
        \includegraphics[width=\linewidth]{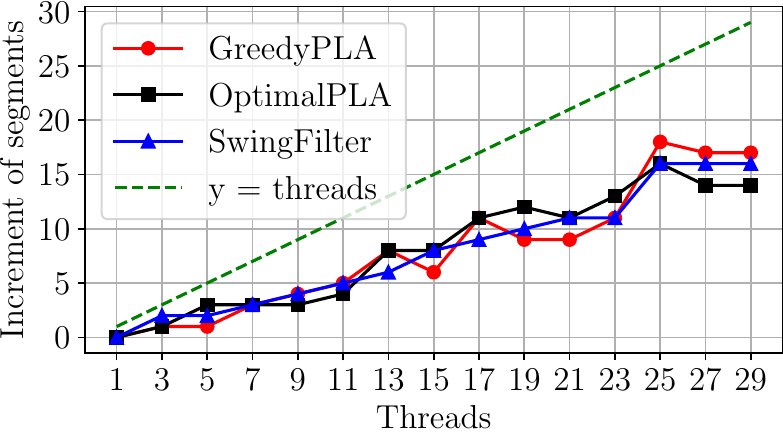}
        \caption{\textsf{osm}: Increment of \#Segments vs. \#Threads}
    \end{subfigure}
    \hfill
    \begin{subfigure}[t]{0.48\columnwidth}
        \includegraphics[width=\linewidth]{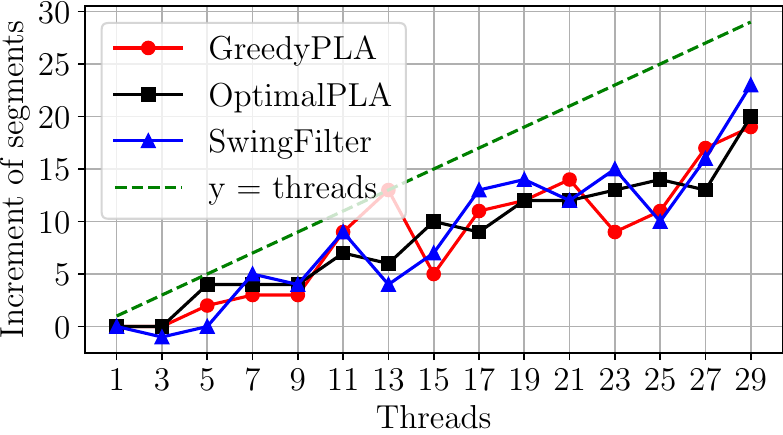}
        \caption{\textsf{uniform}: Increment of \#Segments vs. \#Threads}
    \end{subfigure}

    \begin{subfigure}[t]{0.48\columnwidth}
        \includegraphics[width=\linewidth]{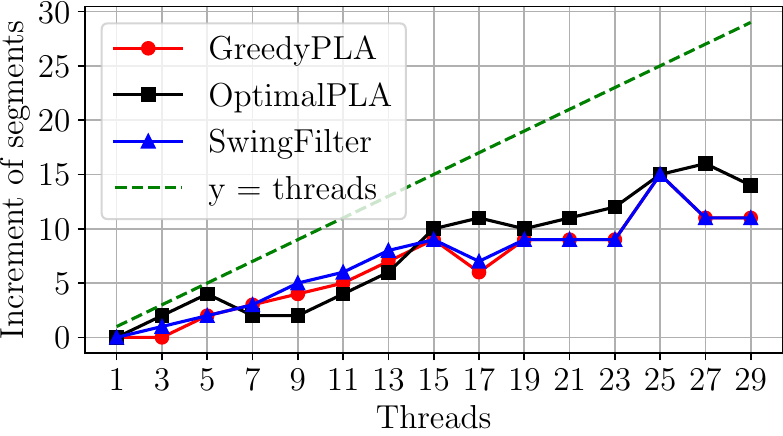}
        \caption{\textsf{normal}: Increment of \#Segments vs. \#Threads}
    \end{subfigure}
    \hfill
    \begin{subfigure}[t]{0.48\columnwidth}
        \includegraphics[width=\linewidth]{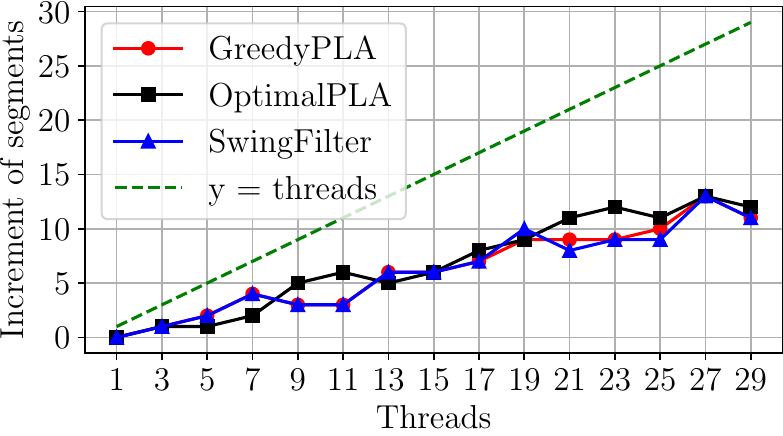}
        \caption{\textsf{lognormal}: Increment of \#Segments vs. \#Threads}
    \end{subfigure}

    \caption{Increment of segment count w.r.t.~the number of threads. The green dashed line marks $y=\text{\#Threads}$.}
    \label{fig:segment_multi_thread}

\end{figure}

\subsubsection{Parallel $\epsilon$-PLA Evaluation (\textbf{RQ2})}

The previous experiments focus on the performance of \textbf{serial} $\epsilon$-PLA algorithms under single-threaded execution. 
We now turn to a more practical setting by examining their behavior in a multi-threaded environment. 
As detailed in \cref{subsec:impl}, we adopt a straightforward data-parallel strategy that partitions the input data and applies the PLA algorithm independently to each partition.
This inevitably introduces additional segments due to data partition. 
To quantify the impact of multi-threading on $\epsilon$-PLA fitting quality, we report the increase in segment count by varying the number of threads (i.e., data partitions), as shown in \cref{fig:segment_multi_thread}. 

The results clearly show that, for both optimal and greedy algorithms, the number of additional segments introduced by multi-threading is upper bounded by the number of threads. 
We now provide a formal proof of this property for the optimal $\epsilon$-PLA algorithm. 

\begin{theorem}[Upper Bound of Segment Increase for Parallel OptimalPLA]\label{theorem:threads}
    Given an input key set $\mathcal{K}$, suppose that the serial OptimalPLA produces $m$ segments over $\mathcal{K}$. 
    When using $t$ threads to partition and process the data independently, the parallel version produces \textbf{at most} $m + t - 1$ segments. 
\end{theorem}
\begin{proof}
Let $p$ denote the number of segments produced by parallel OptimalPLA. 
Observe that any segmentation of the full input sequence (including the serial OptimalPLA result) induces a valid segmentation for each chunk. 
However, chunk-local OptimalPLA computes the \textbf{minimum} number of segments per chunk, so it cannot produce \textbf{more than one} extra segment per boundary compared to the serial segmentation. 
As there are $t-1$ boundaries, the total increase in segment count is at most $t-1$. 
Thus, $p \leq m + (t - 1)$. 
\end{proof}

\textbf{Takeaways.} 
The empirical and theoretical analyses in this section show that the impact of multi-threading on the results of the optimal $\epsilon$-PLA algorithm is bounded by the number of threads. 
In practice, the total number of segments is typically much larger than the number of threads (i.e., $m\gg t$), making this overhead negligible. 
Notably, while \cref{theorem:threads} relies on the optimality guarantees of OptimalPLA, the results can be extended to greedy algorithms by examining the chunk boundaries using a similar approach. 

\begin{figure}[t]
    \centering
    \begin{subfigure}[t]{0.48\columnwidth}
        \includegraphics[width=\linewidth]{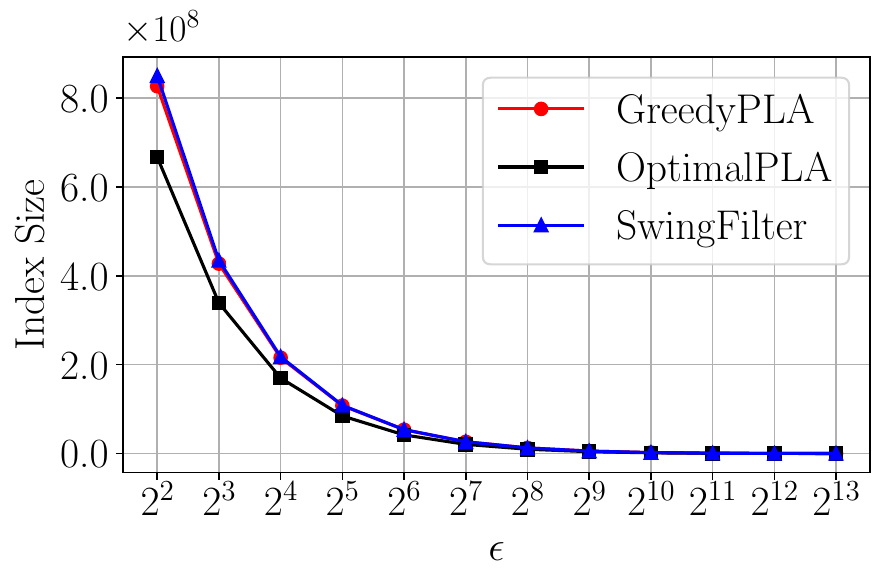}
        \caption{\textsf{fb}: Index size vs. $\epsilon$}
    \end{subfigure}
    \hspace{-0.5em}
    \begin{subfigure}[t]{0.48\columnwidth}
        \includegraphics[width=\linewidth]{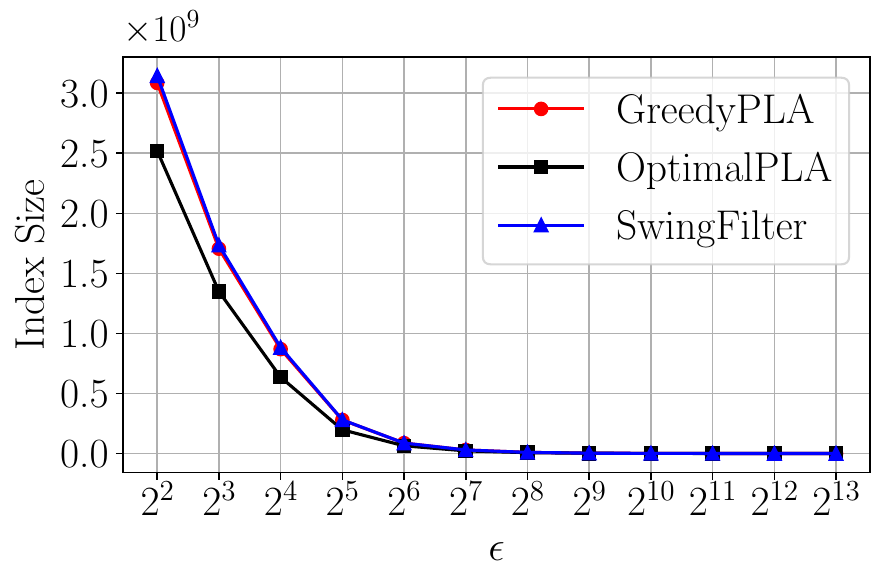}
        \caption{\textsf{books}: Index size vs. $\epsilon$}
    \end{subfigure}
    
    \begin{subfigure}[t]{0.48\columnwidth}
        \includegraphics[width=\linewidth]{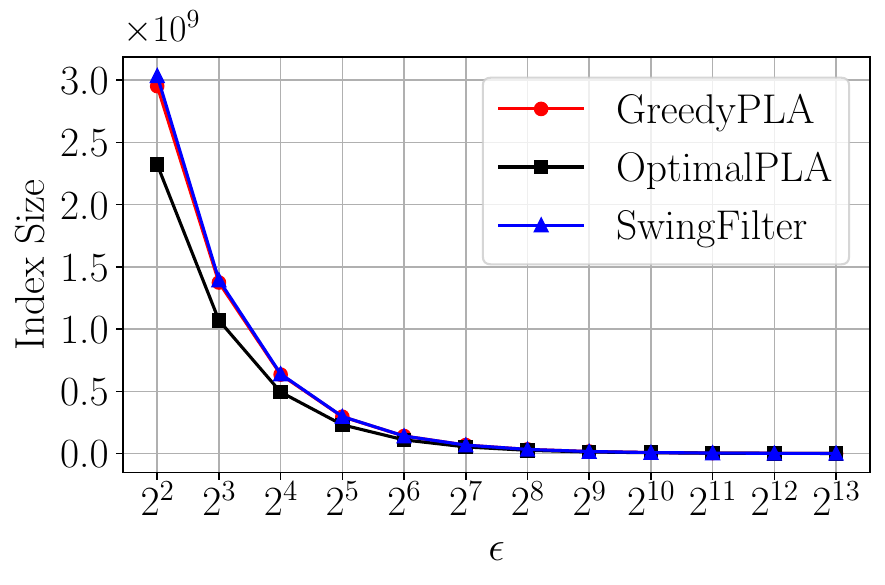}
        \caption{\textsf{osm}: Index size vs. $\epsilon$}
    \end{subfigure}
    \hspace{-0.5em}
    \begin{subfigure}[t]{0.48\columnwidth}
        \includegraphics[width=\linewidth]{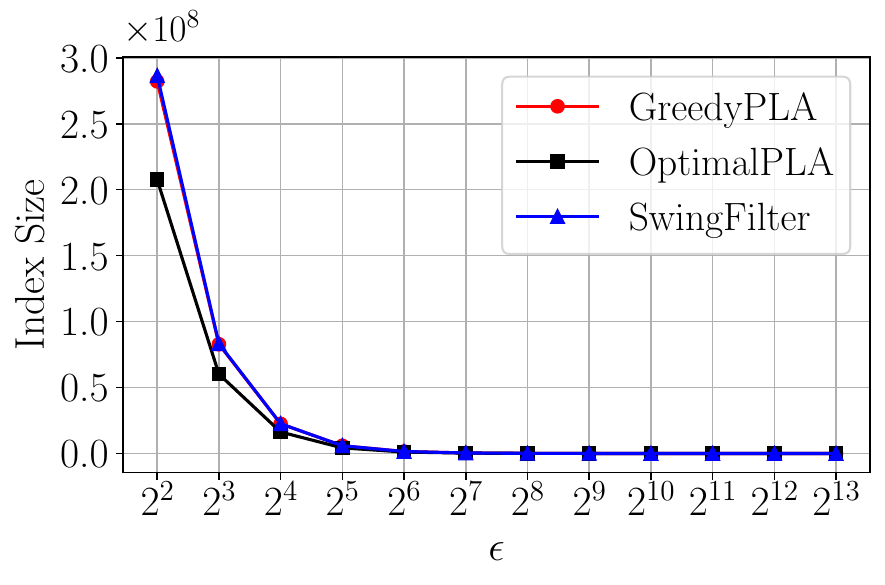}
        \caption{\textsf{uniform}: Index size vs. $\epsilon$}
    \end{subfigure}
    
    \begin{subfigure}[t]{0.48\columnwidth}
        \includegraphics[width=\linewidth]{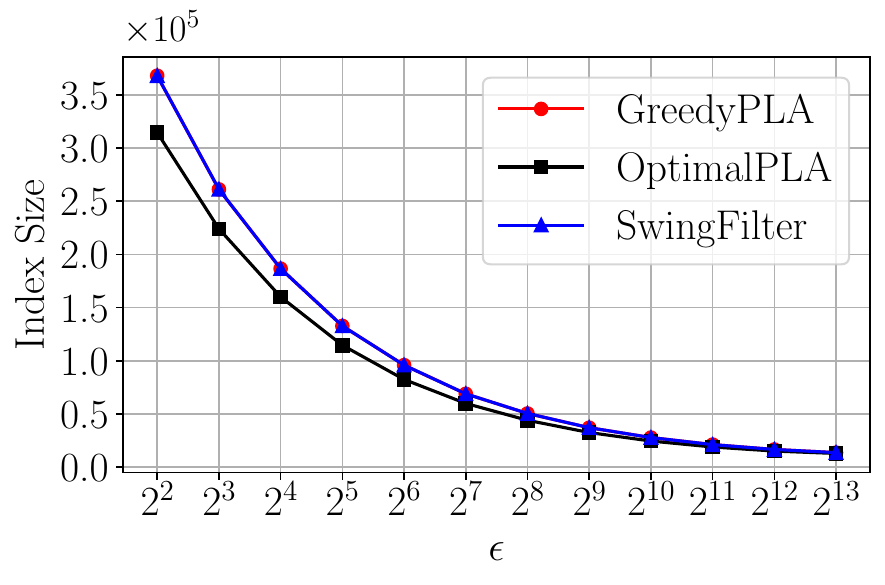}
        \caption{\textsf{normal}: Index size vs. $\epsilon$}
    \end{subfigure}
    \hspace{-0.5em}
    \begin{subfigure}[t]{0.48\columnwidth}
        \includegraphics[width=\linewidth]{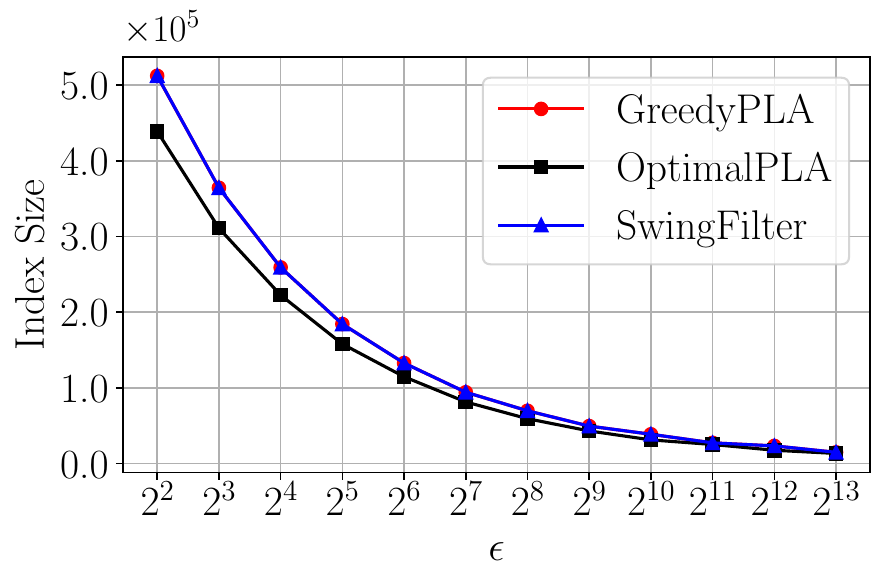}
        \caption{\textsf{lognormal}: Index size vs. $\epsilon$}
    \end{subfigure}
    \caption{Index size of FIT (FITing-Tree + $\epsilon$-PLA) w.r.t.~$\epsilon\in\{2^2, 2^3,\cdots, 2^{13}\}$ for different $\epsilon$-PLA algorithms. The fanout of the internal B$^+$-tree is fixed to 16, so that each internal tree node occupies 256 bytes.}
    \label{fig:fit_indexS}
\end{figure}

\begin{figure}[t]
    \centering
    \begin{subfigure}[t]{0.48\columnwidth}
        \includegraphics[width=\linewidth]{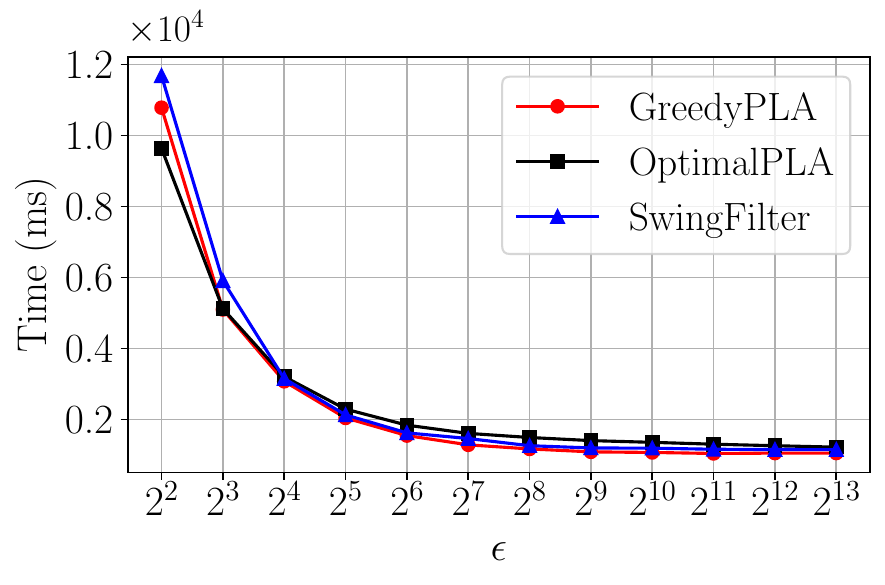}
        \caption{\textsf{fb}: Build Time vs. $\epsilon$}
    \end{subfigure}
    \hspace{-0.5em}
    \begin{subfigure}[t]{0.48\columnwidth}
        \includegraphics[width=\linewidth]{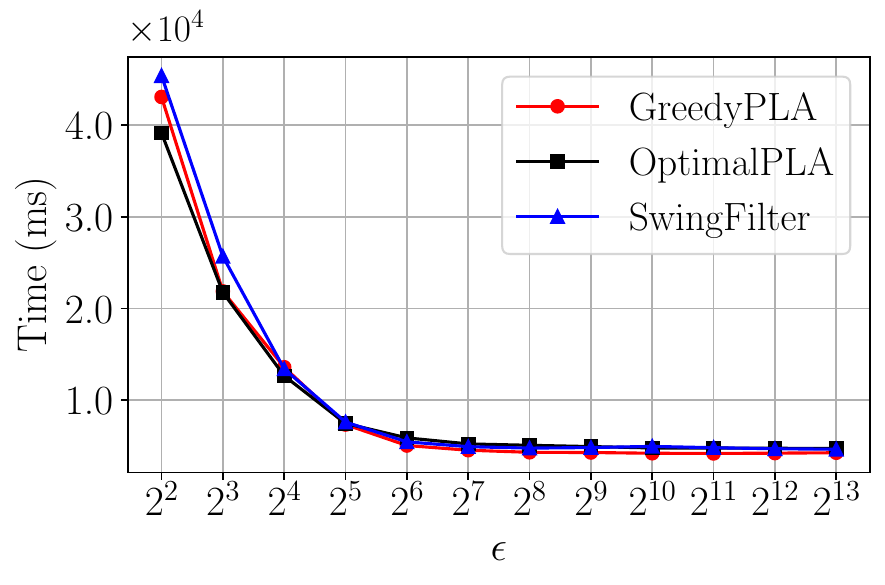}
        \caption{\textsf{books}: Build Time vs. $\epsilon$}
    \end{subfigure}
    
    \begin{subfigure}[t]{0.48\columnwidth}
        \includegraphics[width=\linewidth]{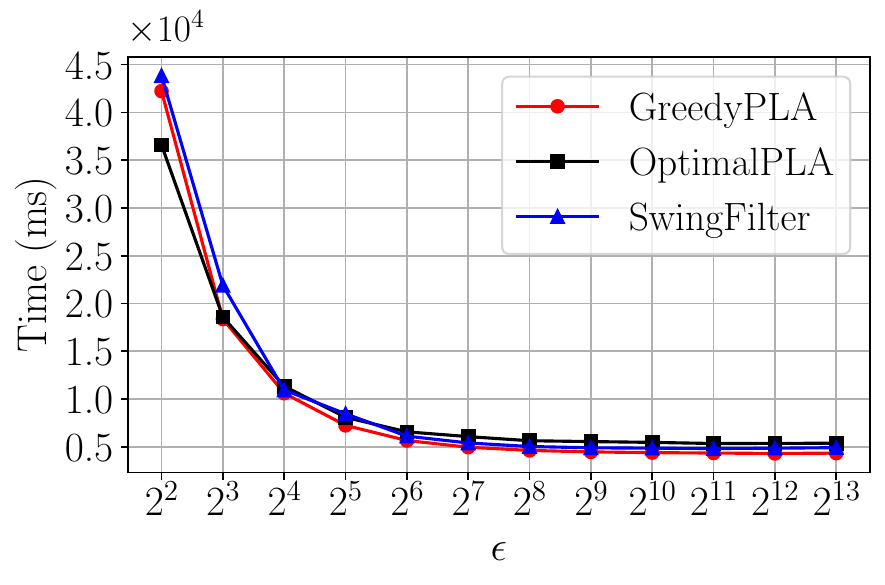}
        \caption{\textsf{osm}: Build Time vs. $\epsilon$}
    \end{subfigure}
    \hspace{-0.5em}
    \begin{subfigure}[t]{0.48\columnwidth}
        \includegraphics[width=\linewidth]{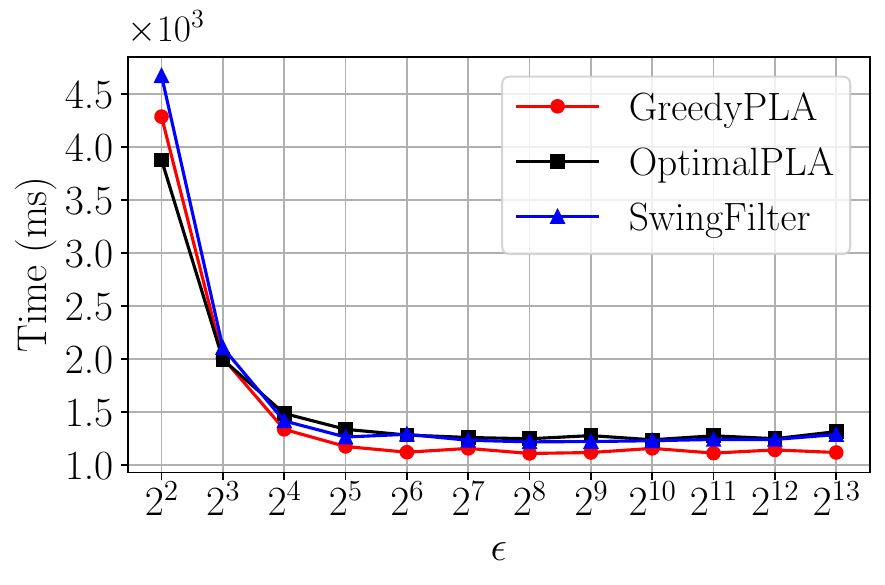}
        \caption{\textsf{uniform}: Build Time vs. $\epsilon$}
    \end{subfigure}
    
    \begin{subfigure}[t]{0.48\columnwidth}
        \includegraphics[width=\linewidth]{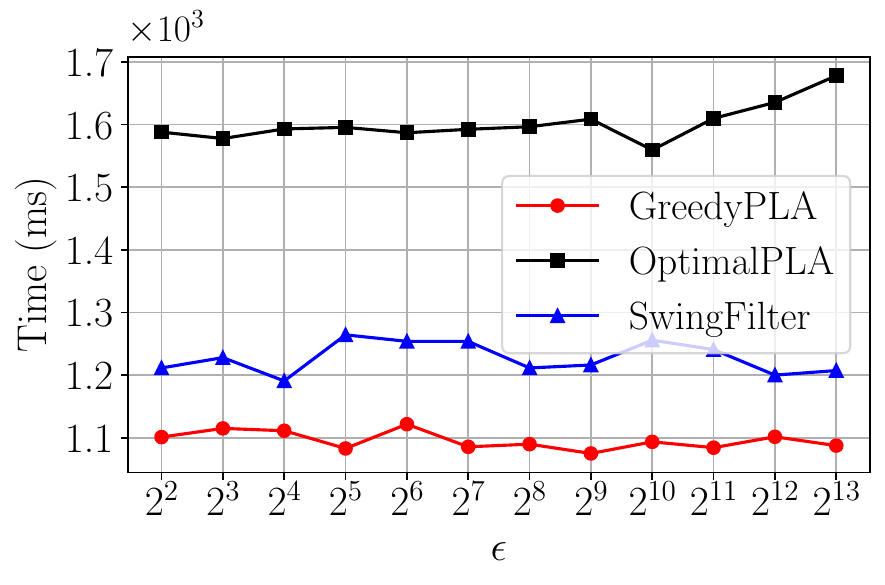}
        \caption{\textsf{normal}: Build Time vs. $\epsilon$}
    \end{subfigure}
    \hspace{-0.5em}
    \begin{subfigure}[t]{0.48\columnwidth}
        \includegraphics[width=\linewidth]{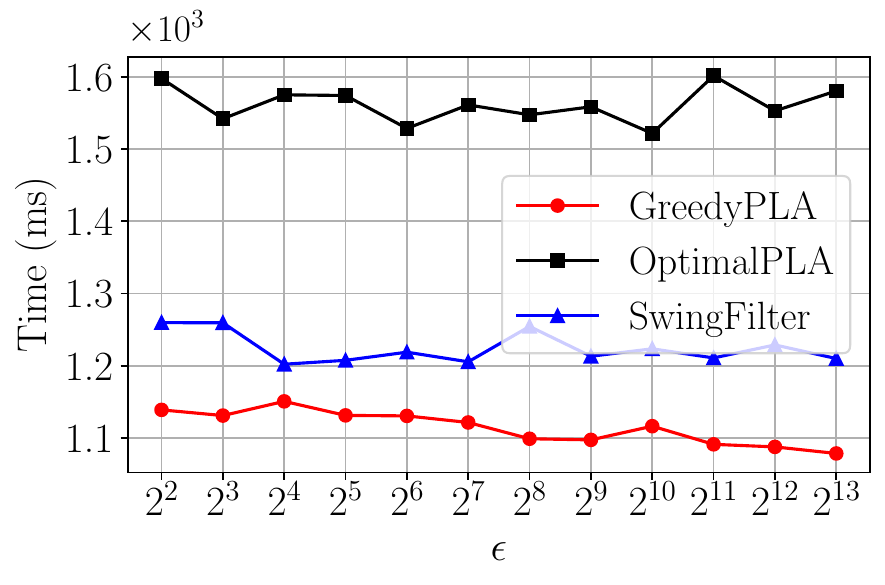}
        \caption{\textsf{lognormal}: Build Time vs. $\epsilon$}
    \end{subfigure}
    \caption{Construction time of FIT (FITing-Tree + $\epsilon$-PLA) w.r.t.~$\epsilon\in\{2^2, 2^3,\cdots, 2^{13}\}$ for different $\epsilon$-PLA algorithms. }
    \label{fig:fit_btime}
\end{figure}

\begin{figure}[t]
    \centering
    \begin{subfigure}[t]{0.48\columnwidth}
        \includegraphics[width=\linewidth]{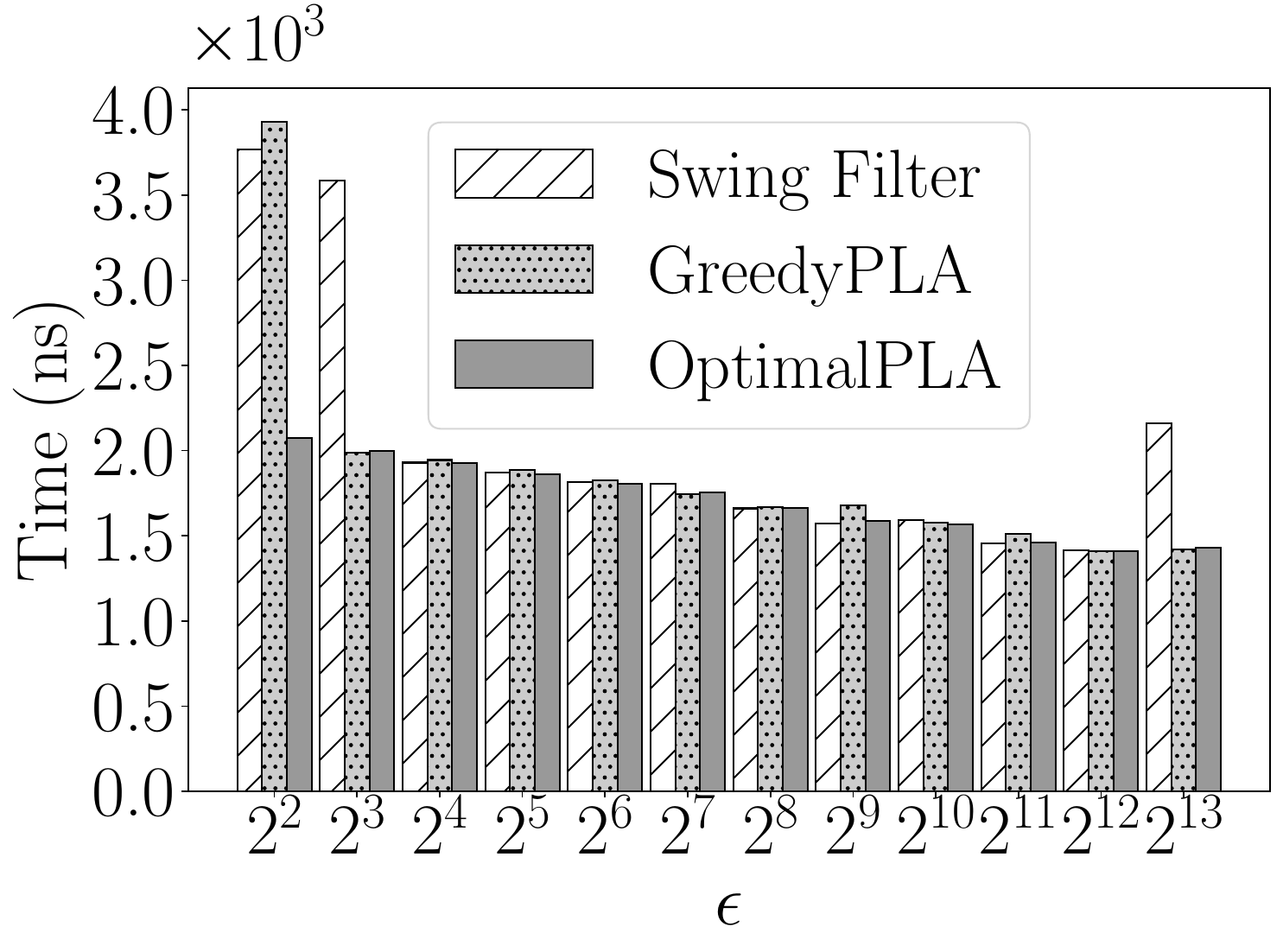}
        \caption{\textsf{fb}: Query Time vs. $\epsilon$}
    \end{subfigure}
    \hspace{-0.5em}
    \begin{subfigure}[t]{0.48\columnwidth}
        \includegraphics[width=\linewidth]{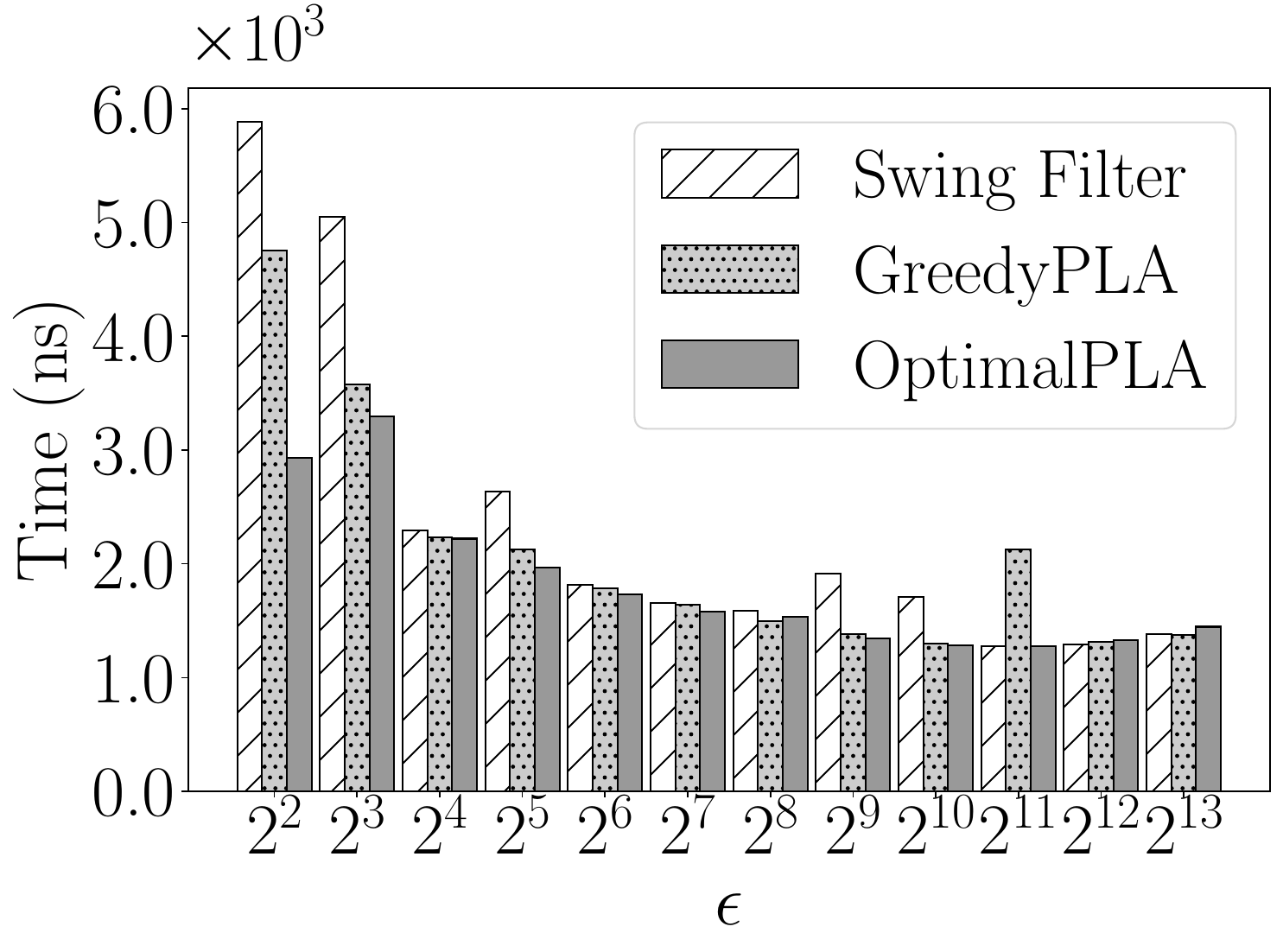}
        \caption{\textsf{books}: Query Time vs. $\epsilon$}
    \end{subfigure}
    
    \begin{subfigure}[t]{0.48\columnwidth}
        \includegraphics[width=\linewidth]{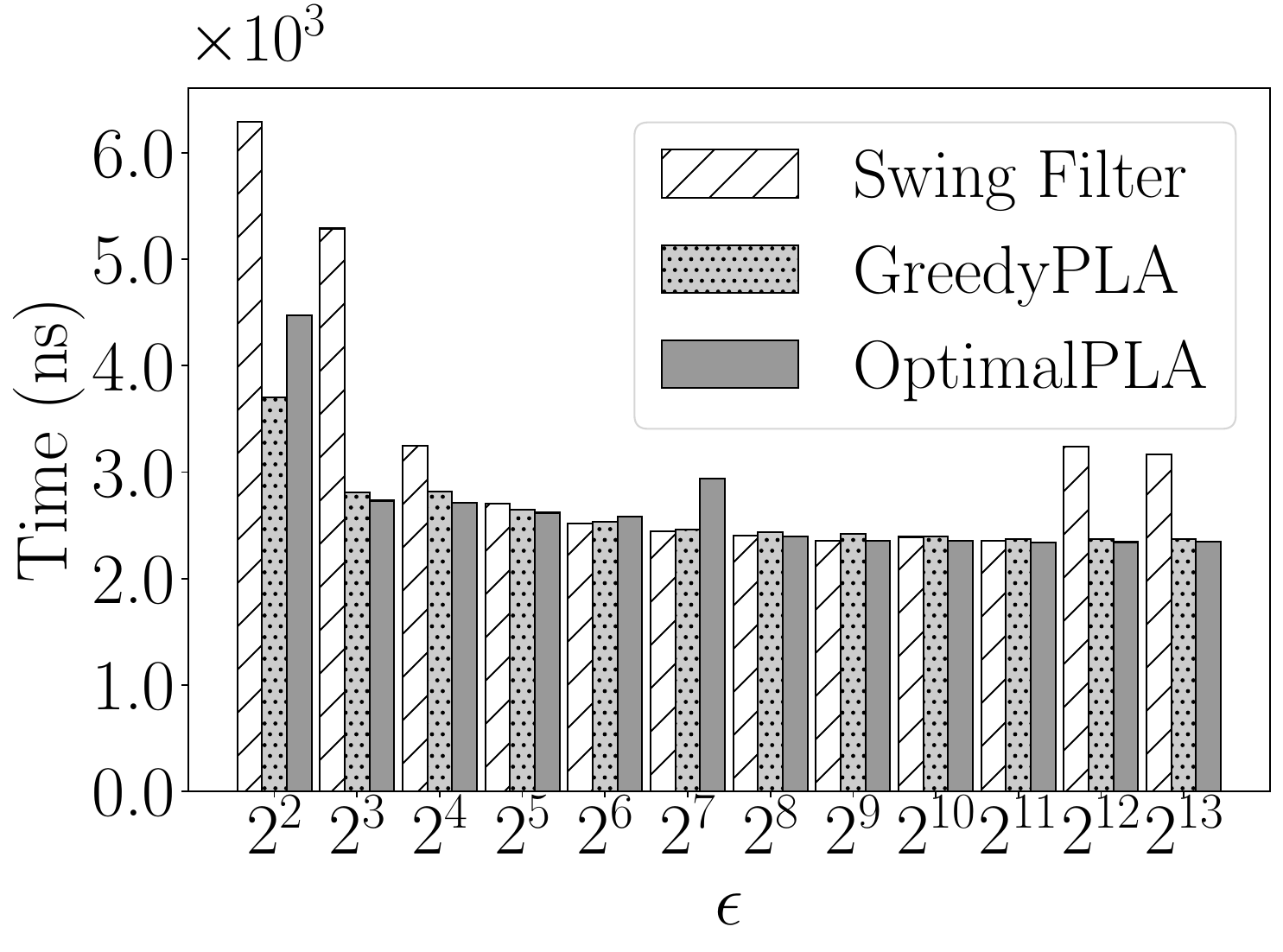}
        \caption{\textsf{osm}: Query Time vs. $\epsilon$}
    \end{subfigure}
    \hspace{-0.5em}
    \begin{subfigure}[t]{0.48\columnwidth}
        \includegraphics[width=\linewidth]{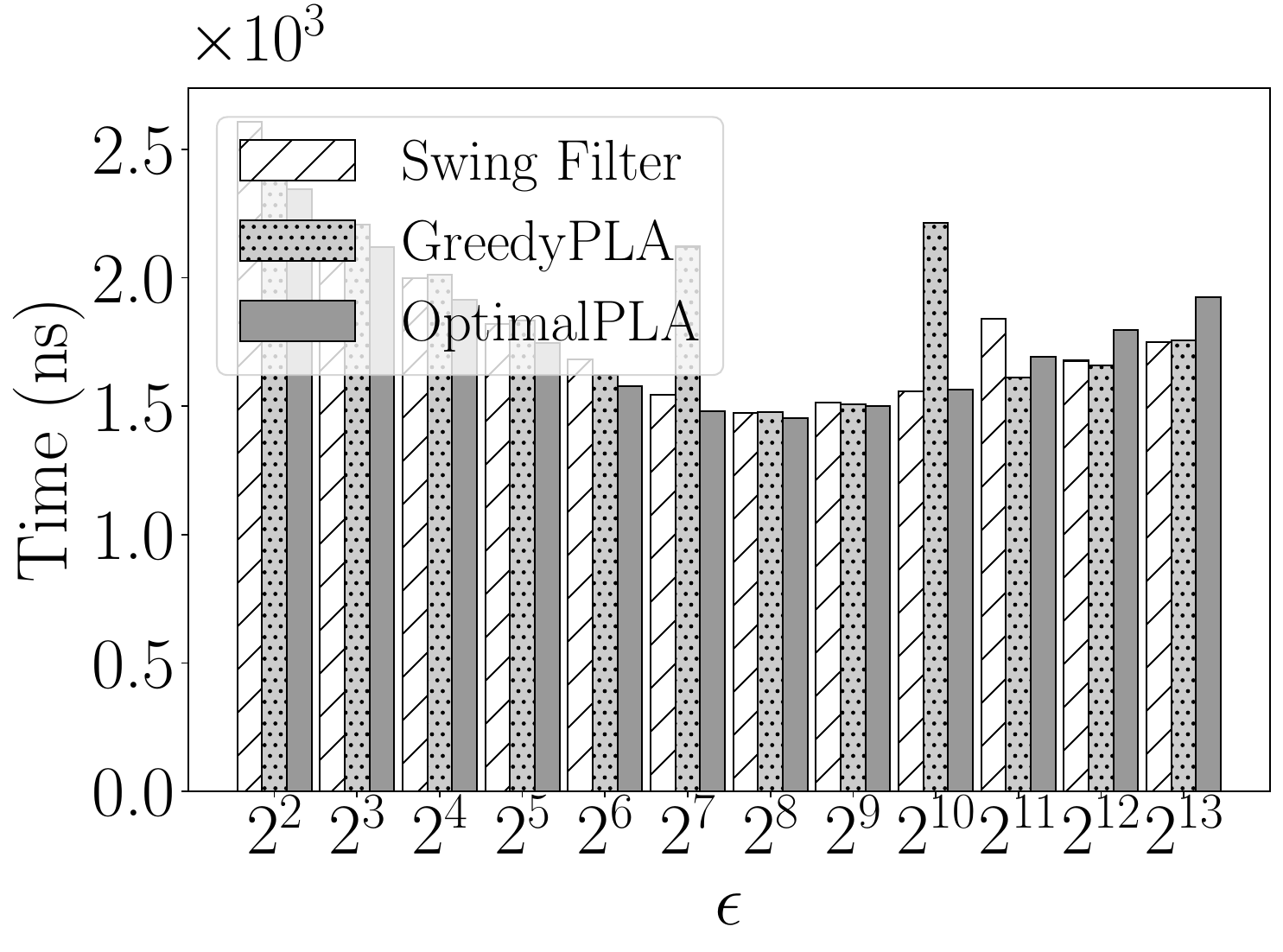}
        \caption{\textsf{uniform}: Query Time vs. $\epsilon$}
    \end{subfigure}
    
    \begin{subfigure}[t]{0.48\columnwidth}
        \includegraphics[width=\linewidth]{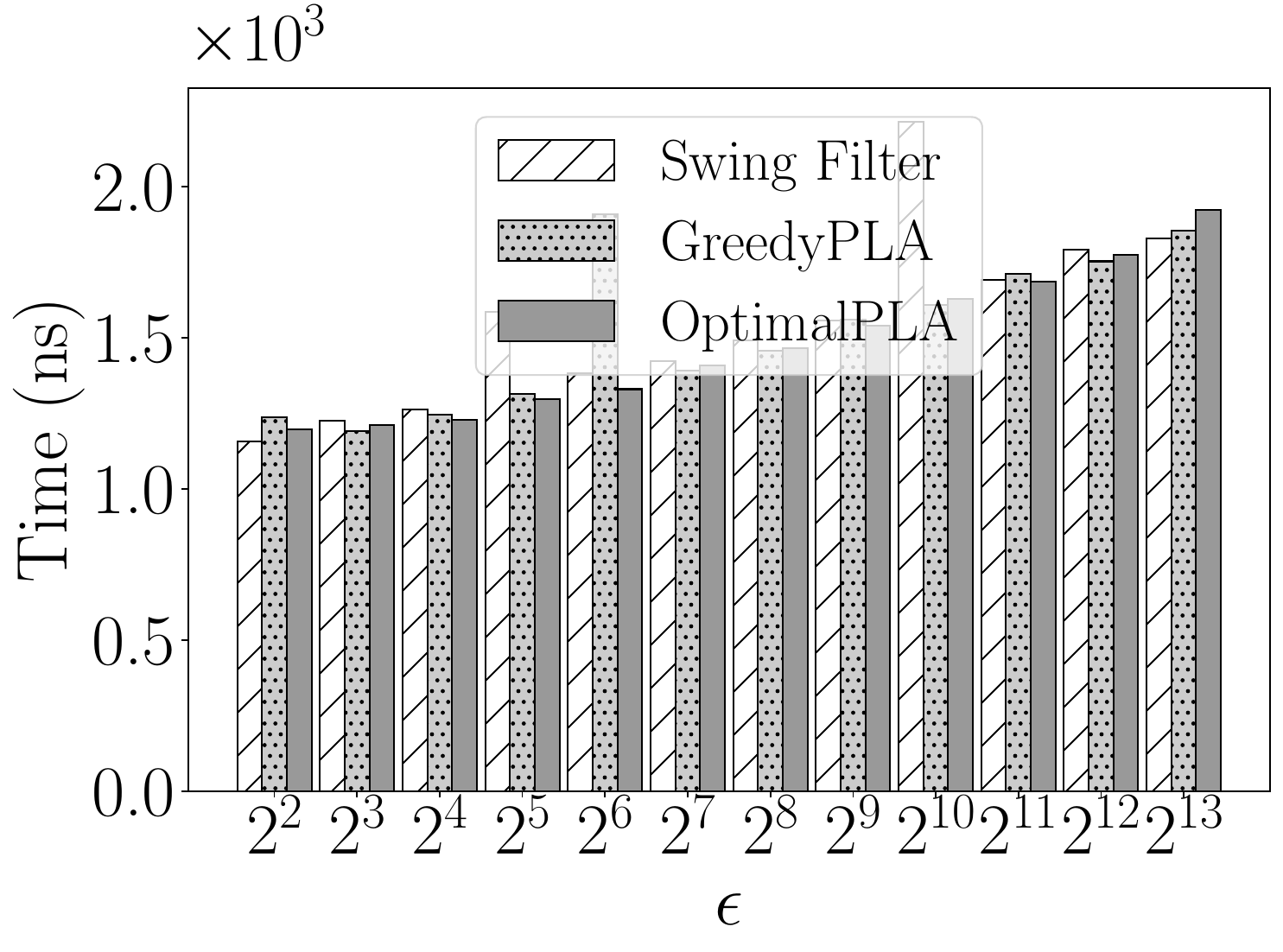}
        \caption{\textsf{normal}: Query Time vs. $\epsilon$}
    \end{subfigure}
    \hspace{-0.5em}
    \begin{subfigure}[t]{0.48\columnwidth}
        \includegraphics[width=\linewidth]{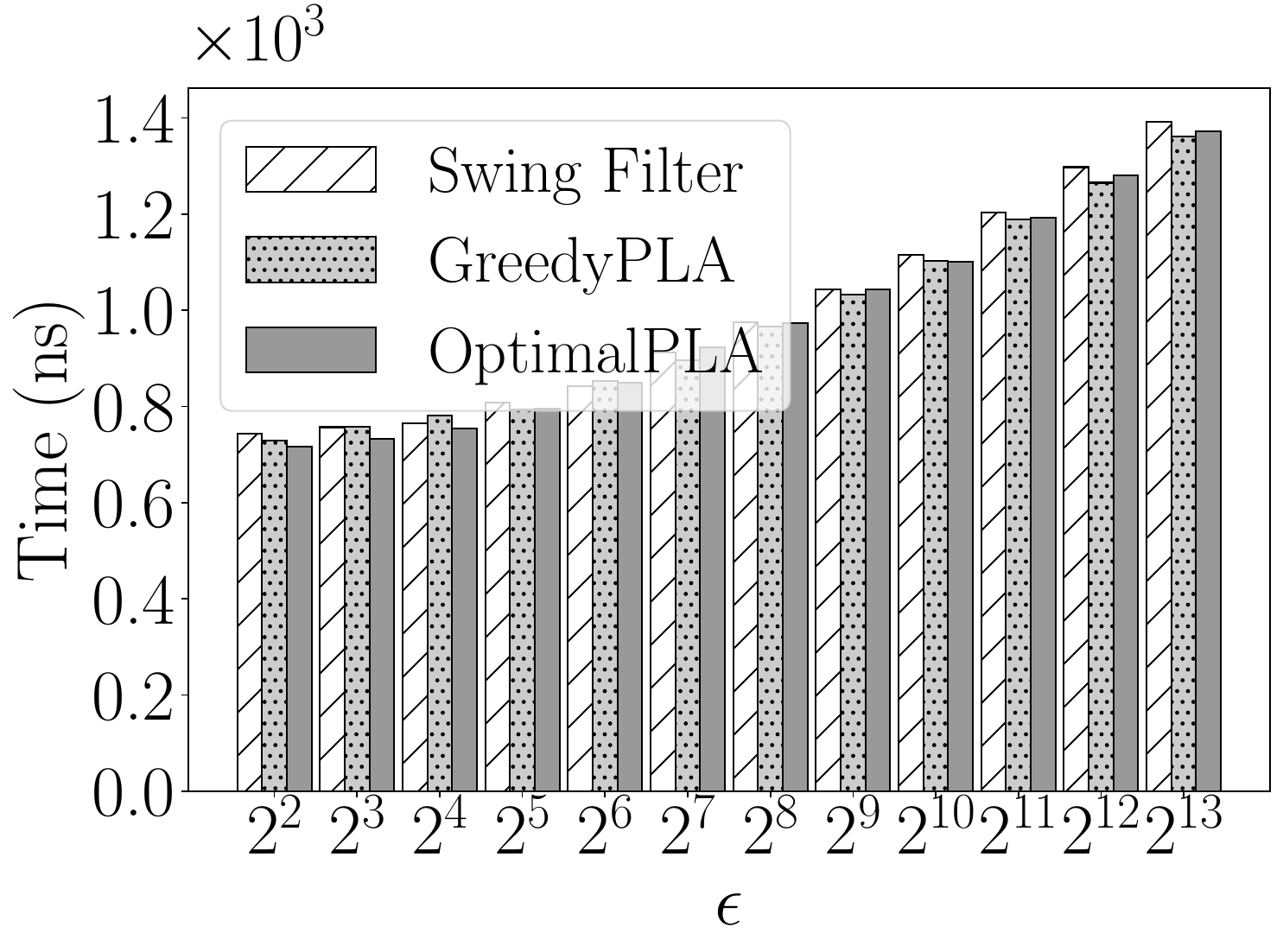}
        \caption{\textsf{lognormal}: Query Time vs. $\epsilon$}
    \end{subfigure}

    \caption{Query processing time of FIT (FITing-Tree + $\epsilon$-PLA) w.r.t.~$\epsilon\in\{2^2, 2^3,\cdots, 2^{13}\}$ for different $\epsilon$-PLA algorithms.}
    \label{fig:fit_qtime}
\end{figure}

\subsection{FIT Evaluation Results (\textbf{RQ3})}\label{subsec:exp_fit}
\subsubsection{FIT Index Size and Construction Time} 
\cref{fig:fit_indexS} reports the index size of FIT constructed using different $\epsilon$-PLA algorithms, with $\epsilon$ ranging from $2^2$ to $2^{13}$. 
The results show that GreedyPLA and SwingFilter incur approximately $\mathbf{1.23}\times$ to $\mathbf{1.38}\times$ more space overhead than OptimalPLA. 
This can be attributed to the number of segments produced by different $\epsilon$-PLA algorithms (see \cref{subsec:exp_standalone}). 
Assuming the fanout of the internal B$^+$-tree is $B$, the total space overhead of FIT is $O(m + m/B)$, where $m$ denotes the number of segments generated by a given $\epsilon$-PLA algorithm. 

We next compare the construction time of FIT using different $\epsilon$-PLA algorithms. 
As shown in \cref{fig:fit_btime}, although GreedyPLA and SwingFilter outperform OptimalPLA in the standalone setting (\cref{subsec:exp_standalone}), they exhibit slightly longer construction times on datasets such as \textsf{fb}, \textsf{osm}, and \textsf{books} when $\epsilon$ is small. 
This is because greedy methods tend to generate more segments, resulting in more leaf nodes and greater overhead in building the internal B$^+$-tree structure, thus increasing total construction time.
As $\epsilon$ increases, the number of segments produced by greedy methods decreases, reducing this overhead and eventually allowing SwingFilter and GreedyPLA to outperform OptimalPLA in FIT. 
Notably, on highly skewed datasets such as \textsf{normal} and \textsf{lognormal}, GreedyPLA consistently achieves up to a $\mathbf{1.45}\times$ speedup over OptimalPLA, due to the significantly fewer segments produced on such datasets. 

\textbf{Takeaways.} 
GreedyPLA and SwingFilter consistently generate more segments across varying $\epsilon$, leading to larger FIT index size and increased building time, particularly when $\epsilon$ is small. 
Therefore, for FIT, to fully leverage the efficiency advantages of greedy $\epsilon$-PLA methods, $\epsilon$ should be set to a relatively large value (e.g., $2^5$ or higher).

\subsubsection{FIT Query Processing Time} 
We now examine the query processing performance of FITing-Trees constructed using different $\epsilon$-PLA algorithms. 
As shown in \cref{fig:fit_qtime}, the query time generally follows a U-shaped trend. 
Initially, greedy methods underperform compared to OptimalPLA due to their higher segment counts, which leads to a deeper internal B$^+$-tree. 
However, as $\epsilon$ increases, the performance gap narrows. 
This is because OptimalPLA, while globally minimizing the number of segments, often produces larger average residuals, resulting in higher last-mile search costs. 
In contrast, GreedyPLA and SwingFilter generate more segments with smaller residuals, leading to more accurate position predictions and ultimately better query efficiency.
Notably, a similar behavior is observed in PGM-Index, and we defer detailed discussions to \cref{subsec:exp_pgm}. 

\begin{figure}[t]
    \centering
    \begin{subfigure}[t]{0.48\columnwidth}
        \includegraphics[width=\linewidth]{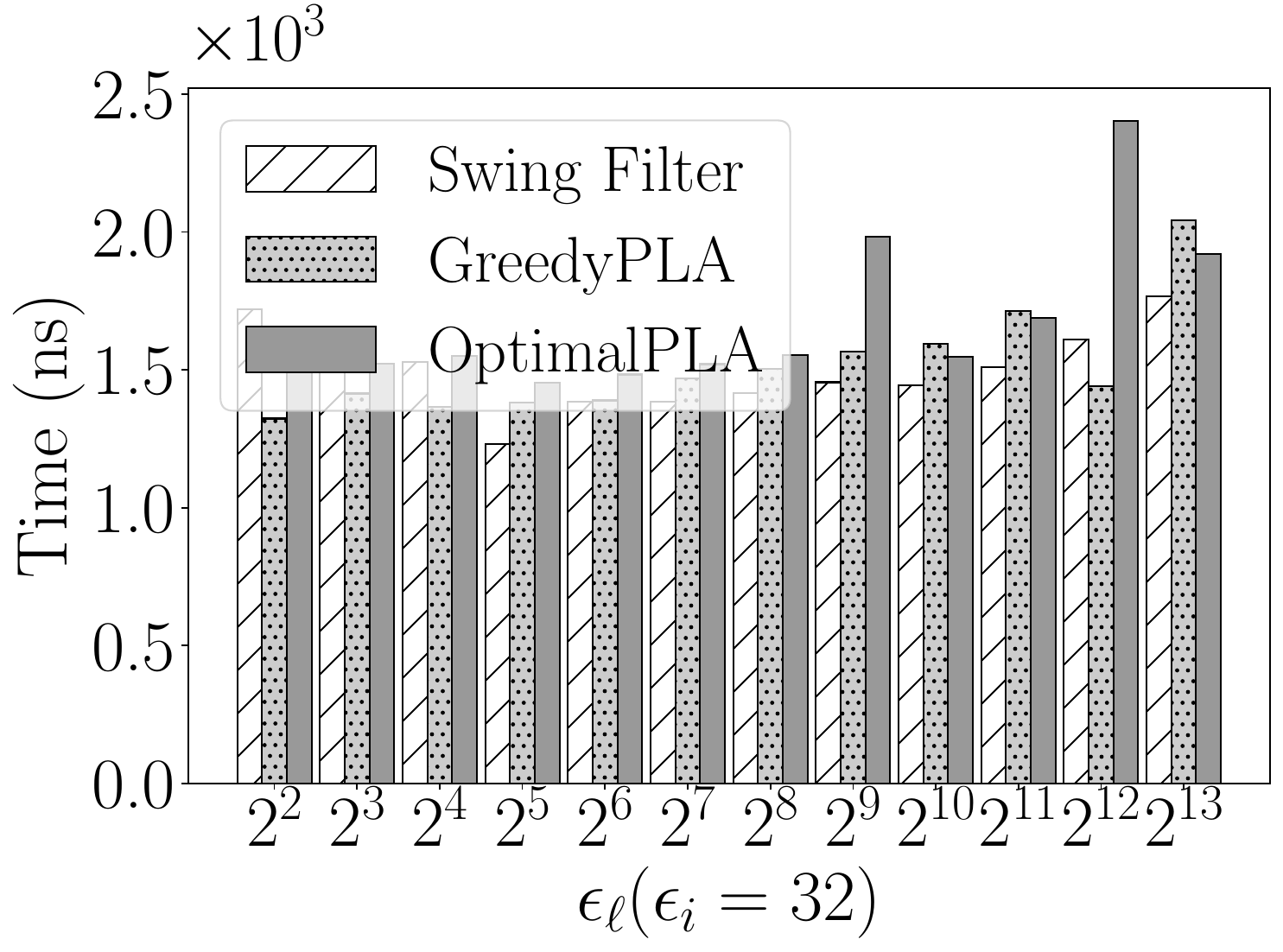}
        \caption{\textsf{fb}: Query Time vs. $\epsilon$}
    \end{subfigure}
    \hfill
    \begin{subfigure}[t]{0.48\columnwidth}
        \includegraphics[width=\linewidth]{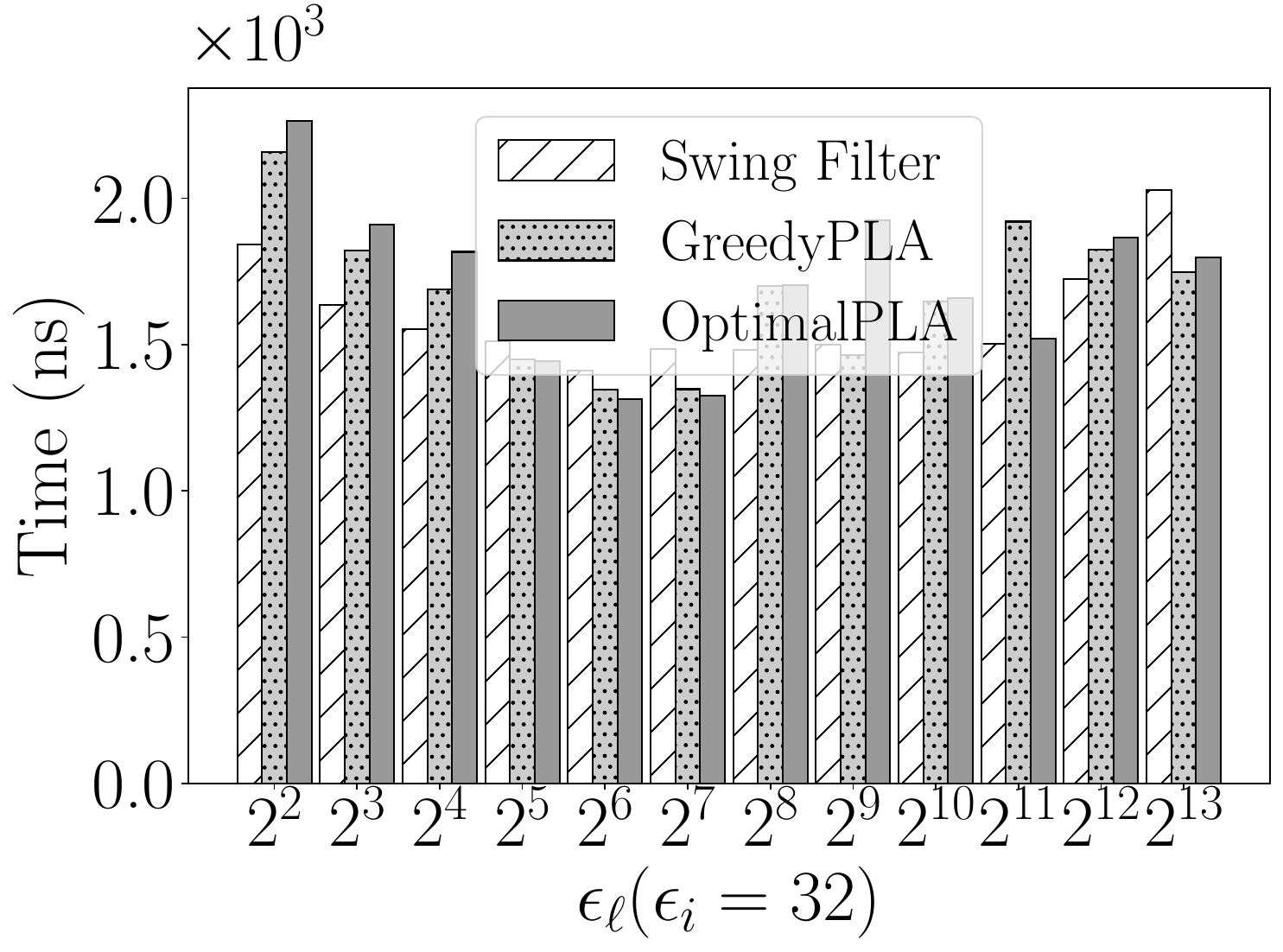}
        \caption{\textsf{books}: Query Time vs. $\epsilon$}
    \end{subfigure}
    
    \begin{subfigure}[t]{0.48\columnwidth}
        \includegraphics[width=\linewidth]{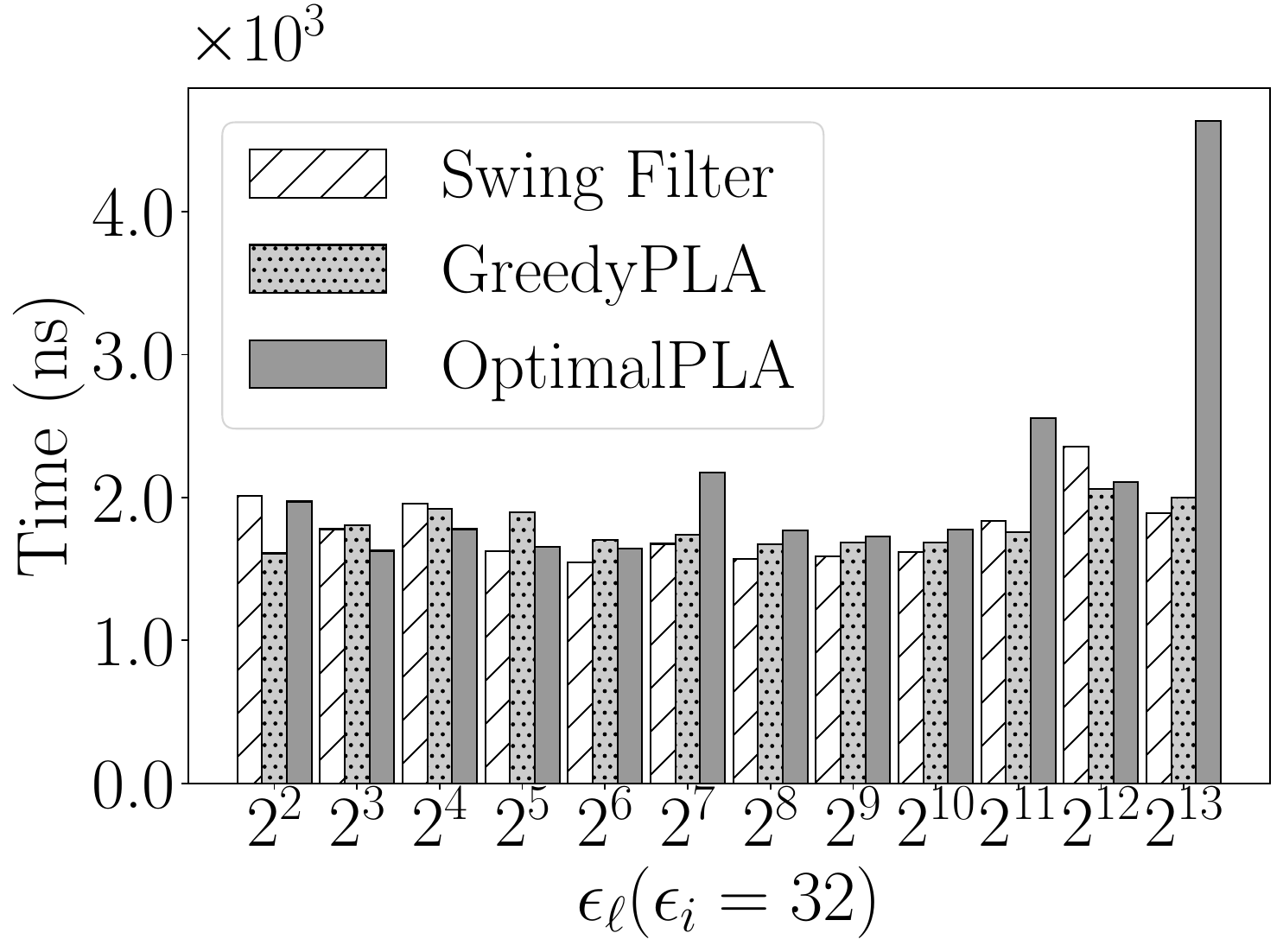}
        \caption{\textsf{osm}: Query Time vs. $\epsilon$}
    \end{subfigure}
    \hfill
        \begin{subfigure}[t]{0.48\columnwidth}
        \includegraphics[width=\linewidth]{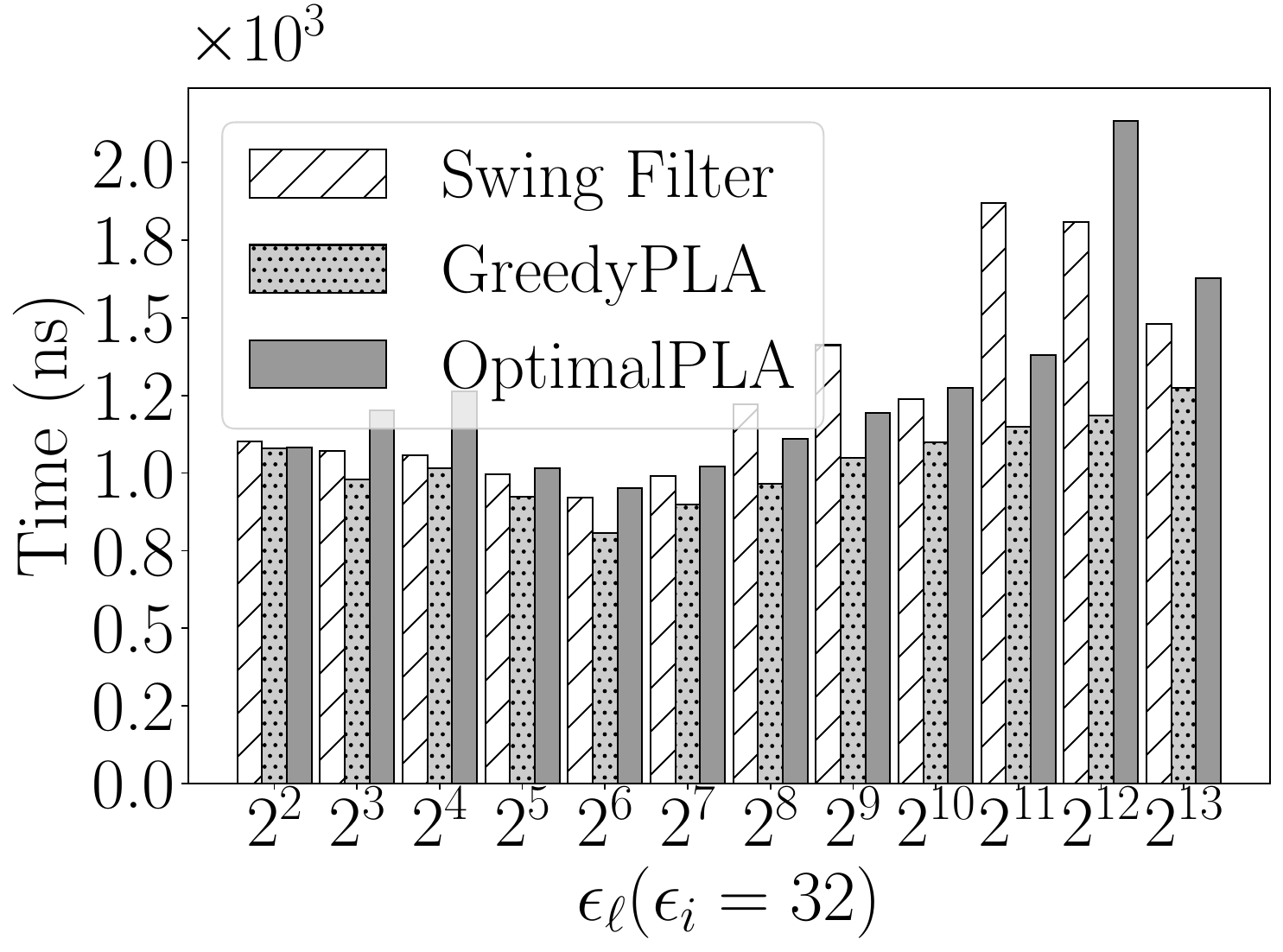}
        \caption{\textsf{uniform}: Query Time vs. $\epsilon$}
    \end{subfigure}

    \begin{subfigure}[t]{0.48\columnwidth}
        \includegraphics[width=\linewidth]{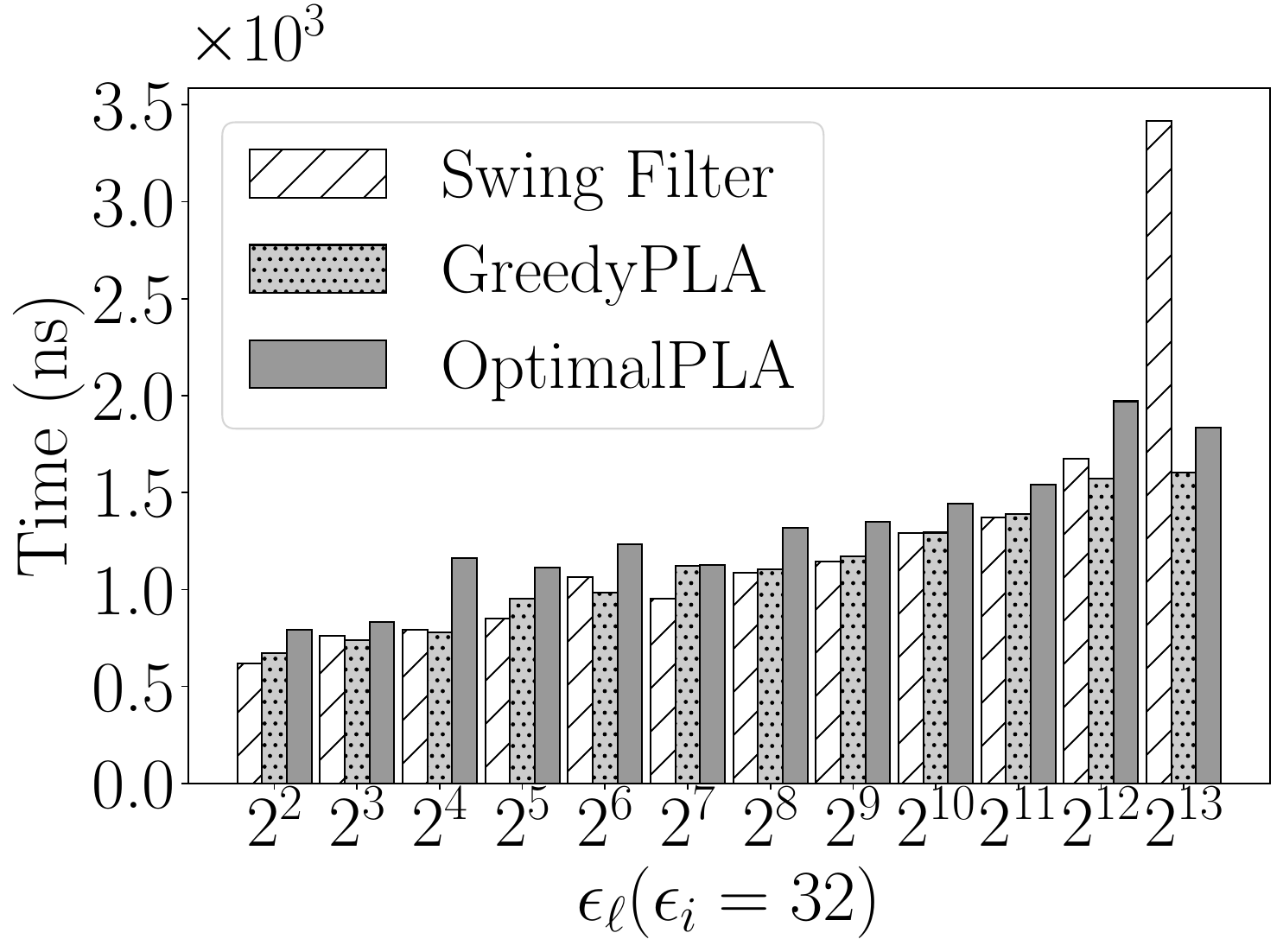}
        \caption{\textsf{normal}: Query Time vs. $\epsilon$}
    \end{subfigure}
    \hfill
        \begin{subfigure}[t]{0.48\columnwidth}
        \includegraphics[width=\linewidth]{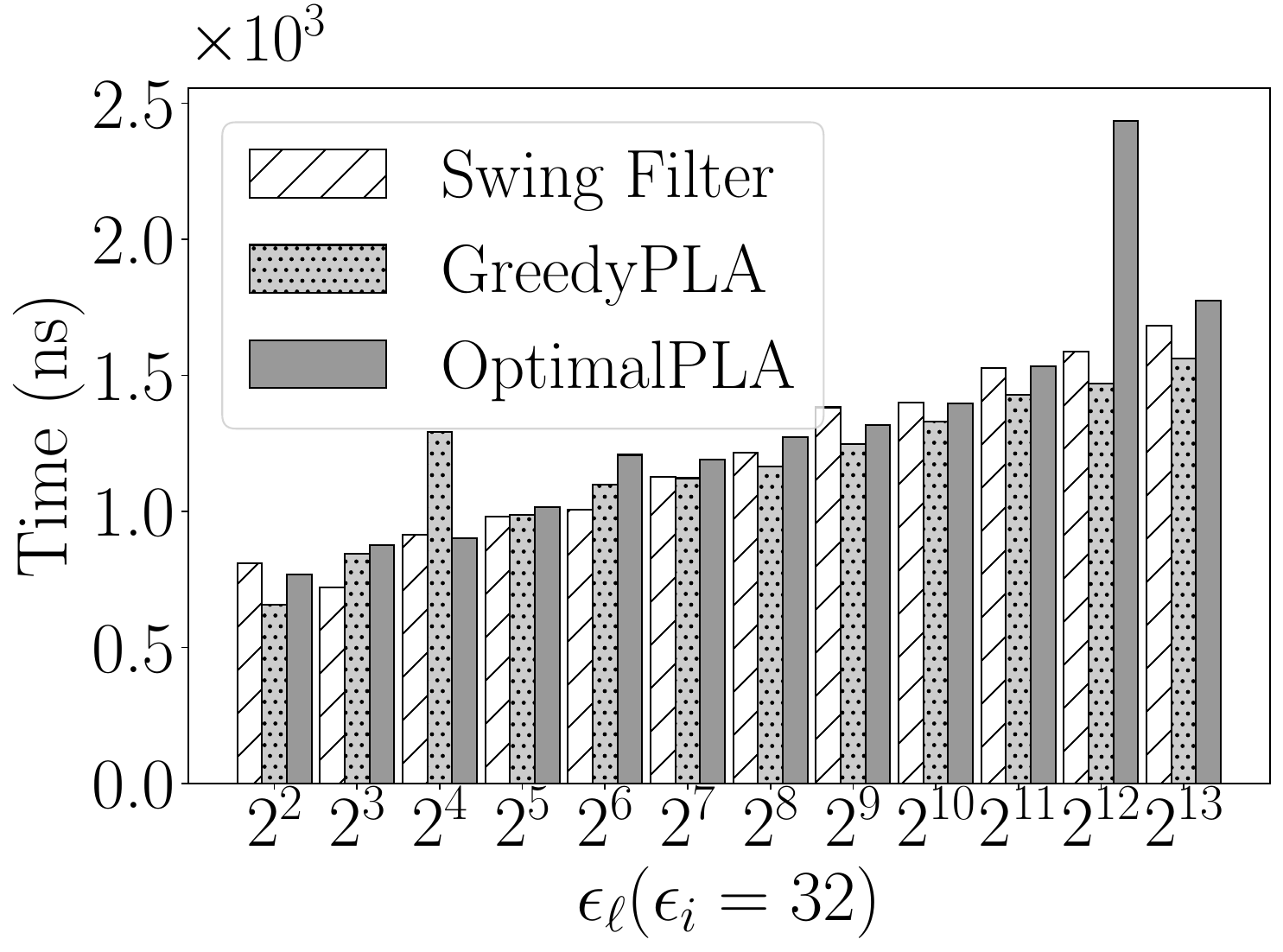}
        \caption{\textsf{lognormal}: Query Time vs. $\epsilon$}
    \end{subfigure}
    
    \caption{Query processing time of PGM (PGM-Index + $\epsilon$-PLA) w.r.t.~$\epsilon\in\{2^2, 2^3,\cdots, 2^{13}\}$ for different $\epsilon$-PLA fitting algorithms.}
     \label{fig:pgm_qtime}
\end{figure}

\begin{figure}[t]
    \centering
    \begin{subfigure}[t]{0.48\columnwidth}
        \includegraphics[width=\linewidth]{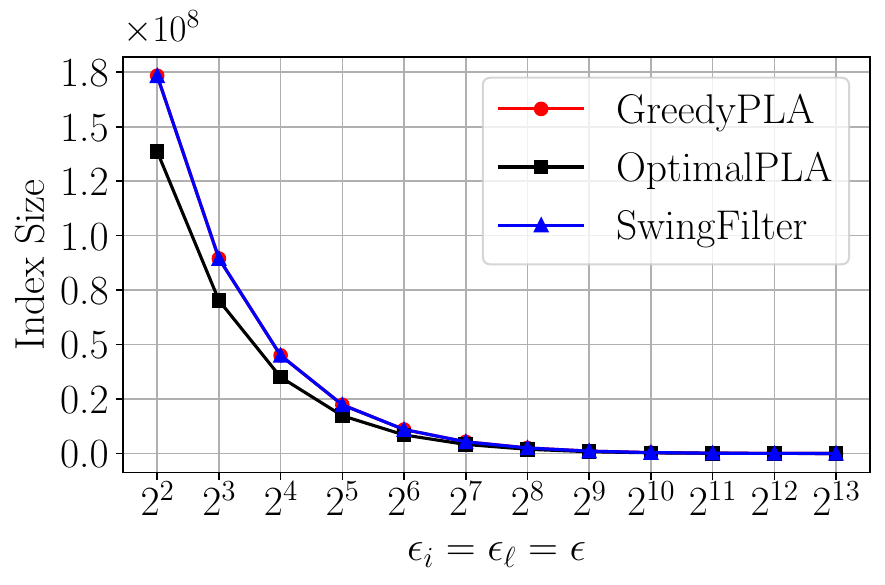}
        \caption{\textsf{fb}: Index Size vs. $\epsilon$}
    \end{subfigure}
    \hfill  
    \begin{subfigure}[t]{0.48\columnwidth}
        \includegraphics[width=\linewidth]{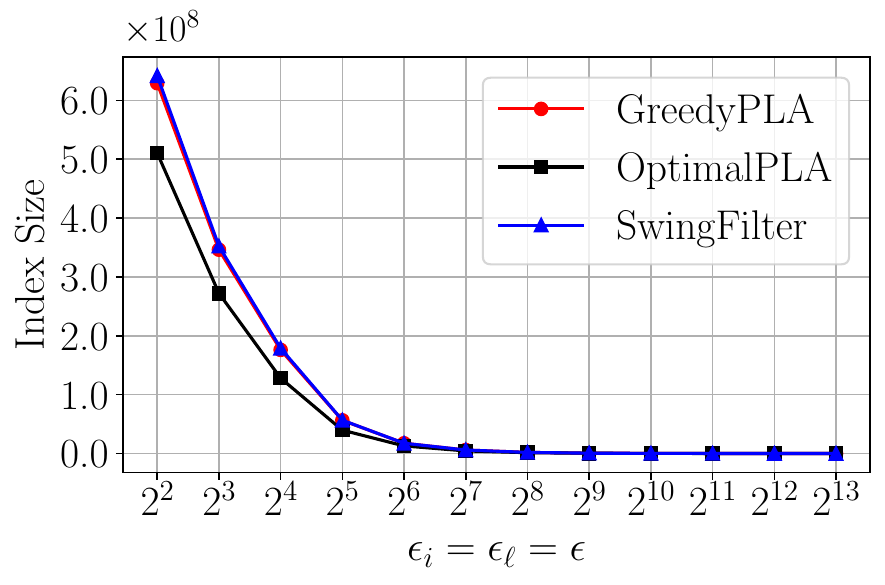}
        \caption{\textsf{books}: Index Size vs. $\epsilon$}
    \end{subfigure}
    
    \begin{subfigure}[t]{0.48\columnwidth}
        \includegraphics[width=\linewidth]{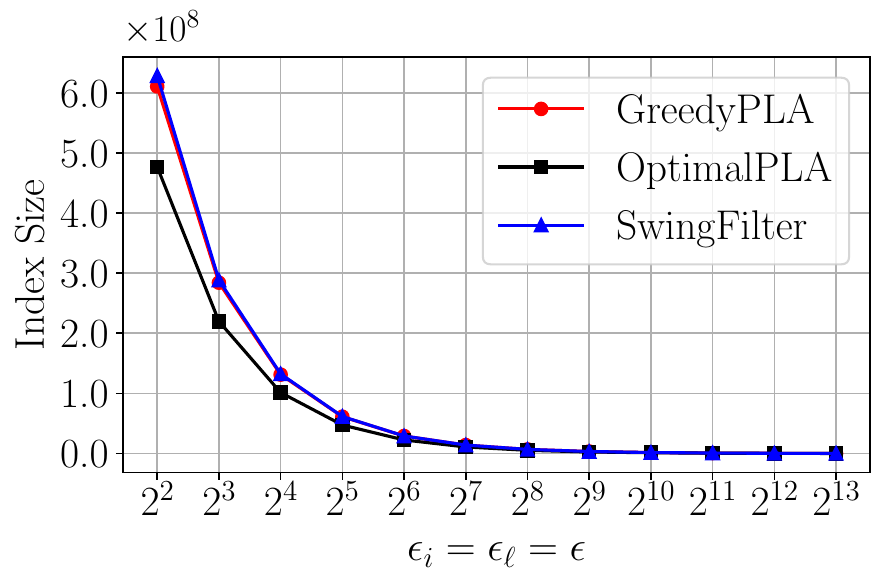}
        \caption{\textsf{osm}: Index Size vs. $\epsilon$}
    \end{subfigure}
    \hfill
    \begin{subfigure}[t]{0.48\columnwidth}
        \includegraphics[width=\linewidth]{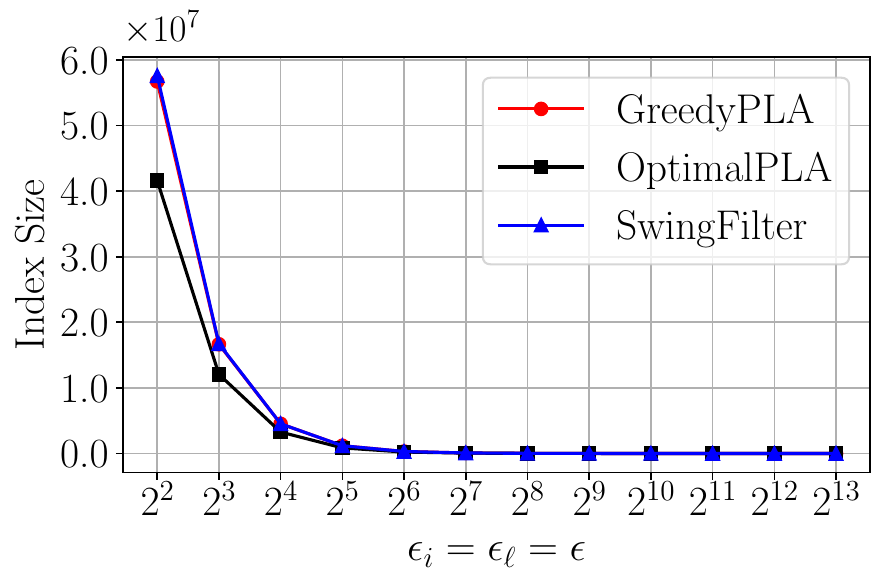}
        \caption{\textsf{uniform}: Index Size vs. $\epsilon$}
    \end{subfigure}

    \begin{subfigure}[t]{0.48\columnwidth}
        \includegraphics[width=\linewidth]{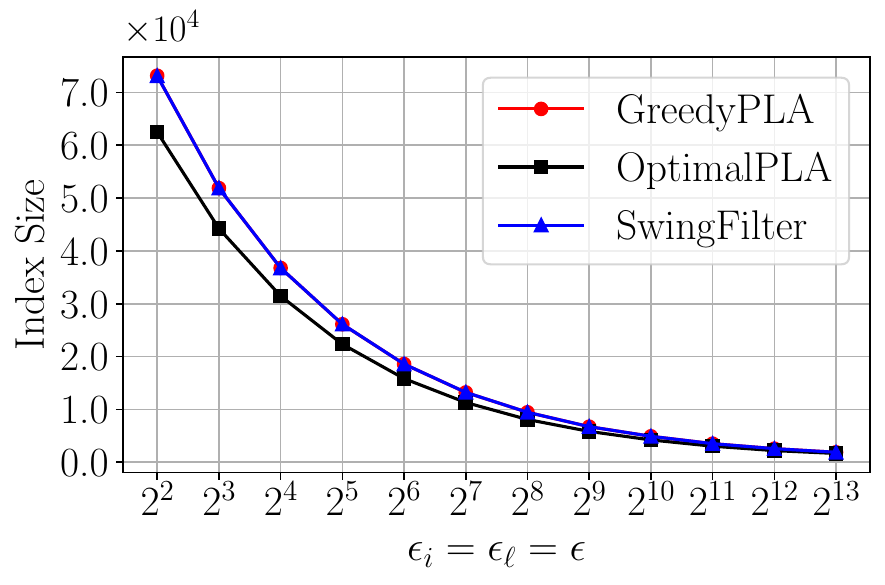}
        \caption{\textsf{normal}: Index Size vs. $\epsilon$}
    \end{subfigure}
    \hfill
    \begin{subfigure}[t]{0.48\columnwidth}
        \includegraphics[width=\linewidth]{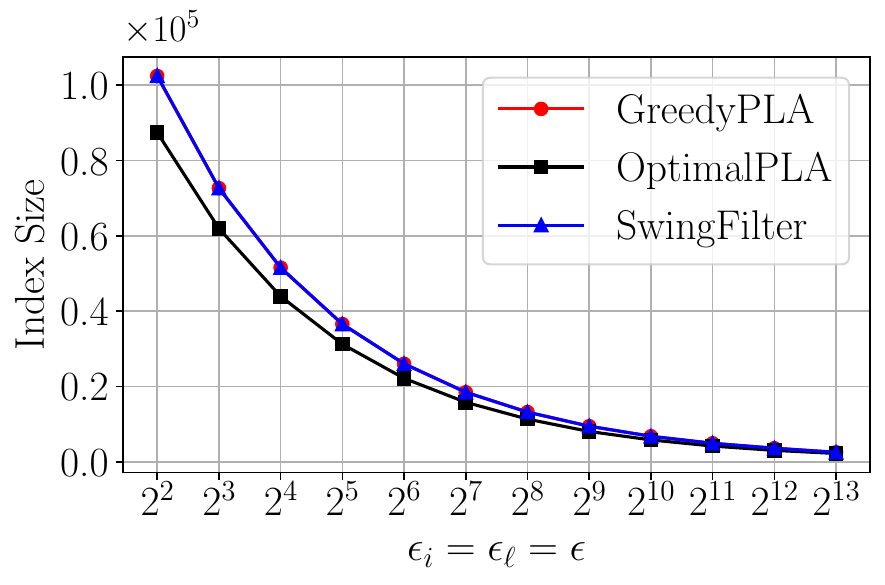}
        \caption{\textsf{lognormal}: Index Size vs. $\epsilon$}
    \end{subfigure}
    \caption{Index size time of PGM (PGM-Index + $\epsilon$-PLA) w.r.t.~$\epsilon\in\{2^2, 2^3,\cdots, 2^{13}\}$ for different $\epsilon$-PLA fitting algorithms.}
    \label{fig:pgm_indexS}
\end{figure}

\begin{figure}[t]
    \centering
    \begin{subfigure}[t]{0.48\columnwidth}
        \includegraphics[width=\linewidth]{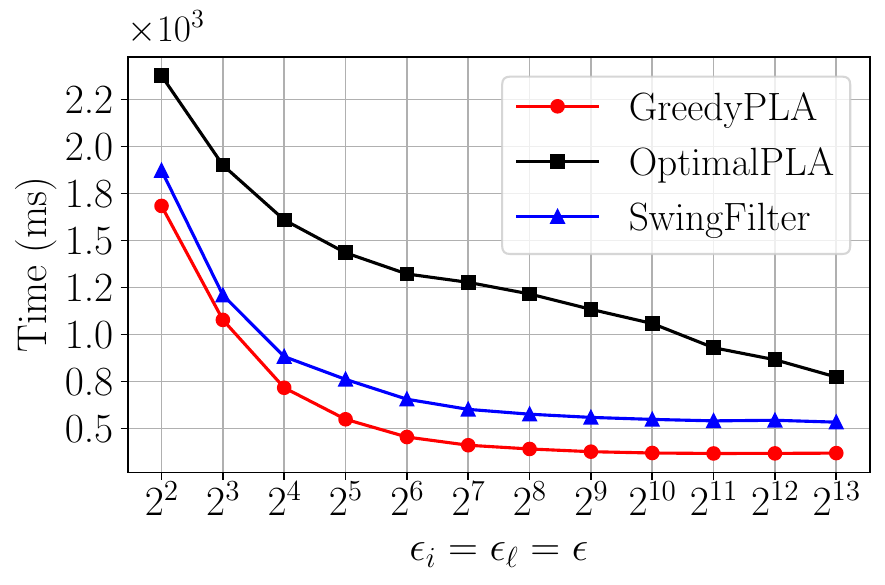}
        \caption{\textsf{fb}: Build Time vs. $\epsilon$}
    \end{subfigure}
    \hfill
    \begin{subfigure}[t]{0.48\columnwidth}
        \includegraphics[width=\linewidth]{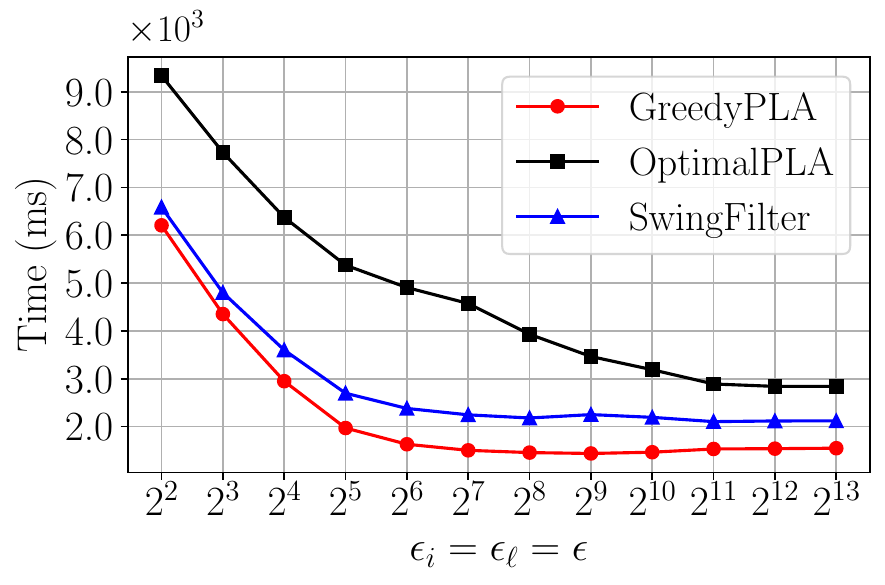}
        \caption{\textsf{books}: Build Time vs. $\epsilon$}
    \end{subfigure}

    \begin{subfigure}[t]{0.48\columnwidth}
        \includegraphics[width=\linewidth]{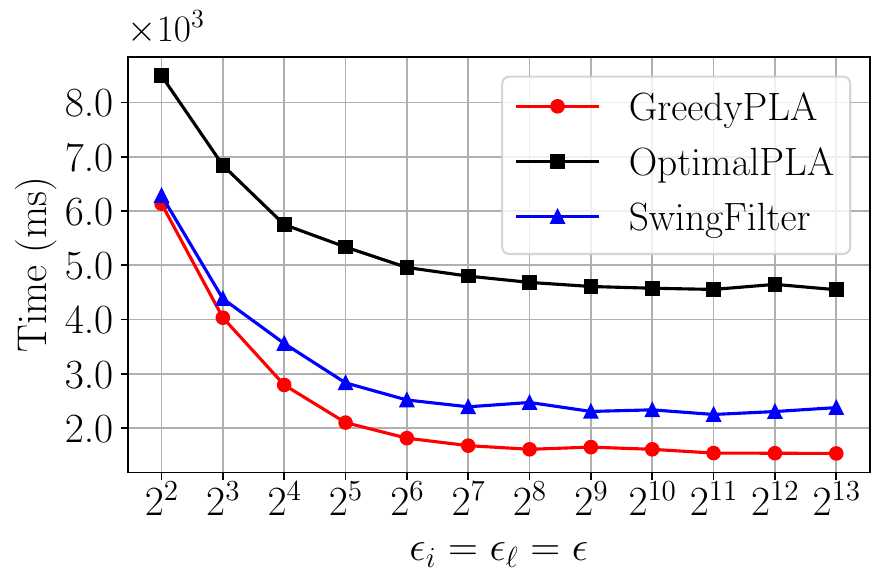}
        \caption{\textsf{osm}: Build Time vs. $\epsilon$}
    \end{subfigure}
    \hfill
    \begin{subfigure}[t]{0.48\columnwidth}
        \includegraphics[width=\linewidth]{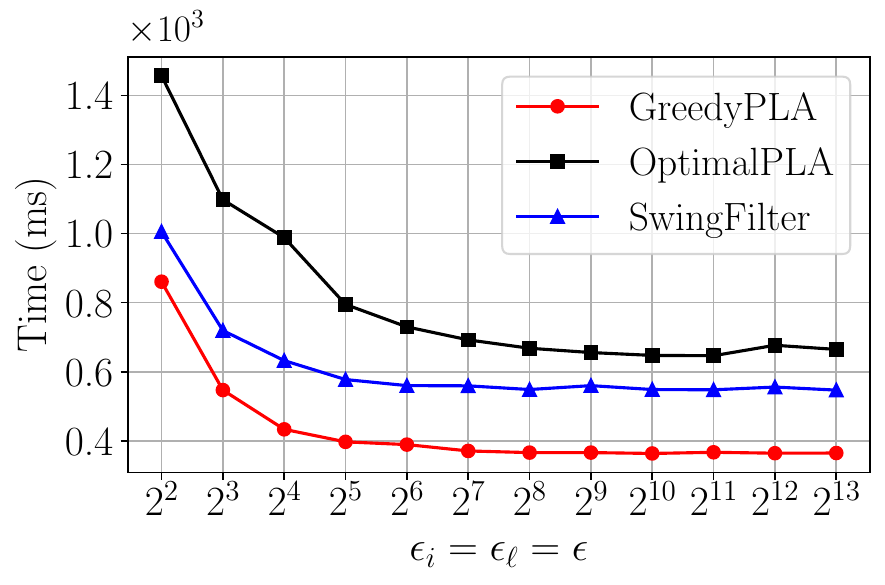}
        \caption{\textsf{uniform}: Build Time vs. $\epsilon$}
    \end{subfigure}

    \begin{subfigure}[t]{0.48\columnwidth}
        \includegraphics[width=\linewidth]{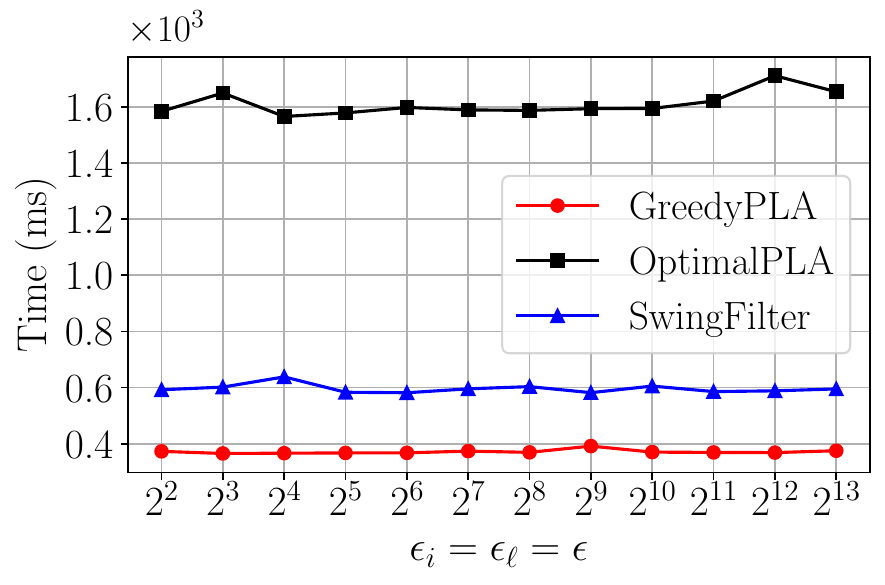}
        \caption{\textsf{normal}: Build Time vs. $\epsilon$}
    \end{subfigure}
    \hfill
    \begin{subfigure}[t]{0.48\columnwidth}
        \includegraphics[width=\linewidth]{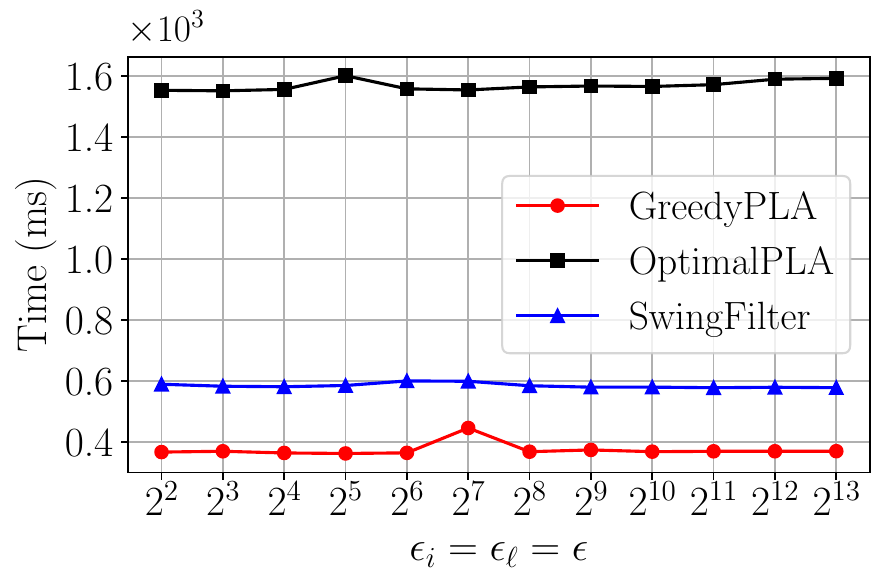}
        \caption{\textsf{lognormal}: Build Time vs. $\epsilon$}
    \end{subfigure}
    \caption{Construction time of PGM (PGM-Index + $\epsilon$-PLA) w.r.t.~$\epsilon\in\{2^2, 2^3,\cdots, 2^{13}\}$ for different $\epsilon$-PLA fitting algorithms.}
    \label{fig:pgm_btime}
\end{figure}

\begin{table}[t]
\centering
\caption{Maximum average residual over all the internal layers of PGM-Index when searching for a set of query keys. 
We fix $\epsilon_i$ at 32 and vary $\epsilon_\ell$ on \textsf{fb} dataset. 
Notably, the last-mile error $\epsilon_\ell$ doesn't affect the internal searching time.} 
\label{tab:residual_compare}
\begin{tabular}{c|ccc}
\toprule
\textbf{$\epsilon_\ell$} & \textbf{OptimalPLA} & \textbf{GreedyPLA} & \textbf{SwingFilter} \\
\midrule
$2^2$   & 14.1936 & 13.1399 & 13.2728 \\
$2^3$   & 13.3582 & 12.4875 & 12.7311 \\
$2^4$  & 12.9545 & 12.1711 & 12.1628 \\
$2^5$  & 14.3894 & 11.8181 & 11.6613 \\
$2^6$  & 14.0896 & 11.4472 & 11.4814 \\
$2^7$  & 15.5535 & 11.2243 & 11.0557 \\
$2^8$  & 12.7284 & 10.5285 & 10.5566 \\
$2^9$  & 15.0138 & 9.8464  & 10.5327 \\
\bottomrule
\end{tabular}
\end{table}

\subsection{PGM Evaluation Results (\textbf{RQ3})}\label{subsec:exp_pgm}

Unlike FITing-Tree, which applies a single error parameter for $\epsilon$-PLA, PGM-Index employs two independent error parameters: one for the internal $\epsilon$-PLA model(s) ($\epsilon_i$) and another for the last-level $\epsilon$-PLA model ($\epsilon_\ell$). 
For internal search settings, we set $2^5$ as a threshold, meaning that when $\epsilon_i\leq2^5$ we use linear search and switch to binary search otherwise.

\subsubsection{PGM Index Size and Construction Time}
\cref{fig:pgm_indexS} reports the index size of PGM-Index with different $\epsilon$-PLA algorithms when varying the error bound $\epsilon$. 
Note that we set the internal error bound $\epsilon_i$ and last-mile error bound $\epsilon_\ell$ to be the same as $\epsilon$. 
The results clearly show the superiority of OptimalPLA as it guarantees the minimum number of segments generated for each layer. 
The index size of the PGM-Index constructed by SwingFilter and GreedyPLA is $\mathbf{1.24}\times$ to $\mathbf{1.38}\times$ larger than that achieved with OptimalPLA. 
These findings align with the results in \cref{subsec:exp_standalone}. 

We then compare the construction time of the PGM-Index. 
In contrast to their performance on FITing-Tree (see \cref{fig:fit_indexS}, \cref{fig:fit_btime}, \cref{fig:fit_qtime}), greedy algorithms achieve more significant improvements in index construction time when applied to PGM-Index. 
For instance, in \cref{fig:pgm_btime}, on the \textsf{books} dataset, GreedyPLA provides up to $\mathbf{3}\times$ speedup over OptimalPLA. 
The reason is that, unlike FITing-Tree, which invokes the $\epsilon$-PLA only once for leaf-level index construction, PGM-Index recursively applies $\epsilon$-PLA throughout all levels of the index. 
As a result, the efficiency advantage of greedy algorithms becomes more pronounced in the context of PGM-Index.

\textbf{Takeaways.}
Greedy $\epsilon$-PLA algorithms achieve significantly faster index construction in the PGM-Index compared to the FITing-Tree, due to PGM-Index's greater dependence on the efficiency of the $\epsilon$-PLA algorithm. 
Notably, this improvement becomes more pronounced on skewed datasets. 

\subsubsection{PGM Query Processing Time}
We now focus on the query efficiency of PGM-Index based on different $\epsilon$-PLA fitting methods. 
When searching for a key, we fix $\epsilon_i =2^5$ for simplicity and uniformity. 
As shown in \cref{fig:pgm_qtime}, there is a U-shaped trend similar to \cref{subsec:exp_fit}.
On each dataset, the $\epsilon$-PLA methods show similar patterns, which will be discussed in \cref{subsec:exp_parameter}. 
However, OptimalPLA underperforms slightly over GreedyPLA and SwingFilter when $\epsilon_\ell = \epsilon_i = 2^5$, as shown in \cref{tab:pgm_statiscs}. 
This might be a counterintuitive phenomenon because all the $\epsilon$-PLA methods share almost the same height and the exact searching range (i.e., same $\epsilon_i$ and $\epsilon_\ell$). 
To explain this, we calculate the average residual between the predicted indexes and the real indexes within each layer. 
We fix $\epsilon_i = 2^5$ and report the largest residual among all the layers with different $\epsilon_\ell$ in \cref{tab:residual_compare}. 
Notably, the average residual of OptimalPLA is often larger than greedy methods over all $\epsilon_\ell$, leading to a higher query latency inside internal segments. 

\textbf{Takeaways.}
While OptimalPLA minimizes the number of segments, it often leads to larger intra-segment residuals compared to greedy $\epsilon$-PLA algorithms, highlighting a trade-off between compression efficacy and prediction precision that directly impacts the search latency of PGM-Index. 

\begin{table}[h]
    \centering
    \caption{Index Size, build time, and query processing time comparison for PGM and FIT w.r.t.~$\epsilon\in\{2^2, 2^6\}$ across different $\epsilon$-PLA algorithms (OPT: OptimalPLA, GRY: GreedyPLA, SWF: SwingFilter) on \textsf{books}. For PGM, $\epsilon_i$ is fixed to $2^5$ and we vary $\epsilon_\ell$. }
    \label{tab:pgm_fit_comparison}
    \begin{adjustbox}{width=\columnwidth,center}
    \begin{tabular}{c|c|c|c|c|c|c|c}
    \toprule
        \multirow{2}{*}{} & \multirow{2}{*}{$\epsilon$} & \multicolumn{3}{c|}{\textbf{PGM}}& \multicolumn{3}{c}{\textbf{FIT}} \\
         & & OPT & GRY & SWF & OPT & GRY & SWF \\
        \midrule
        \multirow{5}{1.4cm}{\centering Index Size \\ (MB)} & $2^2$ & 480.85 & 589.01 & 600.47 & 2403.45 & 2943.43 & 3000.66 \\
                    & $2^3$ & 257.31 & 325.46 & 330.99 & 1286.31 & 1626.74 & 1654.36 \\
                    & $2^4$ & 121.22 & 165.80 & 168.34 & 606.04  & 828.78  & 841.49  \\
                    & $2^5$ & 37.61  & 53.25  & 53.41  & 188.01  & 266.22  & 267.03  \\
                    & $2^6$ & 12.17  & 16.42  & 16.48  & 60.82   & 82.10   & 82.37 \\     
        \hline
        \multirow{5}{1.4cm}{\centering Build Time \\ (s)}  & $2^2$ & 9.36 & 6.25 & 6.52 & 39.16 & 43.10 & 45.41 \\
                    & $2^3$ & 7.66 & 4.31 & 4.79 & 21.77 & 21.88 & 25.69 \\
                    & $2^4$ & 6.30 & 2.77 & 3.53 & 12.66 & 13.60 & 13.42 \\
                    & $2^5$ & 5.38 & 1.98 & 2.71 & 7.47  & 7.35  & 7.62 \\
                    & $2^6$ & 4.91 & 1.61 & 2.35 & 5.89 & 5.06 & 5.46 \\ 
        \hline
        \multirow{5}{1.4cm}{\centering Query Time \\ (\textmu s)}  & $2^2$ & 2.50 & 2.80 & 2.17 & 2.93 & 4.75 & 5.89 \\
                    & $2^3$ & 2.27 & 2.48 & 2.46 & 3.29 & 3.58 & 5.05 \\
                    & $2^4$ & 2.46 & 1.79 & 2.38 & 2.22 & 2.23 & 2.29 \\
                    & $2^5$ & 1.99 & 1.98 & 2.27 & 1.97 & 2.13 & 2.64 \\
         & $2^6$ & 1.77 & 1.55 & 1.96 & 1.73 & 1.79 & 1.81 \\   
        \bottomrule
    \end{tabular}
    \end{adjustbox}
\end{table}

\begin{table*}[ht]
\centering
\caption{Internal details of the PGM-Index using different $\epsilon$-PLA methods on \textsf{fb}, \textsf{books}, and \textsf{osm} datasets ($\epsilon_i = \epsilon_\ell = \epsilon$). The numbers of segments in the last and intermediate layers, listed from bottom to top, are reported in Segments(Last) and Segments(Others), respectively.}
\small
\renewcommand{\arraystretch}{1.2}
\begin{tabular}{c|l|c|c|c|c|c|c}
\toprule
\textbf{Dataset} & \textbf{Method} & $\epsilon$ & \textbf{Height} & \textbf{Segments (Last)} & \textbf{Segments (Others)} & \textbf{Memory (MB)} & \textbf{Query Time (ns)} \\
\midrule
\multirow{12}{*}{fb} & \multirow{4}{*}{OptimalPLA} & $2^2$ & 5 & 8338506  & $\left\{317190,3603,12,2\right\}$ & 132.13 & 1150.21 \\
                     &                             & $2^3$ & 4 & 4235686  & $\left\{69573,92,2\right\}$       & 65.69  & 1472.59 \\
                     &                             & $2^4$ & 4 & 2120493  & $\left\{9324,3,2\right\}$         & 32.50  & 1843.41  \\
                     &                             & $2^5$ & 4 & 1055318  & $\left\{790,4,2\right\}$          & 16.12  & 1963.18 \\
\cline{2-8}
                     & \multirow{4}{*}{GreedyPLA}  & $2^2$ & 6 & 10333876 & $\left\{492381,10322,53,4,2\right\}$ & 165.35 & 1487.99 \\
                     &                             & $2^3$ & 5 & 5341466  & $\left\{119665,293,4,2\right\}$   & 83.33  & 1620.26 \\
                     &                             & $2^4$ & 4 & 2695175  & $\left\{18663,11,2\right\}$       & 41.42  & 1565.07 \\
                     &                             & $2^5$ & 4 & 1347023  & $\left\{1696,4,2\right\}$         & 20.58  & 2047.63 \\
\cline{2-8}
                     & \multirow{4}{*}{SwingFilter}& $2^2$ & 6 & 10628797 & $\left\{521399,11554,64,4,2\right\}$ & 170.27 & 1332.43 \\
                     &                             & $2^3$ & 5 & 5431833  & $\left\{124224,306,4,2\right\}$   & 84.78  & 1525.67 \\
                     &                             & $2^4$ & 4 & 2721025  & $\left\{19092,11,2\right\}$       & 41.82  & 1533.47 \\
                     &                             & $2^5$ & 4 & 1354207  & $\left\{1721,4,2\right\}$         & 20.69  & 1891.13 \\
\midrule
\multirow{12}{*}{osm} & \multirow{4}{*}{OptimalPLA} & $2^2$ & 6 & 29028283 & $\left\{745706,19720,586,23,2\right\}$ & 454.61 & 1371.74 \\
                      &                             & $2^3$ & 5 & 13385178 & $\left\{169542,2460,45,2\right\}$      & 206.86 & 1612.50 \\
                      &                             & $2^4$ & 5 & 6175398  & $\left\{39681,315,4,2\right\}$         & 94.84  & 1968.69 \\
                      &                             & $2^5$ & 4 & 2887532  & $\left\{9726,43,2\right\}$             & 44.20  & 2980.33 \\
\cline{2-8}
                      & \multirow{4}{*}{GreedyPLA}  & $2^2$ & 7 & 36903916 & $\left\{1203364,39339,1441,58,3,2\right\}$ & 582.17 & 1745.94 \\
                      &                             & $2^3$ & 6 & 17180689 & $\left\{276360,4872,97,5,3\right\}$        & 266.42 & 1567.37 \\
                      &                             & $2^4$ & 5 & 7935265  & $\left\{64184,640,11,2\right\}$            & 122.12 & 2139.56 \\
                      &                             & $2^5$ & 4 & 3702941  & $\left\{15629,81,2\right\}$                & 56.74  & 2186.38 \\
\cline{2-8}
                      & \multirow{4}{*}{SwingFilter}& $2^2$ & 7 & 37950515 & $\left\{1270704,42449,1576,66,4,2\right\}$ & 598.98 & 2190.30 \\
                      &                             & $2^3$ & 6 & 17434082 & $\left\{284225,5073,103,5,2\right\}$       & 270.33 & 1985.02 \\
                      &                             & $2^4$ & 5 & 7995084  & $\left\{65118,648,12,2\right\}$            & 123.02 & 2245.46 \\
                      &                             & $2^5$ & 4 & 3718071  & $\left\{15741,82,2\right\}$                & 56.98  & 2018.77 \\
\midrule
\multirow{12}{*}{books} & \multirow{4}{*}{OptimalPLA} & $2^2$ & 5 & 31502388 & $\left\{408728,1242,10,2\right\}$ & 487.06 & 1328.93 \\
                        &                             & $2^3$ & 4 & 16859909 & $\left\{46531,60,2\right\}$       & 257.96 & 1992.28 \\
                        &                             & $2^4$ & 4 & 7943412  & $\left\{4098,7,2\right\}$         & 121.27 & 1933.76 \\
                        &                             & $2^5$ & 3 & 2464239  & $\left\{278,2\right\}$            & 37.61  & 1991.28 \\
\cline{2-8}
                        & \multirow{4}{*}{GreedyPLA}  & $2^2$ & 5 & 38580038 & $\left\{749695,4072,29,2\right\}$ & 600.17 & 1380.49 \\
                        &                             & $2^3$ & 5 & 21321956 & $\left\{99384,172,4,2\right\}$    & 326.93 & 1858.17 \\
                        &                             & $2^4$ & 4 & 10862911 & $\left\{10603,13,2\right\}$       & 165.90 & 1523.33 \\
                        &                             & $2^5$ & 4 & 3489298  & $\left\{638,3,2\right\}$          & 53.25  & 1977.94 \\
\cline{2-8}
                        & \multirow{4}{*}{SwingFilter}& $2^2$ & 5 & 39330211 & $\left\{791500,4463,27,2\right\}$ & 612.17 & 1845.98 \\
                        &                             & $2^3$ & 5 & 21683940 & $\left\{791500,4463,27,2\right\}$ & 332.38 & 1736.26 \\
                        &                             & $2^4$ & 4 & 11029494 & $\left\{10992,12,2\right\}$       & 168.45 & 1743.90 \\
                        &                             & $2^5$ & 4 & 3499927  & $\left\{640,3,2\right\}$          & 53.40  & 2270.63 \\
\bottomrule
\end{tabular}
\label{tab:pgm_statiscs}
\end{table*}

\subsection{Effects of Hyperparameters (\textbf{RQ4})}\label{subsec:exp_parameter}
\begin{figure*}[h]
    \centering

    \begin{subfigure}[t]{\textwidth}
        \centering
        \hfill
        \begin{subfigure}[t]{0.32\textwidth}
            \includegraphics[width=\linewidth]{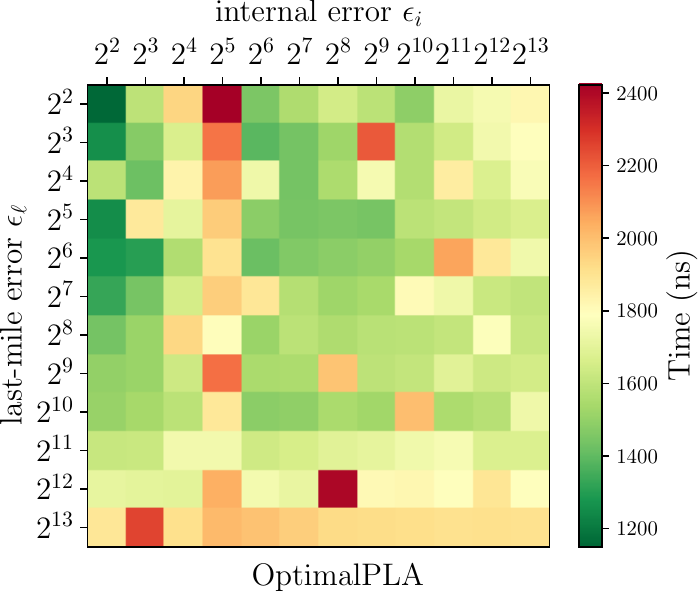}
        \end{subfigure}
        \hfill
        \begin{subfigure}[t]{0.32\textwidth}
            \includegraphics[width=\linewidth]{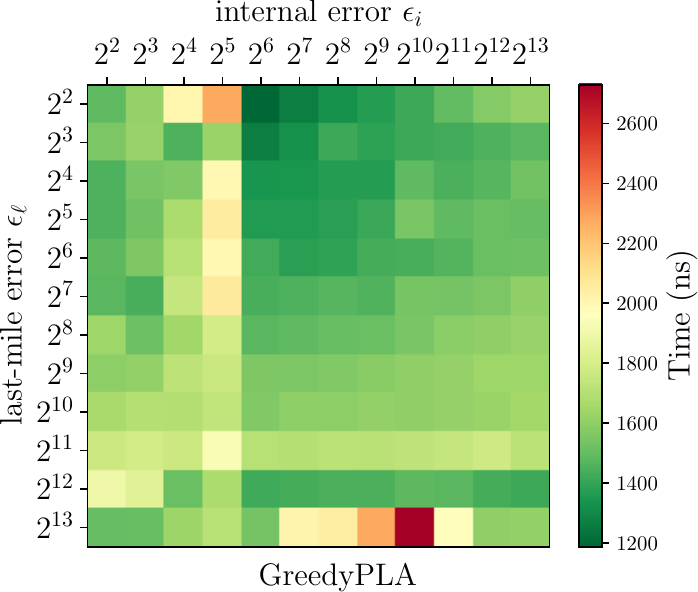}
        \end{subfigure}
        \hfill
        \begin{subfigure}[t]{0.32\textwidth}
            \includegraphics[width=\linewidth]{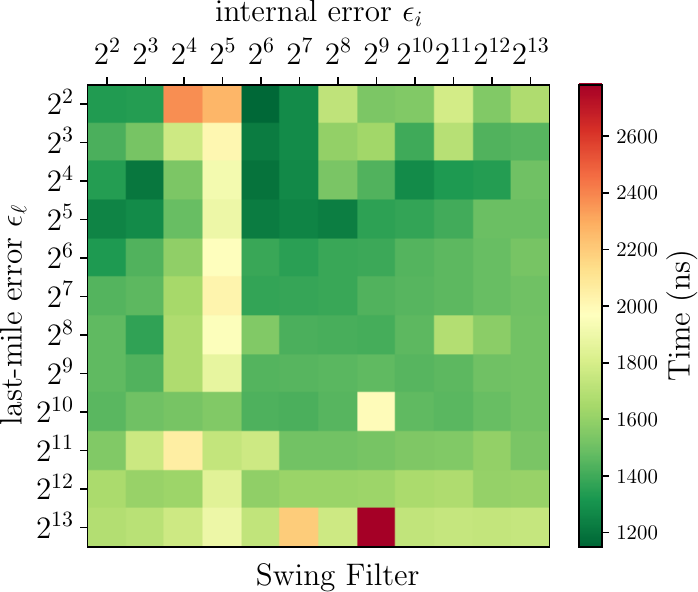}
        \end{subfigure}
        \caption*{\textsf{fb}: Query Time vs. $\epsilon_\ell$ \& $\epsilon_i$} 
    \end{subfigure}

    \begin{subfigure}[t]{\textwidth}
        \centering
        \hfill
        \begin{subfigure}[t]{0.32\textwidth}
            \includegraphics[width=\linewidth]{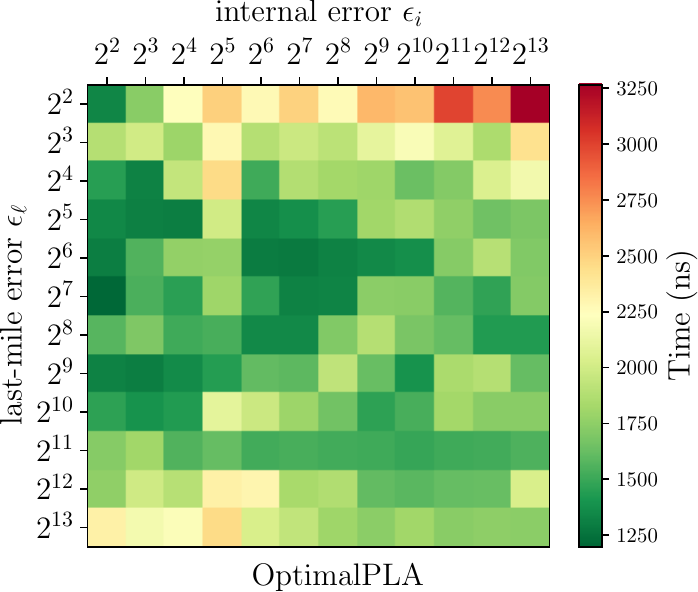}
        \end{subfigure}
        \hfill
        \begin{subfigure}[t]{0.32\textwidth}
            \includegraphics[width=\linewidth]{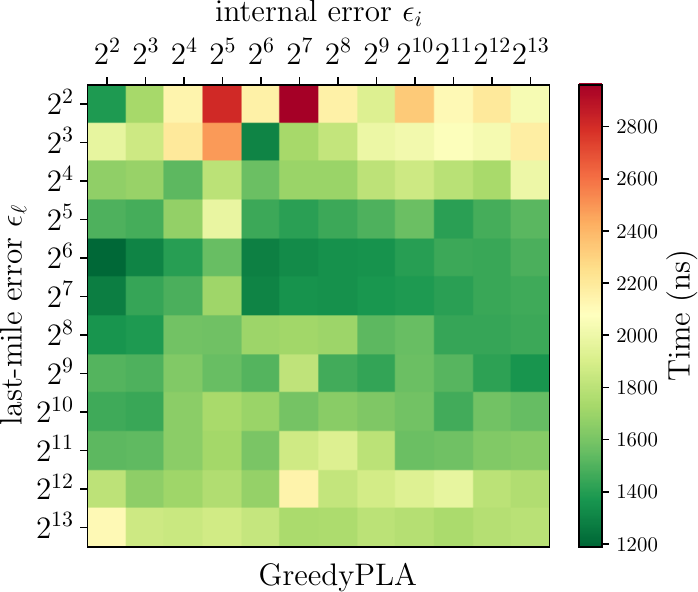}
        \end{subfigure}
        \hfill
        \begin{subfigure}[t]{0.32\textwidth}
            \includegraphics[width=\linewidth]{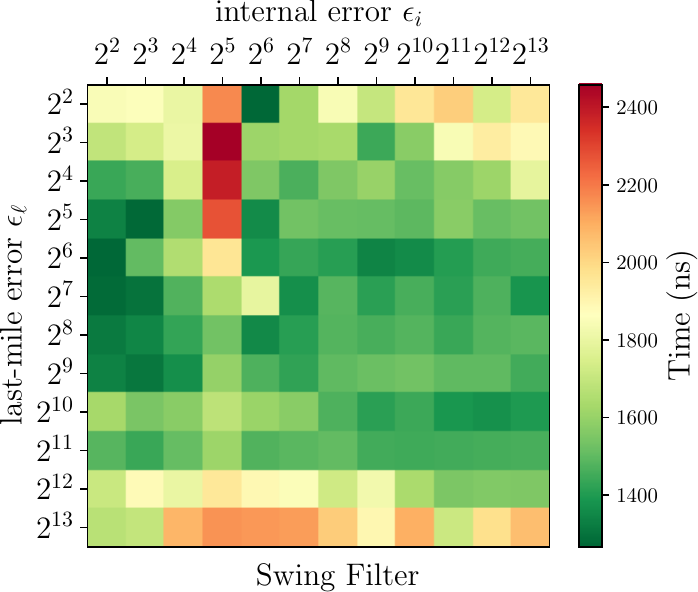}
        \end{subfigure}
        \caption*{\textsf{books}: Query Time vs. $\epsilon_\ell$ \& $\epsilon_i$} 
    \end{subfigure}

    \begin{subfigure}[t]{\textwidth}
        \centering
        \hfill
        \begin{subfigure}[t]{0.32\textwidth}
            \includegraphics[width=\linewidth]{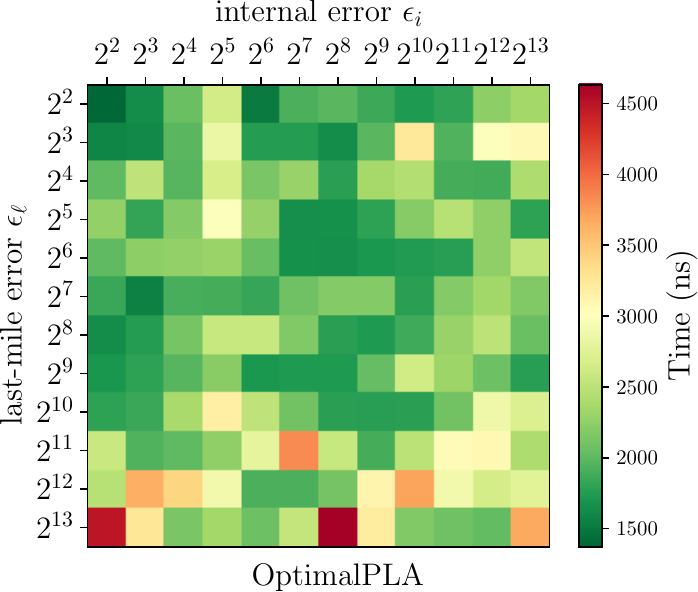}
        \end{subfigure}
        \hfill
        \begin{subfigure}[t]{0.32\textwidth}
            \includegraphics[width=\linewidth]{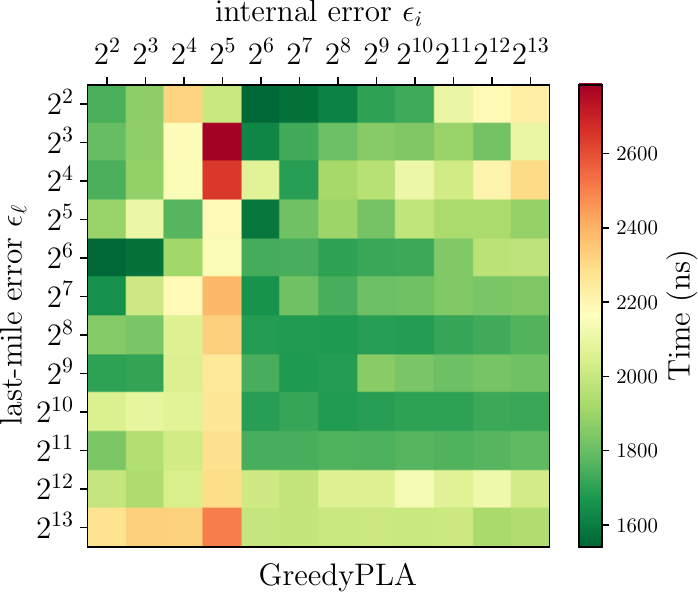}
        \end{subfigure}
        \hfill
        \begin{subfigure}[t]{0.32\textwidth}
            \includegraphics[width=\linewidth]{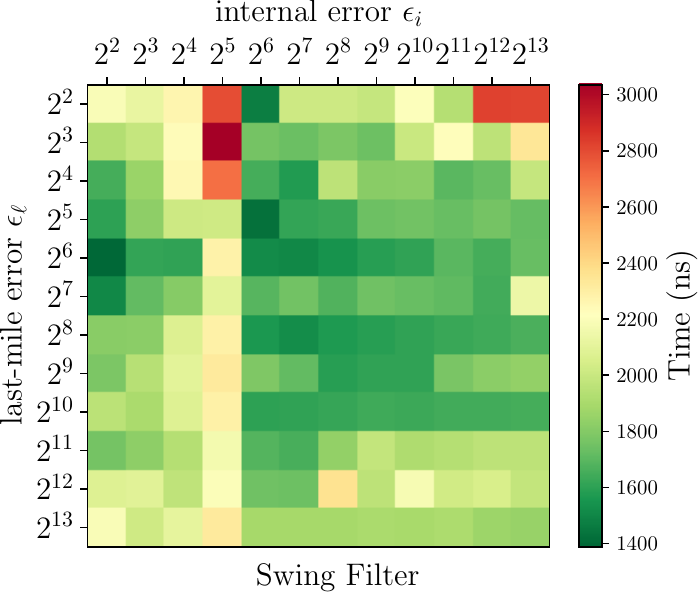}
        \end{subfigure}
        \caption*{\textsf{osm}: Query Time vs. $\epsilon_\ell$ \& $\epsilon_i$} 
    \end{subfigure}
    
    \caption{Query time of PGM-Index with varying $\epsilon_\ell$ and $\epsilon_i$ based on different $\epsilon$-PLA methods on \textsf{fb}, \textsf{books}, \textsf{osm} datasets}
    \label{fig:heatmap1}
    \vspace{-1ex}
\end{figure*}

\begin{figure*}[h]
    \centering

    \begin{subfigure}[t]{\textwidth}
        \centering
        \hfill
        \begin{subfigure}[t]{0.32\textwidth}
            \includegraphics[width=\linewidth]{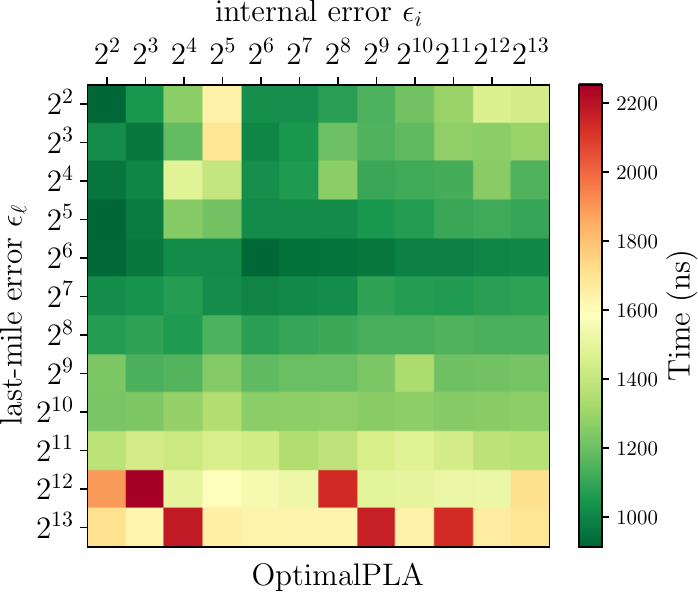}
        \end{subfigure}
        \hfill
        \begin{subfigure}[t]{0.32\textwidth}
            \includegraphics[width=\linewidth]{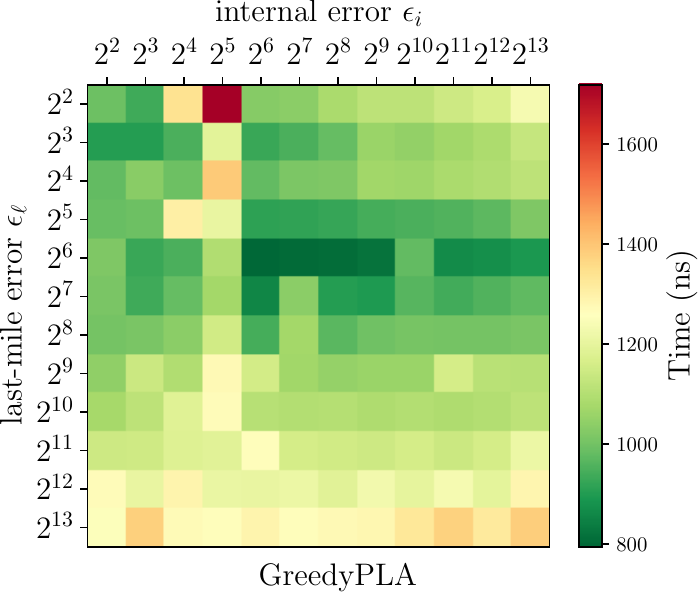}
        \end{subfigure}
        \hfill
        \begin{subfigure}[t]{0.32\textwidth}
            \includegraphics[width=\linewidth]{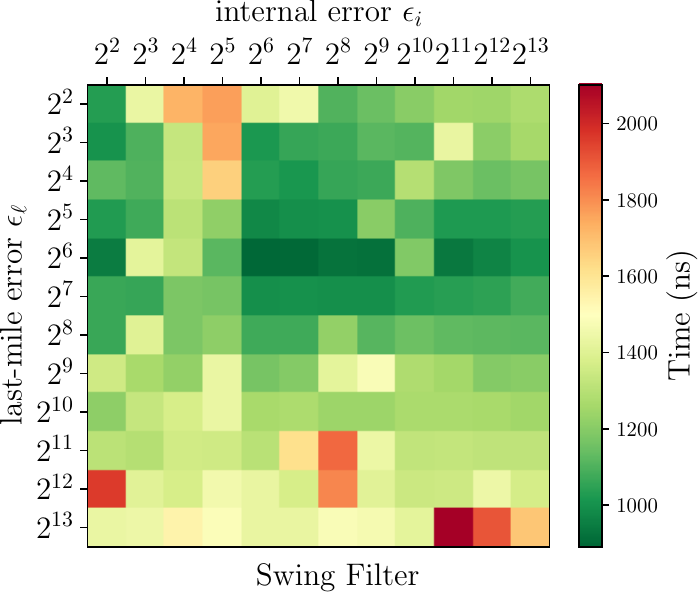}
        \end{subfigure}
        \caption*{\textsf{uniform}: Query Time vs. $\epsilon_\ell$ \& $\epsilon_i$} 
    \end{subfigure}

    \begin{subfigure}[t]{\textwidth}
        \centering
        \hfill
        \begin{subfigure}[t]{0.32\textwidth}
            \includegraphics[width=\linewidth]{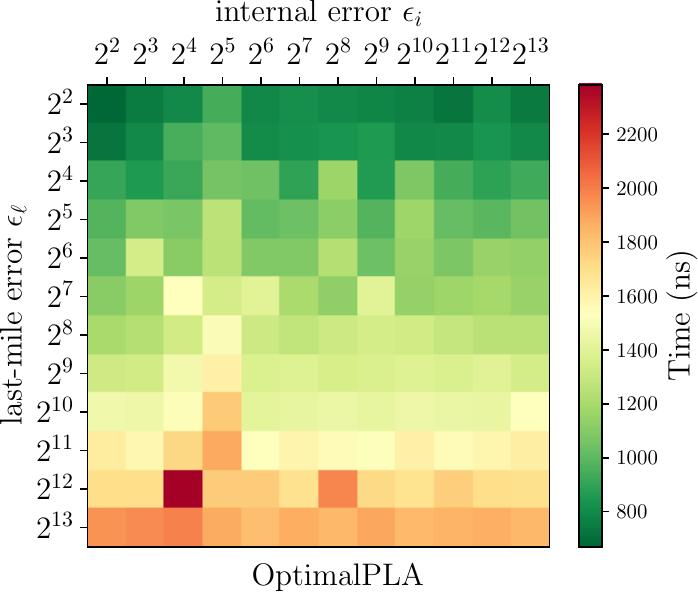}
        \end{subfigure}
        \hfill
        \begin{subfigure}[t]{0.32\textwidth}
            \includegraphics[width=\linewidth]{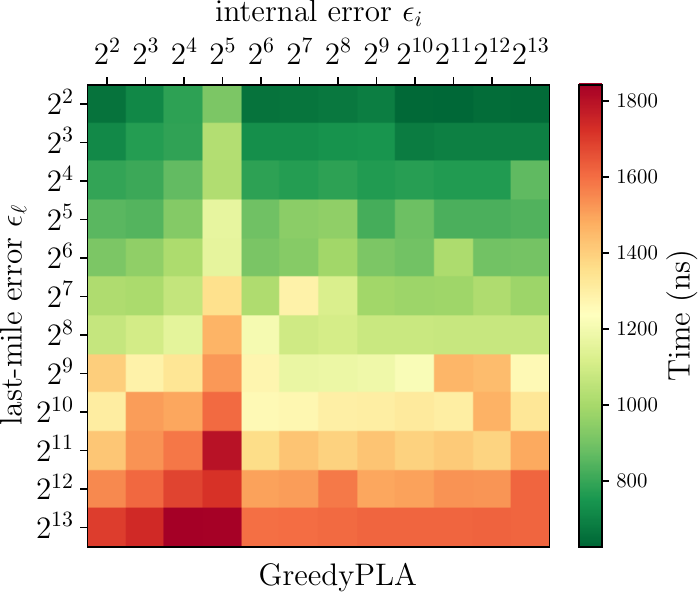}
        \end{subfigure}
        \hfill
        \begin{subfigure}[t]{0.32\textwidth}
            \includegraphics[width=\linewidth]{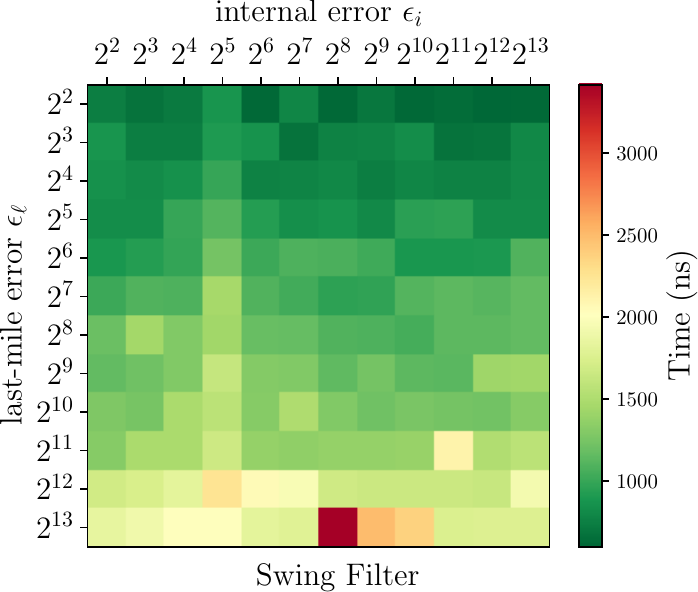}
        \end{subfigure}
        \caption*{\textsf{normal}: Query Time vs. $\epsilon_\ell$ \& $\epsilon_i$} 
    \end{subfigure}

    \begin{subfigure}[t]{\textwidth}
        \centering
        \hfill
        \begin{subfigure}[t]{0.32\textwidth}
            \includegraphics[width=\linewidth]{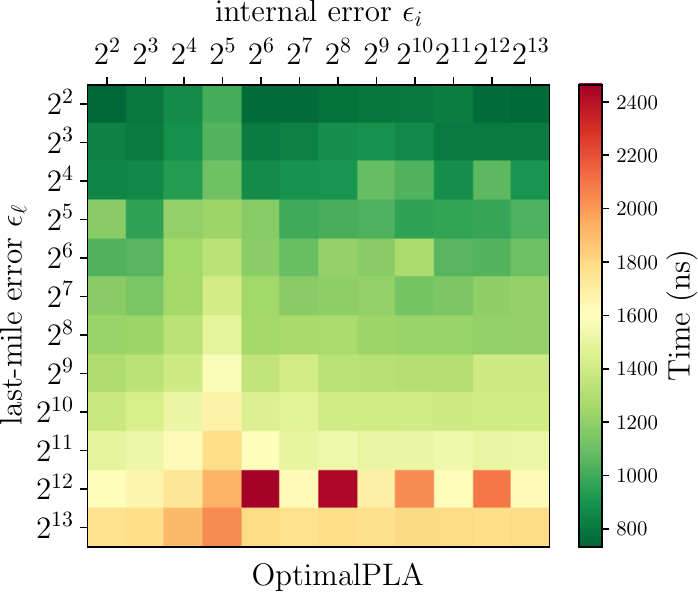}
        \end{subfigure}
        \hfill
        \begin{subfigure}[t]{0.32\textwidth}
            \includegraphics[width=\linewidth]{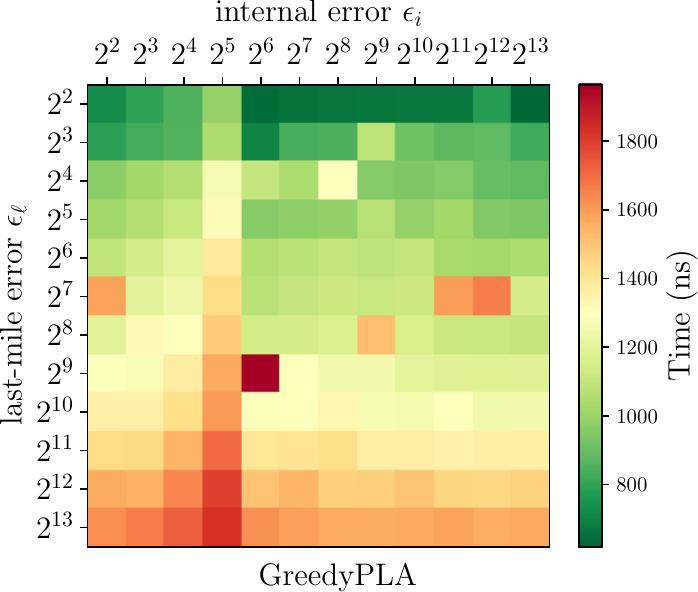}
        \end{subfigure}
        \hfill
        \begin{subfigure}[t]{0.32\textwidth}
            \includegraphics[width=\linewidth]{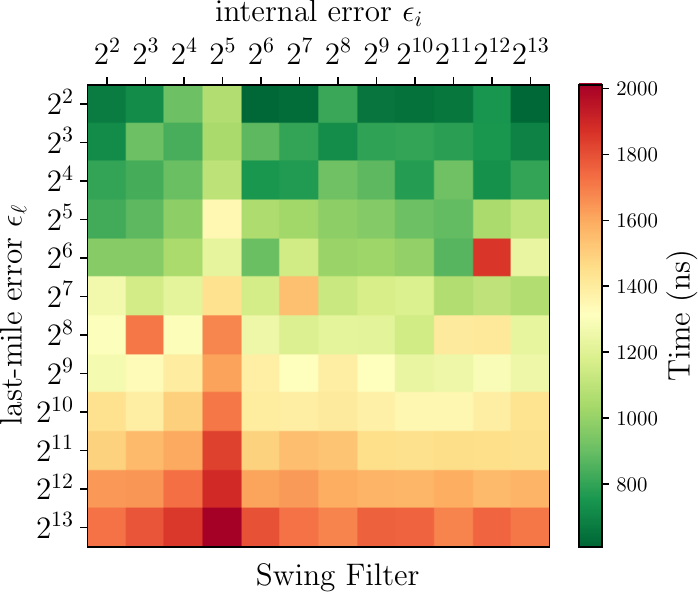}
        \end{subfigure}
        \caption*{\textsf{lognormal}: Query Time vs. $\epsilon_\ell$ \& $\epsilon_i$} 
    \end{subfigure}
    
    \caption{Query time of PGM-Index with varying $\epsilon_\ell$ and $\epsilon_i$ based on different $\epsilon$-PLA methods on synthetic datasets}
    \label{fig:heatmap2}
    \vspace{-1ex}
\end{figure*}

\subsubsection{Evaluation Results of FIT}
We first discuss how $\epsilon$ affects the performance of FITing-Tree. 
As shown in \cref{fig:fit_indexS} and \cref{fig:fit_btime}, when $\epsilon$ increases, both the construction time and index size decrease rapidly. 
This can be attributed to the significant decrease in leaf nodes, indicating that FITing-Tree performs better when applying PLA algorithms that generate fewer segments. 

\subsubsection{Evaluation Results of PGM}
We now focus on the performance of the PGM-Index under varying settings of the internal error $\epsilon_i$ and last-mile error $\epsilon_\ell$.

\cref{fig:pgm_indexS} clearly shows that the index size of PGM-Index is dominated by $\epsilon_\ell$. The memory of the last layer takes up $98\%-99\%$ of the total memory. The $\epsilon_i$ can be negligible when evaluating the index size of PGM-Index.

As for the query time, there are three key factors: \ding{182} the search range, determined by the internal error $\epsilon_i$ and the last-mile error $\epsilon_\ell$, \ding{183} the height of the PGM-Index, and \ding{184} the search method used to locate the target key.

As shown in \cref{fig:heatmap1} and \cref{fig:heatmap2}, when $\epsilon_i$ is fixed and $\epsilon_\ell$ increases, query latency first decreases and then increases, with a turning point at $\epsilon_\ell=2^5$. Before this point, increasing $\epsilon_\ell$ reduces the overall height of the PGM-Index, shortening the search path and reducing query latency. Beyond $\epsilon_\ell = 2^5$, the index height stabilizes, and the expanding search range within segments becomes the dominant factor, resulting in the U-shaped trend observed in \cref{fig:pgm_qtime}.

Notably, when $\epsilon_\ell$ is fixed and $\epsilon_i$ varies, we use linear search if $\epsilon_i\leq2^5$ and binary search otherwise. As shown in \cref{fig:heatmap1} and \cref{fig:heatmap2}, query latency drops noticeably at $\epsilon_i=2^{6}$, highlighting the significant performance impact of switching search strategies.

\subsubsection{Comparison Between FIT and PGM}
We further compare the space-time trade-offs of FIT and PGM in terms of index construction time, memory footprint, and query latency. As shown in \cref{tab:pgm_fit_comparison}, the flattened structure of PGM-Index~\cite{liu2024learned} achieves a smaller index size and faster construction time compared to FITing-Tree. Moreover, the deeper internal B$^+$-tree reduces the query efficiency of FITing-Tree compared to PGM-Index under a small $\epsilon$. As $\epsilon$ increases, the overall height of FITing-Tree diminishes, resulting in comparable querying latency to PGM-Index.
However, PGM-Index exhibits higher sensitivity to the performance of $\epsilon$-PLA methods compared to FITing-Tree. (e.g., when $\epsilon\in\{{2^4,2^5,2^6}\}$, the building time of PGM-Index varies more significantly across different $\epsilon$-PLA algorithms, whereas FITing-Tree shows smaller fluctuations.)

\section{Related Work}\label{sec:related_work}
In this section, we survey and discuss related works from three aspects: 
piecewise linear approximation, learned index, and other learned data sketches.

\textbf{Piecewise Linear Approximation.} 
Efficient PLA models are critical for scientific data processing and time series representation. 
Early PLA techniques primarily focused on minimizing the overall approximation error~\cite{appel1983adaptive,bellman1961approximation,pang2013computing,garofalakis2005wavelet}, but often lacked guarantees on per-point error bounds. 
In contrast, error-bounded PLA algorithms ($\epsilon$-PLA) offer stronger guarantees and typically achieve better space efficiency~\cite{zhao2020mixedsegments}, making them more suitable for applications such as learned indexes. 
Based on their design strategies, existing $\epsilon$-PLA algorithms can be broadly categorized as follows: 

\noindent
\ding{182}~\textit{Optimal PLA.} 
Several works~\cite{bellman1961approximation,terzi2006efficientalg,wu2021optimalsegmented} achieve optimal segmentations for specific error metrics via dynamic programming, which takes super-linear time complexity. 
To address scalability, ParaOptimal~\cite{ORourke81} and its variants like SlideFilter~\cite{elmeleegy2009online} and OptimalPLA~\cite{XiePZZD14} provide error-bounded guarantees while maintaining linear time complexity. 

\noindent
\ding{183}~\textit{Suboptimal PLA.}
Heuristic algorithms like \cite{Keogh2003Segsurvey,liu2008fsw,hu2019csmrtp} adopt either greedy or rule-based strategies to segment time series data. 
More recent greedy algorithms like SwingFilter~\cite{elmeleegy2009online}, GreedyPLA~\cite{XiePZZD14}, and SwingRR~\cite{Lin2020SwingRR} aim to minimize the number of segments while satisfying the error constraint. 

Despite the distinction between optimal and suboptimal strategies, existing work does not clearly establish which approach is superior in the context of learned index design. 
Our work addresses this gap by systematically evaluating their trade-offs in terms of space and construction/query efficiency when integrated into learned index structures.

\textbf{Learned Index Structures.}
Learned indexes model the cumulative distribution function (CDF) over sorted keys using machine learning techniques, offering better space-time trade-offs~\cite{liu2024learned} than traditional tree-based indexes such as B$^+$-Tree~\cite{Comer1979B+-Tree}, FAST~\cite{Kim2010Fast}, and Wormhole~\cite{Wu2019Wormhole}. 
RMI~\cite{kraska2018case} pioneers this paradigm with a static recursive model hierarchy. 
ALEX~\cite{ding2020alex} improves dynamic adaptability via a gap-aware array, while XIndex~\cite{Tang2020XIndex} supports concurrent writes through two-phase compaction. 
FITing-Tree~\cite{galakatos2019fiting} replaces B$^+$-Tree leaves with linear segments, and FINEdex~\cite{Li2021FINEdex} adopts a flattened layout to enhance dynamic performance. 
PGM-Index~\cite{ferragina2020pgm} further improves space efficiency by recursively applying $\epsilon$-PLA to build a fully flattened index. 
Recent advances also target hardware-specific optimizations, including GPU-friendly GIndex~\cite{Liu2024GIndex}, cache-aware FINEdex-Cache~\cite{Zhang2022Cache}, and disk-optimized structures~\cite{Zhang2024Disk}, expanding the applicability of learned indexes in diverse system settings. 

In this work, we focus primarily on PGM-Index and FITing-Tree, as they represent two distinct paradigms on learned index design: fully model-based indexing and model-augmented conventional index structures, respectively. 

\textbf{Other Learned Data Structures.} 
Beyond data indexing, machine learning techniques have also been applied to enhance a variety of data structures. 
Following the taxonomy proposed in~\cite{Ferragina2020LearnedDataStruct}, we categorize learned data structures into three groups based on their application scenarios: 

\noindent
\ding{182} \textit{Data-aware Hashing:} Representative methods include Semantic Hashing~\cite{Ruslan2009Semantichashing} and Spectral Hashing~\cite{Weiss2008Spectralhashing}, which pioneered the use of learned projection vectors to replace random projections in hashing. 

\noindent
\ding{183} \textit{Approximate Membership:} 
Typical examples are the Learned Bloom Filter~\cite{kraska2018case} and its extensions~\cite{Mitzenmacher2018SandwichingBloom, Rae2019Meta, Liu2020Stable}, which learn mapping functions for approximate set membership. 

\noindent
\ding{184} \textit{Frequency Estimation:} 
These methods estimate key frequencies without storing all keys explicitly. 
A representative approach is the Learned Count-Min Sketch~\cite{Hsu2019LearningFrequency}, which pre-trains a model to identify high-frequency items and assigns them exact counters.

Given that $\epsilon$-PLA serves as a simple yet powerful fitting model across a wide range of applications, it is promising to explore its integration with the aforementioned learned data structures, potentially enabling new theoretical insights or empirical performance improvements.
And we leave these for our future work.

\section{Conclusion}\label{sec:conclusion}
This paper revisits error-bounded Piecewise Linear Approximation ($\epsilon$-PLA), a fundamental component in modern learned index structures. 
We contribute a theoretical lower bound of $\Omega(\kappa \cdot \epsilon^2)$ on segment coverage, offering a clearer understanding of the effectiveness of existing fitting algorithms. 
Complementing this analysis, we conduct a comprehensive empirical study of state-of-the-art $\epsilon$-PLA algorithms in two representative learned index frameworks: PGM-Index and FITing-Tree. 
Our findings uncover nuanced trade-offs between error bound, index size, construction cost, and query performance. 
Notably, while optimal algorithms minimize space usage, greedy algorithms may outperform in query and construction time under suitable error bounds. 
These results motivate more principled algorithm selection and parameter tuning when building learned index structures. 
Beyond indexing, we believe $\epsilon$-PLA holds promise for broader use in other learned data structures, which are left for future exploration.



\bibliographystyle{IEEEtran}
\bibliography{ref}

\end{document}